\documentclass[11pt, a4paper]{scrartcl}
\usepackage[a4paper, marginratio={4:6,5:7}]{geometry}
\usepackage[T1]{fontenc}
\usepackage[stable]{footmisc}
\usepackage{xcolor}
\usepackage[british]{babel}
\usepackage{stmaryrd,enumerate,bbm}
\usepackage{amsthm,mathtools}
\usepackage[numbers]{natbib} 
\usepackage{booktabs}

\usepackage[bitstream-charter, cal=cmcal]{mathdesign}
\setkomafont{section}{\normalfont\Large\bfseries}
\setkomafont{subsection}{\normalfont\large\bfseries}
\setkomafont{title}{\normalfont\large\bfseries}
\setkomafont{paragraph}{\normalfont\bfseries}
\setkomafont{subparagraph}{\normalfont\bfseries}
\setkomafont{subsubsection}{\normalfont\bfseries}

\numberwithin{equation}{section}
\theoremstyle{plain}
\newtheorem{foo}{Foo}[section]
\newtheorem{theorem}[foo]{Theorem}
\newtheorem{lemma}[foo]{Lemma}
\newtheorem{corollary}[foo]{Corollary}
\newtheorem{proposition}[foo]{Proposition}
\newtheorem{assumption}[foo]{Assumption}
\theoremstyle{remark}
\newtheorem{remark}[foo]{Remark}
\theoremstyle{definition}
\newtheorem{definition}[foo]{Definition}

\newtheorem*{example*}{Example}

\usepackage{hyperref}
\definecolor{linkcolor}{RGB}{255, 153, 51}
\definecolor{myorange}{RGB}{219, 111, 2}
\definecolor{myblue}{RGB}{0, 89, 179}
\definecolor{black}{RGB}{0, 0, 0}
\definecolor{darkblue}{RGB}{13, 23, 80}
\hypersetup{
	colorlinks,
	citecolor=darkblue,
	filecolor=black,
	linkcolor=darkblue,
	urlcolor=black
}
\addto\extrasenglish{

}
\mathtoolsset{showonlyrefs=true}

\DeclarePairedDelimiter\ceil{\lceil}{\rceil}

\newcommand{\norm}[1]{\left\lVert #1 \right\rVert}
\newcommand{\abs}[1]{\left| #1\right|}

\newcommand{\tnorm}[1]{\interleave #1 \interleave}
\newcommand{\altnorm}[1]{{\left\vert\kern-0.25ex\left\vert\kern-0.25ex\left\vert #1
			\right\vert\kern-0.25ex\right\vert\kern-0.25ex\right\vert}}

\newcommand{\TeXmacs}{T\kern-.1667em\lower.5ex\hbox{E}\kern-.125emX\kern-.1em\lower.5ex\hbox{\textsc{m\kern-.05ema\kern-.125emc\kern-.05ems}}}
\newcommand{\assign}{:=}

\newcommand{\mathd}{\mathrm{d}}
\newcommand{\mathe}{\mathrm{e}}

\newcommand{\tmop}[1]{\ensuremath{\operatorname{#1}}}

\newcommand{\subjclass}[2][]{\medskip\noindent\textbf{MSC (#1).} #2\par}
\newcommand{\keywords}[1]{\medskip\noindent\textbf{Keywords.} #1\par}

\newenvironment{enumerateroman}{\begin{enumerate}[i.] }{\end{enumerate}}

\title{An FBSDE Construction of the Sine-Gordon EQFT for \(\beta^{2} < \frac{6}{7}\, 8\pi\) and Perturbative Renormalization in the Full Subcritical Regime}
\author{
	Sarah-Jean Meyer \thanks{Mathematical Institute, University of Oxford,\\\noindent
		\texttt{sarah-jean.meyer@maths.ox.ac.uk}}}
\date{}

\begin{document}
\maketitle
\begin{abstract}
	We construct the sine-Gordon measure and its renormalized effective potential in finite volume up to the seventh threshold.
	The construction relies on a weak stochastic control problem, its associated weak FBSDE, and an analysis of the renormalization flow equation using a multipole Mayer expansion.
	The analysis is fully inductive, and we include a complete order-by-order analysis in the subcritical regime \(\beta^{2}<8\pi\).
\end{abstract}
\subjclass[2020]{Primary 81S20; Secondary 60H30}
\keywords{stochastic quantization;
	forward-backward stochastic differential equations;
	Euclidean quantum field theory;
	sine-Gordon model; Mayer series}
{\small
	\setcounter{tocdepth}{3}
	\setcounter{secnumdepth}{3}
	\tableofcontents
}

\newpage

\section{Introduction}
In this paper, we study the two-dimensional sine-Gordon Euclidean quantum field theory (EQFT) in finite volume throughout the full subcritical regime \(\beta^2<8\pi\).
We prove two main results.
First, for every \(\beta^2<8\pi\), we develop a perturbative order-by-order analysis of the effective potential based on a modified
multipole version of the classical Mayer expansion.
Second, for \(\beta^2<\frac67\cdot 8\pi\), we construct the finite-volume sine-Gordon measure based on scale-dependent dynamics given by a weak FBSDE.

The first result provides a systematic order-by-order construction of the renormalization flow in a non-polynomial model with infinitely many thresholds in the full subcritical regime.
The second leverages the bounds obtained in the first to construct the measure in the sublinear regime of these estimates for all values of the coupling constant in a finite volume.
The obstruction to extending the second argument to the full subcritical regime is the so-called large-field problem.
The precise meaning of these two statements will be explained below, after we introduce the model.

The finite-volume sine-Gordon model is an example of a non-polynomial, non-Gaussian scalar Euclidean quantum field theory in two dimensions.
Formally, it is given by the Gibbs measure
\begin{align}\label{eq:int-Gibb-ms}
	\text{``}\nu_{\tmop{sG}}^\rho( \mathd \varphi) = \Xi^{-1} \exp \left( -V_{\tmop{sG}}^\rho(\varphi)\right) \mu(\mathd \varphi)\text{''}, \qquad \varphi\in \mathcal{S}'(\mathbb{R}^{2}),
\end{align}
where \(\mu\) is a massive Gaussian free field on the space of tempered distributions \(\mathcal{S}'(\mathbb{R}^{2})\), \(\Xi\) is a normalization constant chosen so that \(\nu^{\rho}_{\tmop{sG}}(\mathcal{S}'(\mathbb{R}^{2})) = 1\), and \(V_{\tmop{sG}}^\rho\) is the cosine interaction, formally written as
\begin{align}\label{eq:int-V-cos}
	\text{``}V_{\tmop{sG}}^\rho(\varphi) \assign \int_{\mathbb{R}^{2}} (\infty\cdot\cos(\beta \varphi(x))-\infty) \rho(x) \mathd x\text{''},
\end{align}
where \(\rho\in C^\infty_c(\mathbb{R}^{2})\) is a smooth volume cut-off.
The symbols \(\infty\) indicate the renormalizations required to make sense of the interaction.

As \(\beta^{2}\) varies in the subcritical range \([0,8\pi)\), the model exhibits an infinite sequence of thresholds
\begin{align}\label{eq:int-beta-thresholds}
	\beta^{2}_{\ell} = \frac{\ell-1}{\ell} \cdot 8\pi, \qquad \ell \in \mathbb{N},
\end{align}
at which additional renormalization is required.
Together with the non-polynomial nature of the interaction, this makes sine-Gordon a particularly rich test case for renormalization methods in Euclidean field theory.
Further background and references on the sine-Gordon model are collected in Section \ref{sec:references}.
\subsection{Outline and main results}
Our approach is based on the continuous renormalization group.
Fix a smooth interpolation \((G_t)_{t\in[0,\infty]}\) between \(G_0=0\) and the covariance \(G_\infty=(m^2-\Delta)^{-1}\)
of the massive free field.
This interpolation defines a scale via the parameter \(t\in [0,\infty]\) and the Gaussian free field can be realized as the terminal value of the Brownian martingale \((W_t)_{t\in[0,\infty]}\) with quadratic variation \(\langle W\rangle_t=G_t\) so that the interacting measure \eqref{eq:int-Gibb-ms} can be formally computed as
\begin{align}\label{eq:sG-from-Gibbs}
	\text{``}\nu_{\tmop{sG}}^{\rho}(\mathcal{O})
	= \frac{\mathbb{E}\left[\mathcal{O}(W_\infty) \mathe^{-V_{\tmop{sG}}^{\rho} (W_{\infty})}\right]}{\mathbb{E}\left[\mathe^{-V_{\tmop{sG}}^{\rho} (W_\infty)} \right]}.\text{''}
\end{align}
To begin to make sense of \eqref{eq:sG-from-Gibbs} and \eqref{eq:int-Gibb-ms}, we introduce a regularization.
Mainly for convenience, we will use the intrinsic regularization coming from the scale interpolation \(G_t\).
More precisely, we choose \(G_t\) such that \(W_t\in C^\infty(\mathbb{R}^{2})\) almost surely whenever \(t<\infty\) and, for constants \(\lambda_T\) and \(\tilde c^{\rho,T}\) to be determined, define the approximate microscopic potential
\begin{align}\label{eq:Vrho-T}
	V_{\tmop{sG}}^{\rho, T}(\varphi) = \int_{\mathbb{R}^{2}} (\lambda_{T} \cos(\beta \varphi ) - \tilde{c}^{\rho, T}) \rho(x) \mathd x,
\end{align}
and the regularized measure
\begin{align}\label{eq:nu-rho-T}
	\nu_{\tmop{sG}}^{\rho, T} (\mathcal{O})
	\assign
	\frac{\mathbb{E} \left[
			\mathcal{O}(W_T)\mathe^{-V^{\rho, T}_{\tmop{sG}}(W_{T})}
			\right]}{\mathbb{E} \left[
			\mathe^{-V_{\tmop{sG}}^{\rho, T} (W_{T})}
			\right]}.
\end{align}

Using Itô's formula, we can write for any scale interpolation \((V^{\rho, T}_{t})_{t\in [0,T]}\) with \(V^{\rho,T}_{T}= V_{\tmop{sG}}^{\rho,T}\),

\begin{align}
	V_{\tmop{sG}}^{\rho,T}(W_{T})
	= & V^{\rho,T}_{T} (W_{T})                                                                                             \\
	= & V^{\rho,T}_{t} (W_{t})
	+ \int_t^{T} \left(
	\partial_{s} V^{\rho,T}_{s}(W_{s}) + \frac{1}{2} \Delta_{\dot{G}_{s}} V^{\rho,T}_{s}(W_{s})
	\right) \mathd s
	+ \int_{t}^T \nabla V^{\rho,T}_{s}(W_{s}) \mathd W_{s}                                                                 \\
	= & V^{\rho,T}_{t}(W_{t}) + \int_t^T \mathcal{H}^{\rho, T}_{s}(W_{s })\mathd s - \log \mathcal{E}_{t,T}^{V^{\rho, T}}.
\end{align}
Here, we define \(\dot{G}_t = \partial_t G_t\) and
\begin{itemize}
	\item \(\nabla V\) denotes the functional gradient defined via the directional derivatives for \(h\in C^\infty_{\tmop{c}}(\mathbb{R}^{2})\) as
	      \begin{align}
		      \langle \nabla V(\varphi), h \rangle = \tmop{D} V(\varphi)[h],
	      \end{align}
	\item \(\Delta_{\dot{G}_{t}}V = \tmop{Tr}(\dot{G}_{t} \tmop{D}^{2}V_{t})\) denotes the functional Laplacian associated to the positive quadratic form \(\dot{G}\),
	\item we use the following shorthand for the remainder for \(s\in [0,T]\),
	      \begin{align}\label{eq:int-def-Hcal}
		      \mathcal{H}^{\rho, T} _{s} (\varphi)
		      \assign \partial_{s} V^{\rho,T}_{s}(\varphi) + \frac{1}{2} \Delta_{\dot{G}_{s}} V^{\rho,T}_{s}(\varphi) - \frac{1}{2} \langle \nabla V^{\rho, T}_s(\varphi), \dot{G}_s \nabla V^{\rho, T}_s(\varphi) \rangle_{L^{2}},
	      \end{align}
	\item we define the stochastic exponential
	      \begin{align}\label{eq:def-stoch-exp}
		      \mathcal{E}_{t,s}^{V^{\rho,T}}
		      \assign & \exp\left(
		      \int_t^s -\nabla V^{\rho, T}_u (W_u) \mathd W_u
		      - \frac{1}{2} \int_t^s \langle \nabla V^{\rho, T}_u(W_u), \dot{G}_u \nabla V^{\rho, T}_u(W_u) \rangle_{L^{2}} \mathd u
		      \right).
	      \end{align}
\end{itemize}
Whenever \(\mathcal{E}^{V^{\rho,T}}\) is a uniformly integrable martingale and the SDE
\begin{align}\label{eq:int-truncated-SDE}
	X_{t}^{\rho, T}  = \int_0^{t\wedge T} \dot{G}_{s} (-\nabla V^{\rho, T}_s(X^{\rho, T}_s)) \mathd s + W_t,\qquad t\geqslant0,
\end{align}
has a unique strong solution, we use Girsanov's theorem to write \eqref{eq:sG-from-Gibbs} in terms of \(X\) as
\begin{align}\label{eq:sG-via-X}
	\nu_{\tmop{sG}}^{\rho, T}(\mathcal{O})
	= \mathbb{E}[\mathcal{O}(X_{T}^{\rho,T}) M_{0, T}^{\rho,T}],
\end{align}
where
\begin{align}\label{eq:def-M-density}
	M_{0, T}^{\rho,T} \assign \frac{\mathe^{-\int_0^T \mathcal{H}^{\rho,T}_s(X_s^{\rho,T})\mathd s}}{\mathbb{E}\left[ \mathe^{-\int_0^T \mathcal{H}^{\rho,T}_{s} (X_s^{\rho,T})\mathd s} \right]}.
\end{align}
If \(\beta^{2}\geqslant4\pi\), the measure \(\nu_{\tmop{sG}}^{\rho}\) and the Gaussian free field \(\mu\) are mutually singular (see \cite[Theorem 1.4]{gubinelliFBDSEApproachSine2026}).
By choosing an appropriate interpolation of \(V^{\rho, T}_{\tmop{sG}}\),
we may, however, be able to prove convergence of \eqref{eq:sG-via-X} as \(T \to \infty\).
In other words, \(\nu_{\tmop{sG}}^{\rho} \ll \tmop{Law}(X^{\rho})\), where \(X^{\rho} = \lim_{T \to \infty } X^{\rho, T} \) in an appropriate sense.

There are two competing mechanisms that decide whether \((V^{\rho, T}_{t})_{t\in [0,T]} \) is a good scale interpolation.
First, the scale interpolation \((V^{\rho,T}_{t})_{t\in[0,T]}\) of the potential needs to be sufficiently nice to ensure strong existence and uniqueness globally in \(t\in [0,T]\) for the SDE \eqref{eq:int-truncated-SDE}.
Second, \((V^{\rho,T}_{t})_{t\in [0,T]}\) must capture all of the singular part of \(\nu_{\tmop{sG}}^{\rho}\) as \(T \to \infty\) to ensure that \(\nu_{\tmop{sG}}^{\rho}\ll  \tmop{Law}(X_{\infty}^{\rho, \infty})\).

The first requires good a priori estimates on \eqref{eq:int-truncated-SDE}.
In the absence of strongly coercive terms, this essentially allows only sublinear growth in \( \varphi\) to appear in the drift \(-\dot{G}_s\nabla V^{\rho, T} _s(\varphi)\).
The second, on the other hand, requires the remainder \(\mathcal{H}^{\rho, T}_{s}\) to be small as \(s \to \infty\).
In other words, \((V^{\rho,T}_{s})_{s\in [0,T]}\) has to be a good approximation to the renormalization flow equation, namely the PDE
\begin{align} \label{eq:int-RG-flow}
	\begin{cases}
		\mathcal{H}^{\rho, T}_t = 0, \\
		V_{T}^{\rho, T} = V_{\tmop{sG}}^{\rho, T}.
	\end{cases}
\end{align}
For the sine-Gordon interaction, the required renormalizations obtained from a naive Picard iteration lead to more and more non-linear dependencies on the field \(\varphi\) as we increase the approximation order.

In this paper, we give a systematic construction of approximate solutions \((V^{\rho,T}_{t})_{t\in [0,T]}\) to \(\mathcal{H}^{\rho, T} \approx 0 \) with control uniform in \(T>0\) in the full subcritical regime \(\beta^{2} < 8\pi\).
For \(\beta^{2}< \beta_{7}^{2}=\frac{6}{7} 8\pi\), we show that the constructed approximation \((V^{\rho,T}_{t})_{t\geqslant0}\) also allows us to control both the SDE \eqref{eq:int-truncated-SDE} and the remaining density \(M_{0,T}^{\rho,T}\) uniformly in the regularization parameter \(T>0\) and to prove convergence of \eqref{eq:sG-via-X}.

We give a simplified version of the approximation result for the renormalization flow equation below. The full statement can be found in Proposition \ref{prop:Va-general-estimates}.
\begin{theorem}\label{thm:int-Va-estimates}
	Let \(\chi(x) = \langle x \rangle ^{-3}\). Let \(\beta^{2}<8\pi\), \(\lambda\in \mathbb{R}\), and \(\rho\in C^\infty_c(\mathbb{R}^{2})\).
	There exists a family of scale interpolations \((V_{t}^{\rho,T})_{t\in [0,T]}\) of \(V^{\rho,T}_{\tmop{sG}}\)
	such that the associated force \(F_t^{\rho,T}=-\nabla V_t^{\rho,T}\) and the error term \(\mathcal H_t^{\rho,T}\) satisfy the following properties.
	There are \(L(\beta^{2}),\kappa(\beta^{2})>0\) and \(K_V(\beta^{2}),K_F(\beta^{2})\geqslant0\) such that the following bounds hold.

	For \(\varphi\in \mathcal{S}'(\mathbb{R}^{2})\) such that \(\langle t \rangle^{-1/2}\|\nabla\varphi\|_{L^\infty(\chi)}\leqslant D\),
	\begin{align}
		\sup_{T\geqslant 0}|V_t^{\rho,T}(\varphi)|
		 & \lesssim_{\rho} \abs{\lambda}(1+\abs{\lambda})^{L(\beta^{2})-2}\, \mathe^{\frac{\beta^2}{2}G_t(0)} (1+D)^{K_V(\beta^{2})},                    \\
		\sup_{T\geqslant 0}\|\dot G_t F_t^{\rho,T}(\varphi)\|_{L^\infty(\chi)}
		 & \lesssim_{\rho} \abs{\lambda}(1+\abs{\lambda})^{L(\beta^{2})-2}\, \mathe^{\frac{\beta^2}{2}G_t(0)} \|\dot G_t\|_{L^1(\chi^{-1})} (1+D)^{K_F},
	\end{align}
	and the error terms \(\mathcal{H}_{t}^{\rho, T}\) and \(H^{\rho, T}_t \assign -\nabla \mathcal{H}^{\rho, T}_{t}\) satisfy
	\begin{align}
		\sup_{T\geqslant0}\left(
		|\mathcal H_t^{\rho,T}(\varphi)| + \|H_t^{\rho,T}(\varphi)\|_{L^\infty(\chi)}
		\right)
		\lesssim_{\rho} \abs{\lambda}^{L(\beta^{2})}(1+\abs{\lambda})^{L(\beta^{2})-2} \langle t \rangle^{-1-\kappa(\beta^{2})} (1+D)^{2K_F(\beta^{2})}.
	\end{align}
	Moreover, \(V_t^{\rho,T}\), \(\dot G_t F_t^{\rho,T}\), and \(\mathcal H_t^{\rho,T}\) are locally Lipschitz in \(\varphi\) (with Lipschitz constants of similar form), uniformly in \(T\), and for each fixed \(t<\infty\) admit a pointwise limit as \(T\to\infty\).
\end{theorem}
The point is that, throughout the full subcritical regime, one can construct approximate effective potentials whose remainder remains integrable and finite as \(T \to \infty \).
Indeed, ignoring the field dependence, the decay of \(\mathcal H_t^{\rho,T}\) in Theorem \ref{thm:int-Va-estimates} implies that \(\nabla \mathcal H\) is integrable in the ultraviolet for all \(\beta^2<8\pi\).
The price to pay for bounds that are uniform in the regularization parameter is a dependence on the size of the fields \(\varphi\).
As \(\beta^2\) approaches \(8\pi\), the exponent \(K_F(\beta^2)\) deteriorates in the sense that \(K_F(\beta^{2}) \to \infty \).
This makes the control of the SDE \eqref{eq:int-truncated-SDE} in the full subcritical regime intractable for us right now.
The first threshold at which the estimates on the force become genuinely superlinear is
\begin{align}
	\beta_7^2=\frac67\cdot 8\pi.
\end{align}
Below this threshold, we are able to control the SDE \eqref{eq:int-truncated-SDE} corresponding to the approximation given by Theorem \ref{thm:int-Va-estimates}.
Let us drop the parameter \(T\) whenever \(T=\infty\).
\begin{theorem}\label{thm:sG-6/7}
	Let \(\beta^2 < \frac{6}{7} 8 \pi\), \(\lambda \in \mathbb{R}\), fix \(\rho \in C^\infty_{c}(\mathbb{R}^{2})\), and let \(T \in [0, \infty]\).
	There exists a family of scale interpolations \(\bigl((V_t^{\rho,T})_{t\in[0,T]\setminus \{\infty\}}\bigr)_{T\in[0,\infty]}\) such that for every \(\varepsilon > 0\),
	\begin{itemize}
		\item \eqref{eq:int-truncated-SDE} has a unique strong solution for every \(T\in [0,\infty]\),
		\item \eqref{eq:sG-via-X} holds for every \(T<\infty\),
		\item for every bounded continuous observable
		      \(\mathcal{O} : H^{- \varepsilon} (\chi) \to\mathbb{R}\), the limit
		      \begin{align}
			      \nu^{\rho}_{\tmop{sG}} (\mathcal{O})
			      \assign \lim_{T \rightarrow \infty}
			      \nu^{\rho, T}_{\tmop{sG}} (\mathcal{O}),
		      \end{align}
		      exists and is given by
		      \begin{align}\label{eq:sG-Girsanov}
			      \nu^{\rho}_{\tmop{sG}} (\mathcal{O})
			      = \mathbb{E} \left[\mathcal{O}(X_{\infty}^{\rho}) M_{0, \infty}^{\rho}\right]
			      = \frac{\mathbb{E} \left[
					      \mathcal{O} (X_{\infty}^{\rho}) \mathe^{- \int_0^{\infty}
						      \mathcal{H}_s^{\rho} (X_s^{\rho}) \mathd s} \right]}{\mathbb{E} \left[
					      \mathe^{- \int_0^{\infty} \mathcal{H}_s^{\rho} (X_s^{\rho}) \mathd s}
					      \right]},
		      \end{align}
		      where \(X^\rho\) is the unique solution to the untruncated SDE
		      \begin{align}\label{eq:Xrho-6/7}
			      X_{t}^{\rho} = \int_0^t \dot{G}_{s} (F_s^{\rho}(X_{s}^{\rho})) \mathd {s} + W_t.
		      \end{align}
	\end{itemize}
\end{theorem}
The density \(M_{0,\infty}^{\rho}\), defined in \eqref{eq:def-M-density} and appearing in \eqref{eq:sG-Girsanov}, can equivalently be described using a weak stochastic control problem.
We will prove a slightly more general statement for the Laplace transform
\begin{align}\label{eq:def-LP}
	\Lambda_{\tmop{sG}}^{\rho, T}(g) \assign \nu_{\tmop{sG}}^{\rho, T}(\mathe^{-g}).
\end{align}
Here, \(g\) will be chosen in a suitable subclass of functionals later.
Let \(\mathcal A\) be the class of predictable controls adapted to the solution \eqref{eq:int-truncated-SDE} and define the stochastic exponential of \(u\) by
\begin{align}\label{eq:def-Ecal-u}
	\mathcal{E}_{t,s}(u)
	:= \exp\left(
	\int_t^s u_r \,\mathd W_r
	- \frac{1}{2} \int_t^s \|Q_r u_r\|_{L^2}^2 \mathd r
	\right),
	\qquad 0\leqslant t<s \leqslant \infty.
\end{align}
Whenever \(\mathcal{E}_{0,t}(u)\) is a uniformly integrable martingale (U.I.), we define the tilted probability measure by \(\mathd\mathbb P^u = \mathcal{E}_{0,T}(u) \mathd \mathbb{P}\), under which \(W_{t}^{u} = W_{t} - \int_0^t \dot{G}_{s} u_s \mathd s\) is a Brownian motion.
Denoting the expectation under \(\mathbb{P}^u\) by \(\mathbb{E}^u\), define the cost functional
\begin{align}
	\mathcal{J}^{\rho, T}(g; u) =
	\begin{cases}
		\mathbb{E}^u \left[
			g(X_{T}^{\rho, T}) + \int_0^T \mathcal{H}^{\rho, T}_s(X_s^{\rho, T} ) \mathd s + \frac{1}{2} \int_0^T \norm{Q_s u_s}^{2}_{L^{2}}\mathd s
		\right],\; & \mathcal{E}(u) \text{ is U.I.,} \\
		\infty, \; & \text{otherwise.}
	\end{cases}
\end{align}
and the value function
\begin{align}\label{eq:int-def-Wcal}
	\mathcal{W}^{\rho,T}(g) \assign \inf_{u\in \mathcal{A}} \mathcal{J}^{\rho,T}(g;u).
\end{align}
We show the following.
\begin{proposition}\label{prop:fin-vol-var}
	Let \(\rho\in C^\infty_c(\mathbb{R}^{2})\).
	For every admissible \(g\) (see Assumption \ref{hyp:g-allowed}), the infimum in the weak control problem \eqref{eq:int-def-Wcal} is attained and
	\begin{align}\label{eq:LP=W}
		-\log(\Lambda_{\tmop{sG}}^{\rho}(g)) = \mathcal{W}^{\rho} (g) - \mathcal{W}^{\rho} (0).
	\end{align}
	Moreover, the optimal control \(\bar u^g\) is unique up to \(\tmop{ker}(Q_s)\), that is, for any optimal \(v\in \mathcal{A}\),  \(Q_s(\bar{u}_s-v_s)= 0\) \(\mathd \mathbb{P} \otimes \mathd s\)-almost surely.
	Under the optimal measure \(\mathbb P^{\bar u^0}\), the terminal law of \(X^{\rho}\) is \(\nu_{\tmop{sG}}^{\rho}\) and in particular
	\begin{align}
		\mathcal{E}_{0,\infty}(\bar{u}^{0}) = M_{0,\infty}^{\rho}.
	\end{align}
\end{proposition}

\begin{remark}
	It is natural to ask about possible extensions of the construction and the perturbation theory.
	Regarding the construction, extensions to infinite volume or to more singular values of \(\beta^{2}\) would be very interesting.
	Infinite volume control of the effective force \(F^{\rho} \) follows easily from the estimates contained here. The control of the density \(M_{0,\infty}^{\rho} \) appears much less straightforward.
	Extending the parameter range for \(\beta^{2}\) would require either a better understanding of the large-field problem or global-in-time existence for the FBSDE \eqref{eq:int-truncated-SDE} (see also Remark \ref{rem:large-fields}),
	or improved estimates on the effective force to replace the superlinear bounds in Theorem \ref{thm:int-Va-estimates}.
	For the perturbation theory, we hope to address the more general (sub)critical \(d\)-dimensional setting with field-dependent renormalizations in the future.
\end{remark}

Combined, the variational problem and Theorem \ref{thm:int-Va-estimates} are sufficient to recover several structural properties of the finite-volume measure.
In particular, one obtains the corresponding Laplace principle, the rotation invariance of the measure (provided that the volume cut-off \(\rho\) is radial), non-Gaussianity, reflection positivity, and mutual singularity with the Gaussian free field all in finite volume
using the same arguments as for the strong formulation in \cite{gubinelliFBDSEApproachSine2026}.
Since no new ideas are required to extend these properties to the weak formulation, we do not include their proofs here.

\paragraph{Organization of the article.}
The paper is organized in two largely independent parts.

\subsubsection*{Part 1 (Section \ref{sec:FBSDE}): Construction of the measure.}
In this first part, we prove Theorem \ref{thm:sG-6/7} and Proposition \ref{prop:fin-vol-var} assuming the existence of a family of effective potentials as in Theorem \ref{thm:int-Va-estimates}.
We start by proving a version of Theorem \ref{thm:sG-6/7} for the regularized measure \(\nu_{\tmop{sG}}^{\rho,T}\) using Girsanov's theorem.
We then show the representation is stable in the limit \(T\to\infty\).

\subsubsection*{Part 2 (Sections \ref{sec:Flow} and \ref{sec:details-flow}): Construction of the approximate effective potential.}
The second part develops a perturbative solution theory for the renormalization flow equation; that is, we construct a family of effective potentials in the full subcritical regime \(\beta^2<8\pi\), order by order, using a modified Mayer expansion and prove Theorem \ref{thm:int-Va-estimates}.
The Mayer expansion for the cosine interaction forms the starting point and we first explain why it must be modified once \(\beta^2\geqslant 4\pi\) using the dipole term as a prototype.
We then generalize these modifications to allow for a fully inductive renormalization procedure.
The main estimates on the effective potential and the force are stated in Section \ref{sec:pot-for-estimates}, and the detailed proofs are given in Section \ref{sec:details-flow}.
Even though we cannot leverage the analysis of the perturbation series in the full subcritical regime, it gives a systematic analysis of the sine-Gordon renormalization flow, which, to the best of our knowledge, has not previously appeared in the literature.

\subsection{Related work}\label{sec:references}
Besides being a prototypical example for a scalar, non-Gaussian EQFT, the sine-Gordon model is of interest for its numerous relations to other relevant models such as the Coulomb gas, see e.g. \cite{brydgesDebyeScreening1980a}, the fermionic Thirring model, see \cite{colemanQuantumSineGordonEquation1975, bauerschmidtColemanCorrespondenceFree2023}, and dimer models \cite{berestyckiMassiveHolomorphicityNearcritical2026}.
In the massless case, similarly to the Liouville theory, the model is also expected to be integrable \cite{lukyanovExactExpectationValues1997,lukyanovFormFactorsSolitoncreating2001}.
Due to its non-polynomial nature, the model poses unique challenges, and today there are many constructions under various combinations of assumptions on the volume and the size of the model parameters \(\beta, \lambda\) and \(m^{2}\).
In the discrete renormalization group approach, we refer to the works of Dimock and Hurd \cite{dimockSineGordonRevisited2000} covering the full subcritical range \(\beta^{2}<8\pi\) and \(m^{2}>0\) in infinite volume for \(\lambda \ll  1\) and the related recent construction by Pelaič \cite{pelaicSubcriticalFinitevolumeMassive2025} in the full subcritical range in finite volume for arbitrary \(\lambda\) addressing an issue in \cite{dimockSineGordonRevisited2000}.
Another line of work by Benfatto, Nicoló and several co-authors \cite{benfattoMassiveSineGordonEquation1982,nicoloMassiveSineGordonEquation1983, nicoloMassiveSineGordonEquation1986}
establishes upper and lower bounds for the partition function in the full subcritical regime in finite volume to prove UV-stability without constructing the model directly.

The sine-Gordon model can be generalized to any dimension \(d\geqslant1\), but is an EQFT only when \(d=2\). In the massless setting \(m=0\) in finite volume, we mention \cite{lacoinProbabilisticApproachUltraviolet2022} in the case \(d=1\), and the recent work \cite{vihkoMasslessSineGordonModel2026} for general \(d\geqslant1\) extending results of \cite{bauerschmidtColemanCorrespondenceFree2023} on the smeared correlation functions to the \(d\)-dimensional setting.

Despite these efforts, the full construction, that is, for all values of \(\lambda\in \mathbb{R}\), \(\beta^{2}\in [0,8\pi)\) in the infinite volume \(\mathbb{R}^{2}\), is still open for both the massive and massless cases.

In contrast to the SPDE constructions for the \(\Phi^4_{2,3}\) model, which saw early success thanks to the strong coercivity (see e.g. \cite{mourratDynamicPhi32017,gubinelliPDEConstructionEuclidean2021}),
the Langevin dynamics for the sine-Gordon model have only
been partially resolved very recently.
Local well-posedness has been known for some time in the language of regularity structures \cite{chandraDynamicalSineGordonModel2018}, and the first extensions of the global solution theory beyond \(4\pi\)
only appeared in the last few years in \cite{chandraPrioriBounds2d2024,shenGlobalWellposedness2D2025a} for \(\beta^{2}\) slightly above \(4\pi\) and \cite{bringmannGlobalWellposednessDynamical2026} for \(\beta^{2}<6\pi\).
We also mention some recent progress for the hyperbolic model \cite{gubinelliSimpleConstructionSineGordon2025,zineHyperbolicSineGordonModel2025}.

Since the interaction is non-polynomial, the perturbative analysis of the renormalization flow equation for the sine-Gordon potential is already non-trivial, and to the best of our knowledge, to date, no systematic order-by-order account has appeared in the literature.
However, thanks to the boundedness of the cosine interaction, there have been several works on the sine-Gordon model based on the continuous renormalization group we use here for portions of the subcritical regime.
We refer to \cite{bauerschmidtStochasticDynamicsPolchinski2024} for a general review on the continuous renormalization group approach.
The pioneering analysis of Brydges and Kennedy \cite{brydgesMayerExpansionsHamiltonJacobi1987} established convergence of the Mayer expansion up to the first non-trivial threshold \(\beta^{2}< \beta^{2}_3 = \frac{16}{3} \pi\).
Recently, this result has been extended in \cite{bauerschmidtLogSobolevInequalityContinuum2021} up to \(\beta^{2}<6\pi\). This result has since been used to derive a number of properties of the massive and massless sine-Gordon measure using the resulting coupling with the Gaussian free field \cite{bauerschmidtMaximumCouplingSineGordon2022,bauerschmidtColemanCorrespondenceFree2023,
	bauerschmidtTwistedDiracOperators2025}.
Similar couplings have also been achieved for the \(\Phi^4_2\)-model, the Liouville and \(\sinh\)-Gordon models in \cite{barashkovMultiscaleCouplingMaximum2023,hofstetterCouplingLiouvilleSinhGordon2025} and in \cite{duchConstructionGrossNeveuModel2024} in the case of critical fermions.

Regarding the FBSDE \eqref{eq:int-exact-SDE} and variational problem,
the works \cite{barashkovStochasticControlApproach2022,gubinelliFBDSEApproachSine2026} are closely related, using different variations of the control problem \eqref{eq:int-def-Wcal} in the regimes \(\beta^{2}<4\pi\) and \(\beta^{2}<6\pi\), respectively.
We give a more detailed discussion of the relation to the results obtained here at the end of this section.
The variational method was first used to study the finite volume \(\Phi^4_{3}\) theory \cite{barashkovVariationalMethodPhi43-2020}, and we refer also to the study of the infinite volume problem of some \(2\)-dimensional EQFTs by the same authors \cite{barashkovVariationalMethodEuclidean2023}.
For subcritical fermions, the FBSDE approach was implemented in \cite{vecchiStochasticAnalysisSubcritical2025}.

\paragraph{Relation to the strong FBSDE formulation.}
A related FBSDE construction was used in \cite{gubinelliFBDSEApproachSine2026} by Gubinelli and the author to construct the sine-Gordon measure up to \(\beta^{2} < \beta_4^{2} = 6\pi\).
The approach we present here is the corresponding weak formulation of this FBSDE in the sense of \cite{antonelliWeakSolutionsForward2003}.
The strong FBSDE used in \cite{gubinelliFBDSEApproachSine2026} corresponds to choosing the drift in \eqref{eq:int-truncated-SDE} as the exact solution to \eqref{eq:int-RG-flow} so that
\begin{align}\label{eq:exact-sG-law}
	\tmop{Law}(X_{T}^{\rho, T}) = \nu_{\tmop{sG}}^{\rho, T}.
\end{align}

This leads to the forward-backward SDE
\begin{align}\label{eq:int-exact-SDE}
	X^{\rho, T} _{t} = \int_{0}^{t} \dot{G}_s \left(\mathbb{E}_s \left[ -\nabla V^{\rho, T} _{\tmop{sG}}(X^{\rho, T} _T)\right]\right) \mathd s + W_t,\qquad t\in [0,T].
\end{align}
Only then do we make use of the scale interpolation \((V_t^{\rho, T})_{t\geqslant0}\) of \(V^{\rho, T}_{\tmop{sG}}\) to rewrite \eqref{eq:int-exact-SDE} to obtain better control as \(T \to \infty\).
Using the same notation as before, the SDE \eqref{eq:int-exact-SDE} is then equivalent to the coupled FBSDE
\begin{align}\label{eq:strong-FBSDE}
	\begin{cases}
		X_{t}^{\rho,T} = \int_{0}^t \dot{G}_s \left(
		F_s^{\rho, T} (X_s^{\rho, T}) + R_{s}^{\rho, T}
		\right) \mathd s + W_t \\
		R_{t}^{\rho, T} = \mathbb{E}_t \int_t^T -\nabla \mathcal{H}_s^{\rho, T} (X_s^{\rho, T}) \mathd s + \mathbb{E}_{t}\int_t^T \tmop{D} F_s^{\rho, T}(X_s^{\rho, T}) \dot{G}_s R_s^{\rho, T} \mathd s.
	\end{cases}
\end{align}
For \(\beta^{2}<6\pi\), the system can be controlled uniformly in \(T>0, \rho\prec 1\) in \cite{gubinelliFBDSEApproachSine2026}, which in turn constructs the sine-Gordon measure using \eqref{eq:exact-sG-law}.
Moreover, the limits of \(X^{\rho,T}, R^{\rho,T}\) as \(T \to \infty\) are shown to be optimal for the strong control problem
\begin{align}
	\mathcal{W}^{\mathtt{s}, \rho, T} (g) \assign \inf_{u\in \mathcal{A}^W} \mathbb{E} \left[
		g(X_{T}^{\rho, T}(u)) + \int_0^T \mathcal{H}_s^{\rho, T}(X_{s}^{\rho,T}(u)) \mathd s + \frac{1}{2}\int_0^T \norm{Q_s u_s}^{2}_{L^{2}}\mathd s
		\right],
\end{align}
where \(\mathcal{A}^{W}\) denotes the class of all predictable processes adapted to the filtration generated by \(W\), and \(X_t^{\rho, T}\) is the controlled version of the diffusion \eqref{eq:int-truncated-SDE},
\begin{align}
	X_{t}^{\rho, T}(u)
	= \int_0^t \dot{G}_s
	\left(F_s^{\rho,T}(X_s^{\rho,T}(u))+u_s\right)\mathd s + W_t.
\end{align}

We now compare this to the weak control problem \eqref{eq:int-def-Wcal}.
Whenever the cost function of \eqref{eq:int-def-Wcal} is finite,
the stochastic exponential associated to the control \(u\) is a uniformly integrable martingale.
Therefore, under \(\mathbb P^u\) the process \(W_{t}^{u} = W_{t} - \int_0^t \dot{G}_{s} u_s \mathd s \) is a Brownian motion with \(\langle W^u \rangle_t =  G_t\).
For the solution \(X^{\rho,T}\) to \eqref{eq:int-truncated-SDE} this means that
\begin{align}\label{eq:X-rho-u}
	X_t^{\rho,T} = \int_0^t \dot{G}_s (F_s^{\rho, T}(X_s^{\rho,T}) + u_s) \mathd s + W_t^u.
\end{align}
In the FBSDE and stochastic control theory literature, for \(\bar{r}\) the optimal control in \eqref{eq:int-def-Wcal}, \((X,\bar{r})\) is a weak solution to \eqref{eq:strong-FBSDE} in the sense of \cite{antonelliWeakSolutionsForward2003}.
In \eqref{eq:strong-FBSDE}, the coupling between \(X\) and \(R\) means that even linear growth in the force \(F_t^{\rho, T}\) spoils any simple a priori estimate.
In contrast, in the weak formulation the equations are decoupled.
We can first solve \eqref{eq:int-truncated-SDE} and then use good pathwise control of \(X^{\rho, T}\) to control the optimal control \(\bar{u}^{g}\).
A drawback is that controlling the infinite-volume limit in this weaker formulation appears less straightforward, and it would be interesting to extend the results obtained here to the infinite volume setting.

As long as \(T<\infty\) and \(\rho\in C^\infty_{c}(\mathbb{R}^{2})\), both the weak and strong formulations are equivalent.
This equivalence can be used to show existence (but not uniqueness) of a solution to \eqref{eq:strong-FBSDE} for any value \(\lambda\in \mathbb{R}\) using the results presented here, at least in a finite volume \(\rho\in C^\infty_{c}(\mathbb{R}^{2})\), for \(\beta^{2}<\frac{6}{7} 8\pi\).
The converse is not true, as \eqref{eq:strong-FBSDE} may have solutions of non-feedback type, that is, solutions for which \(R\) is not measurable with respect to the filtration generated by the forward process \(X\).
Such solutions cannot, however, be directly related to the measure \(\nu_{\tmop{sG}}^{\rho}\); see also \cite[Remark 4.10]{gubinelliFBDSEApproachSine2026}.

\paragraph*{Acknowledgements.}
I thank Peter Paulovics for helpful discussions and comments on an earlier version of this article. I also thank Massimiliano Gubinelli and Jaka Pelaič for helpful discussions.
This research has been supported by the EPSRC CDT in
Mathematics of Random Systems: Analysis, Modelling and Simulation
EP/S023925/1.
For the purpose of open access, the author has applied a \href{https://creativecommons.org/licenses/by/4.0/}{CC--BY} public copyright licence to any author accepted manuscript arising from this submission.

\subsection{Notation and assumptions}
We fix some notation to be used throughout.
\paragraph{Weights and norms}
\begin{itemize}
	\item Let \(\rho\in C^\infty_{c}(\mathbb{R}^{2})\) be a fixed smooth volume cut-off satisfying the pointwise bound \(\rho(x)\leqslant1\). We denote these conditions compactly as \(\rho\prec 1\).
	\item Let \(\chi\in L^1( \mathbb{R}^{2})\) be a polynomial weight decaying sufficiently fast at spatial infinity, e.g. \(\chi(x) = \langle x\rangle^{-3}\).
	      We will often use implicitly that
	      \begin{align}\label{eq:chi-triangle}
		      \chi(x)\chi(y)^{-1} \lesssim \chi(x-y)^{-1}.
	      \end{align}
	\item For any smooth and non-negative function \(w:\mathbb{R}^{2} \to \mathbb{R}\) and \(p,q\in [1,\infty], s\in \mathbb{R}\) we define the multiplicatively weighted Lebesgue
	      and Sobolev spaces \(L^{p}(w)\), \(W^{s, p}(w)\) and \(H^{s}(w)=W^{s, 2}(w)\), as well as the weighted Besov spaces \(B^s_{p,q}(w)\), on \(\mathbb{R}^{2}\). More precisely, the \(L^p(w)\) norm is defined as
	      \begin{align}\label{eq:def-weighted-Lp}
		      \norm{f}^{p}_{L^p(w)} = \norm{w\cdot f}_{L^p}^{p} = \int_{\mathbb{R}^{2}} \abs{w(x) f(x)}^p \mathd x,
	      \end{align}
	      with the usual modifications for \(p=\infty\),
	      see also \cite[Chapter 2]{bahouriFourierAnalysisNonlinear2011} for the concrete definition of \(B_{p,q}^{s}(w)\).
	\item For any multiindex \(\alpha\) and \(\gamma,t>0\), and \(h:[0,\infty) \times \mathbb{R}^{2} \to \mathbb{R}\),
	      we define the seminorms
	      \begin{align}\label{eq:def-fbsde-norms}
		      [h({t})]_{\gamma, \alpha,  t} \assign \langle t \rangle^{-\frac{\abs{\alpha}+2\gamma}{2}} \norm{\nabla^{\alpha} h(t)}_{L^\infty(\chi)},
		      \qquad
		      [h]_{\gamma, \alpha } \assign \sup_{t\geqslant 0 } \, [h(t)]_{\gamma, \alpha, t},
	      \end{align}
	      and the norms
	      \begin{align}\label{eq:def-fbsde-norms-2}
		      \norm{h}_{\gamma}:= \sup_{\alpha: \abs{\alpha}\in \{0,1\}} [h]_{\gamma, \alpha},
		      \qquad
		      \norm{h}_{\gamma, t} \assign \sup_{\alpha: \abs{\alpha}\in \{0,1\}} [h]_{\gamma, \alpha, t}.
	      \end{align}
	      For \(I\subset [0,\infty)\) we define the restricted versions of \eqref{eq:def-fbsde-norms-2}
	      \begin{align}\label{eq:nrm-gamma-restr}
		      \norm{z}_{\gamma, I}  := \sup_{t\in I} \sup_{\abs{\alpha}\leqslant1} \langle t\rangle^{-\frac{\abs{\alpha}}{2}-\gamma} \norm{\nabla^\alpha z_{t}}_{L^\infty(\chi)}.
	      \end{align}
	\item For any \(\alpha\in (0,1)\) we define the Hölder seminorms
	      \begin{align}\label{eq:def-hölder-seminorm}
		      [f]_{\mathcal{C}^{\alpha} } \assign \sup_{x \neq y\in \mathbb{R} ^{2}} \frac{\abs{f(x)-f(y)}}{\abs{x-y}^{\alpha} }.
	      \end{align}
	\item Let \(I\subset \mathbb{N}\) be finite. We define the Steiner diameter of a finite collection of points \(x_I=(x_i)_{i\in I}\subset \mathbb{R}^d\) as the length of the shortest tree connecting all the points in \(x_I\) and denote it by \(\tmop{St}(x_I)\). More precisely,
	      \begin{align}\label{eq:def-Steiner-weight}
		      \tmop{St}(x_I) \assign \min_{x_J\supset x_I} \min_{\tau(x_J)} L(\tau),
	      \end{align}
	      where the first minimum runs over all finite collections of points \(x_{J}\subset \mathbb{R}^{d}\) containing \(x_I\),
	      and the second minimum runs over all trees connecting the points \(x_J\).
	      The function \(L(\tau)\) is the total length of the tree \(\tau\), that is the sum of the lengths of all the edges in the tree \(\tau\).
	      We will frequently use the following estimate
	      \begin{align}\label{eq:Steiner-triangle}
		      \tmop{St}(x_{I}) \leqslant \tmop{St}(x_{I_1 }) + \tmop{St}(x_{I_2}) + {d} (x_{I_1}, x_{I_2}),
	      \end{align}
	      where \({d} (x_{I_1}, x_{I_2}) \assign \min_{i\in I_1, j\in I_2} \abs{x_i-x_j}\) is the standard Euclidean distance between sets; see \cite[Lemma 25]{duchStochasticQuantisationFractional2024} for a proof.
	\item For \(r>0\) we define
	      \begin{align}\label{eq:exp-Steiner}
		      \omega^r_t(x_{1:n}) = \exp(r \sqrt{t} \tmop{St}(x_{1:n})).
	      \end{align}
	      We will often use \(\omega_t\) to convert Steiner distances into improved scaling using
	      \begin{align}\label{eq:Steiner-absorbtion}
		      \tmop{St}(x_{1:n}) \lesssim_{r} t^{-\frac{1}{2}} \omega^r_t (x_{1:n}).
	      \end{align}
\end{itemize}

\textbf{Scale interpolation.}
We fix a heat kernel decomposition to interpolate the covariance of the free field.
More precisely, we define
\begin{align}\label{eq:def-G}
	\dot{G}_{t}= Q_t^{2}, \qquad G_t = \int_{0}^{t} Q_s^{2}\mathd s,
\end{align}
where
\begin{align}\label{eq:def-Q}
	Q_{t} = \left(\frac{1}{t^2} \exp(-\frac{m^{2}-\Delta}{t})\right)^{\frac{1}{2}},
\end{align}
so that \(G_{\infty} = \int_{0}^{\infty} Q_t^{2}\mathd t = (m^{2} - \Delta)^{-1}\).
The operators \(\dot{G}\)  and \(Q\) both have a translation-invariant kernel on \(L^{2}(\mathbb{R}^{2})\), which we denote by the same symbol.
It is easy to check that
\begin{align}\label{eq:def-G-Q-kernel}
	Q_{t}(x) \assign
	\frac{1}{2\pi} \mathe^{-\frac{m^{2}}{2t}} \mathe^{-\frac{t}{2} \abs{x}^{2}},
	\qquad
	\dot{G}_t(x) \assign \frac{1}{4\pi t} \mathe^{-\frac{m^{2}}{t}} \mathe^{-\frac{t}{4} \abs{x}^{2}}.
\end{align}
It is technically important that \(\dot{G}_t= Q_t\ast Q_t\)  has a positive convolutional square root,
and we also use the exponential decay of the heat kernel to measure the quasi-locality of the approximate potential \(V\)
and the force \(F\) later in Section \ref{sec:Flow}.
Further properties of the heat kernels are collected in Appendix \ref{app:hk-estimates}.

\paragraph{Probabilistic notation}
\begin{itemize}
	\item We will generally use Greek letters \(\mu, \nu\) to refer to measures on \(\mathcal{S}'(\mathbb{R}^{2})\)
	      and reserve \(\mathbb{P}, \mathbb{Q}\) for measures on the path space \(C([0,T], \mathcal{S}'(\mathbb{R}^{2}))\).
	\item Let \((\Omega, \mathcal{F}, \mathbb{P})\) be a filtered probability space carrying a Brownian motion \(W\) with quadratic variation \(\langle W \rangle_{t} = G_t \). We equip this space with the augmentation of the filtration generated by \(W\) to ensure the predictable representation property.
	\item Given a measure \(\nu\), we write \(\nu(f)= \int f \mathd \nu\) for the expectation under this measure.
\end{itemize}

\paragraph{Other notation}
\begin{itemize}
	\item For \(n\in \mathbb{N}\), we write \([n]:=\{1,\dots,n\}\).
	\item We fix \(\delta \assign \delta(\beta^{2}) \assign 1-\frac{\beta^{2}}{8\pi}\) as the distance to criticality of the sine-Gordon model in our normalization.
	\item For \(\lambda\in \mathbb{R}\), we write \(\bar{\lambda}= \left| \lambda \right| \). We also define the effective coupling constant
	      \begin{align}\label{eq:def-lambda-t}
		      \lambda_{t} \assign \lambda \mathe^{\frac{\beta^{2}}{2} G_{t}(0)},  \qquad \bar{\lambda}_t = \abs{\lambda_{t}}.
	      \end{align}
	      The estimates on \(G\) in Lemma \ref{lem:G-pw} imply that \(\bar{\lambda}_t \lesssim \bar{\lambda}\langle t \rangle^{1-\delta} \).
	\item For \(\ell\in \mathbb{N}\), we define the thresholds \(\beta^{2}_{\ell}\)
	      for the first appearance of relevant terms of order \(\bar{\lambda}^{\ell}\) as
	      \begin{align}\label{eq:beta-ell}
		      \beta^{2}_\ell \assign \frac{\ell-1}{\ell}\cdot 8 \pi,
		      \qquad
		      \delta^{(\ell)} \assign \delta(\beta^{2}_\ell) = \frac{1}{\ell}.
	      \end{align}
	      We collect the relevant thresholds in Table \ref{tab:beta-delta}.
	      \begin{table}[htbp]
		      \centering
		      \caption{Regimes in sine-Gordon depending on \(\delta = \delta(\beta^{2})\). The constant \(K_F(\beta^{2})\) is as in Theorem \ref{thm:int-Va-estimates}.}
		      \label{tab:beta-delta}
		      \begin{tabular}{l|cccl}
			      \(\ell\)                     & \(2 \)          & \(4 \)          & \(7\)                     & \( \to \infty\) \\
			      \midrule
			      \(\beta^{2}_{\ell}\)         & \(4\pi\)        & \(6\pi\)        & \(\frac{6}{7}\cdot 8\pi\) & \( \to 8\pi\)   \\
			      \(\delta(\beta^{2}_{\ell})\) & \(\frac{1}{2}\) & \(\frac{1}{4}\) & \(\frac{1}{7}\)           & \( \to 0\)      \\
			      \(K_F(\beta^{2})\)           & \(0\)           & \(0\)           & \(<1\)                    & \( \to \infty\)
		      \end{tabular}
	      \end{table}
	\item For a given \(\beta^{2}\in (0, \beta_\ell^{2})\), we always choose \(\kappa=\kappa(\beta^{2})>0\) sufficiently small such that
	      \begin{align}\label{eq:def-delta-kappa}
		      \delta_{\kappa} := \delta(\beta^{2})-\kappa > \delta(\beta^{2}_{\ell}) = \frac{1}{\ell}.
	      \end{align}
\end{itemize}
For \(K\geqslant 1\), define the \(K\)-th order Taylor remainder of a function \(h\in C^{K}(\mathbb{R}^{d})\) around \(x\) evaluated at \(y\) to be
\begin{align}\label{eq:Taylor-rem-single-var}
	R^{K}_{x}[h](y)
	:=h(y)-\sum_{|\alpha|<K}\frac{1}{\alpha!}\nabla^{\alpha} h(x)[(y-x)^{\otimes\alpha}].
\end{align}

To study the Laplace transform \(\Lambda^{\rho, T} (g)\), we need to fix a class of allowed perturbations \(g\).
\begin{assumption}\label{hyp:g-allowed}
	We always assume that \(g:H^{-\varepsilon}(\chi) \to \mathbb{R}\) is a uniformly bounded and Fréchet differentiable functional with uniformly bounded derivative, that is, there is an \(L=L(g)\in (0, \infty)\) such that
	\begin{align}\label{eq:g-bounded}
		\sup_{\varphi\in H^{- \varepsilon}(\chi) }
		\left(
		\abs{g(\varphi)} + \norm{\nabla g(\varphi)}_{L^\infty}
		\right) \leqslant L <\infty.
	\end{align}
\end{assumption}
The class of functions \(g\) satisfying the above conditions is rate-function-determining (see e.g. \cite[Lemma 5]{barashkovStochasticControlApproach2022}).
When \(T=\infty\), we may drop the superscript \(T\) in the notation, e.g. we write \(X^{\rho}=X^{\rho,\infty }  \) and similarly for \(V,F,\mathcal{H}\) and \({H}\).
\section{Construction of the sine-Gordon measure up to \(\beta^{2}<\frac{6}{7}\, 8\pi\)}\label{sec:FBSDE}

In this section, we construct the finite-volume sine-Gordon measure assuming the existence of good approximate
solutions to the flow equation, as made more precise by Assumption \ref{hyp:sublinear-est} below.
We first prove the truncated version of Theorem \ref{thm:sG-6/7} (see Proposition \ref{prop:Girsanov}) for the truncated measures \(\nu_{sG}^{\rho,T}\) using Girsanov's theorem.
This will establish the relationship between \(X^{\rho, T}\) and the regularized measures.
Since we always assume \(T<\infty\), the argument requires nothing more than boundedness and continuity assumptions on the scale interpolation (see Assumption \ref{hyp:scale-interpolation}) and does not depend on \(\beta^{2}\).

In Section \ref{sec:UV-limit}, we use the condition \(\beta^{2}< \frac{6}{7}\; 8\pi\) and the existence of sublinear estimates on the force to prove Theorem \ref{thm:sG-6/7}.
Finally, in Section \ref{sec:var-problem}, we connect Theorem \ref{thm:sG-6/7} to the weak variational problem and prove Proposition \ref{prop:fin-vol-var}.

\subsection{The regularized measure}\label{sec:Girsanov}
In this section, we prove the truncated version of Theorem \ref{thm:sG-6/7}.
We assume that \(T<\infty\) and that we are given a scale interpolation \((V_{t}^{\rho,T} )_{t\in [0,T]}\) of
\(V^{\rho,T}\) satisfying the assumption below.
\begin{assumption}\label{hyp:scale-interpolation}
	Let \(T<\infty\) and \(\rho\prec 1\). We assume that the scale interpolation \((V_t^{\rho,T})_{t\in[0,T]}\) satisfies the following conditions:
	\begin{enumerate}[i)]
		\item \(V_t^{\rho,T}\in C_b^1([0,T], C_b^2(L^\infty(\chi)))\) and depends only on the restriction of \(\varphi\) to \(\tmop{supp}(\rho)\);
		\item the quantities \(\partial_t V_t^{\rho,T}(\varphi)\), \(\Delta_{\dot G_t} V_t^{\rho,T}(\varphi)\), and \(\langle \nabla V_t^{\rho,T}(\varphi), \dot G_t \nabla V_t^{\rho,T}(\varphi)\rangle_{L^2}\) are uniformly bounded in \(t\in[0,T]\) and \(\varphi\in L^{\infty}(\chi) \);
		\item \(F_t^{\rho,T}=-\nabla V_t^{\rho,T}\) is bounded and Lipschitz on \(L^\infty(\chi)\), uniformly in \(t\in[0,T]\).
	\end{enumerate}
\end{assumption}
The conditions will not be satisfied uniformly in \(\rho \prec 1\) and \(T>0\) even when \(\beta^{2}\) is small.

We first show the regularized version of Theorem \ref{thm:sG-6/7}.
\begin{proposition}\label{prop:Girsanov}
	Let \(\mathcal{O}: H^{-\varepsilon}(\chi) \to \mathbb{R}\) be bounded and continuous.
	It holds that
	\begin{align}\label{eq:Girsanov-T-finite}
		\nu_{\tmop{sG}}^{\rho, T}  (\mathcal{O}) = \frac{\mathbb{E}\left[ \mathcal{O}(X_T^{\rho, T}) \mathe^{-\int_{0}^{T} \mathcal{H}_{s}^{\rho, T}(X_s^{\rho,T})\mathd {s}}\right]}{\mathbb{E}\left[\mathe^{-\int_{0}^{T} \mathcal{H}_{s}^{\rho, T}(X_s^{\rho,T})\mathd {s}}\right]},
	\end{align}
	with \(\mathcal{H}^{\rho, T}\) as defined in \eqref{eq:int-def-Hcal}.
\end{proposition}

\begin{proof}
	Regarding the estimate on \(\mathcal{H}\), the assumptions on \(V^{\rho,T} \) imply that for any \(\varphi\in H^{-\varepsilon}(\chi)\),
	\begin{align}\label{eq:Girsanov-non-0}
		0< \mathe^{-C_{\rho, T}} \leqslant \exp\left(- \int_{0}^{T} \mathcal{H}_{s}^{\rho,T} (\varphi) \mathd {s}\right) \leqslant \mathe^{C_{\rho, T}} < \infty.
	\end{align}

	Following the computations in the introduction (\eqref{eq:nu-rho-T}--\eqref{eq:sG-via-X}), we need to check only that \(\mathcal{E}^{V^{\rho, T}}\) as defined in \eqref{eq:def-stoch-exp} is a uniformly integrable martingale and that \eqref{eq:int-truncated-SDE} has a unique strong solution.

	The well-posedness of \eqref{eq:int-truncated-SDE} on \([0,T]\) is standard and can be obtained, for example, from a simplified version of the proof of Proposition \ref{prop:Xrho-wp} included in the next section. The proof is therefore omitted.
	The required regularity \(X_{T}^{\rho, T}\in H^{-\varepsilon}(\chi)\) is Lemma \ref{lem:X_T-regularity} below.

	Regarding the stochastic exponential, the boundedness and compact support from Assumption \ref{hyp:scale-interpolation} ensure that
	\begin{align}
		\int_0^T \norm{Q_s F_s^{\rho, T}(W_s)}_{L^{2}}^{2} \mathd s
		\lesssim_{\rho} \int_{0}^{T} \langle s \rangle^{-2} \norm{F_s^{\rho,T}(W_s)}_{L^\infty}^{2} \mathd s
		\lesssim C_{\rho, T},
	\end{align}
	which directly implies uniform integrability for \(\mathcal{E}^{V^{\rho, T}}\).
	As a result, we may define the equivalent probability measure
	\(\mathd \mathbb{P}^{F^{\rho,T}} := \mathcal{E}_{0,T} \mathd {\mathbb{P}} \).
	Under this tilted measure, Girsanov's theorem implies that the process
	\begin{align}\label{eq:Girsanov-BM}
		\mathd \tilde{W}_{t} = \mathd {W_t} - \dot{G}_{t} F_{t}^{\rho,T}(W_t) \mathd {t},
	\end{align}
	is a \(\mathbb{P}^{F^{\rho,T}} \)-Brownian motion.

	Therefore, using that \(V_0(0)\) is deterministic,
	\begin{align}\label{eq:Girsanov-insert}
		\mathbb{E}\left[ \mathcal{O}(W_{T}) \mathe^{-V_{T}^{\rho,T} (W_{T})}\right]
		= & \mathbb{E}\left[ \mathcal{O}(W_{T}) \mathe^{-V_{0}^{\rho,T} (0)}\mathe^{-\int_{0}^{T} \mathcal{H}^{\rho, T}_s(W_{s})\mathd {s}} \mathcal{E}_{0,T}\right]   \\
		= & \mathbb{E}^{F^{\rho,T}} \left[\mathcal{O}(W_{T}) \mathe^{-V_{0}^{\rho,T} (0)}\mathe^{-\int_{0}^{T} \mathcal{H}^{\rho, T}_s(W_{s})\mathd {s}}\right]        \\
		= & \mathbb{E}\left[\mathcal{O}(X^{\rho,T} _{T}) \mathe^{-V_{0}^{\rho,T} (0)}\mathe^{-\int_{0}^{T} \mathcal{H}^{\rho, T}_s(X^{\rho,T} _{s})\mathd {s}}\right],
	\end{align}
	where \(\mathbb{E}^{F^{\rho,T}}\) denotes expectation with respect to \(\mathbb{P}^{F^{\rho,T}}\).
	Combining all of the above with the definition of the approximate measure \(\nu_{\tmop{sG}}^{\rho, T}\),
	we arrive at the claim.
\end{proof}
\begin{lemma}\label{lem:X_T-regularity}
	There is a unique strong solution to \eqref{eq:int-truncated-SDE} with \(X^{\rho, T} \in C([0,T]; H^{-\varepsilon}(\chi))\) almost surely.
\end{lemma}
\begin{proof}
	Strong uniqueness and existence on \([0,T]\) for \eqref{eq:int-truncated-SDE} follow directly from the standard well-posedness for Lipschitz SDEs with bounded coefficients.
	The regularity is proved in the same way as in the more general Lemma \ref{lem:Zh-aprioi} and Proposition \ref{prop:Xrho-wp} below and is therefore omitted here.
\end{proof}

\subsection{Controlling the UV limit}\label{sec:UV-limit}
In this section, we prove Theorem \ref{thm:sG-6/7} assuming the existence of an approximate effective potential with sublinear bounds for the force.
More precisely, we always assume the following.
\begin{assumption}\label{hyp:sublinear-est}
	Let \(\beta^{2}<\beta^{2}_{7}:=\frac{6}{7}\cdot 8\pi\) and fix \(\ell^{\ast}=7\) and \(\rho \prec 1\).
	For every \(D>1\), \(\gamma\in (0,\kappa)\), let \(\mathcal{Y}_{\gamma, D}:=\{h\in \mathcal{S}'(\mathbb{R}^{2});\; [h]_{\gamma, 1}\leqslant D\}\).
	There is a family of scale interpolations \((V^{\rho, T})_{T>0}\) and \(K_F=K_F(\beta^{2})\in [0,1)\) such that with \(F^{\rho, T} = - \nabla V^{\rho, T} \) the following properties hold for any \(t\geqslant0\), \(E\in \{L^\infty(\chi), L^2(\chi)\}\), and \(\varphi\in\mathcal{Y}_{\gamma,D}\).
	\begin{enumerate}[a)]
		\item \textbf{Compatibility.} For each \(T<\infty\), Assumption \ref{hyp:scale-interpolation} holds.
		\item \textbf{Support properties.}
		      For any \(t,T>0\), we have \(V_{t}^{\rho, T} = V_{t\wedge T}^{\rho ,T}\); moreover, \(V^{\rho, T}\) depends only on the restriction of \(\varphi \) to \(\tmop{supp}(\rho )\).
		\item \textbf{Uniform bounds.}
		      \begin{align}\label{eq:hyp-V,F,H-uniform-bounds}
			      \sup_{T\geqslant 0 }\abs{V^{\rho, T}_{t}[\varphi]}
			       & \lesssim_{\rho} \bar{\lambda}(1+\bar{\lambda})^{\ell^\ast-2} \langle t  \rangle ^{1 - \delta_\kappa} (1+[\varphi]_{\gamma, 1, t})^{2},                           \\
			      \sup_{T\geqslant 0}\norm{Q_{t}F^{\rho, T}_{t}[\varphi]}_{E}
			       & \lesssim_{\rho} \bar{\lambda}(1+\bar{\lambda})^{\ell^\ast-2} \langle t  \rangle ^{- \delta_\kappa} (1+[\varphi]_{\gamma, 1, t})^{K_F},                           \\
			      \sup_{T\geqslant 0 } \abs{\mathcal{H}_{t}^{\rho, T}[\varphi]}
			       & \lesssim_{\rho} \bar{\lambda}^{\ell^\ast}(1+\bar{\lambda})^{\ell^\ast-2} \langle t  \rangle ^{-\delta_{\kappa} \ell^\ast}(1 + [\varphi]_{\gamma, 1, t})^{2 K_F}.
		      \end{align}
		\item \textbf{Local Lipschitz continuity.}
		      For every \(D>0\) and \(\varphi, \tilde{\varphi}\in \mathcal{Y}_{\gamma, D}\), it holds that
		      \begin{align}\label{eq:hyp-V,F,H-local-Lipschitz}
			      \sup_{T\geqslant 0 }\abs{V^{\rho, T}_{t}[\varphi]-V^{\rho, T}_{t}[\tilde{\varphi}]}
			       & \lesssim_{\rho} D^{2} \bar{\lambda}(1+\bar{\lambda})^{\ell^\ast-2} \langle t  \rangle ^{-\delta_{\kappa}}  \norm{\varphi-\tilde{\varphi}}_{t, \gamma},                        \\
			      \sup_{T\geqslant 0}\norm{Q_{t}F^{\rho, T}_{t}[\varphi]-Q_{t}F^{\rho, T}_{t}[\tilde{\varphi}]}_{E}
			       & \lesssim_{\rho} D^{K_F} \bar{\lambda}(1+\bar{\lambda})^{\ell^\ast-2} \langle t  \rangle ^{-\delta_{\kappa}}  \norm{\varphi-\tilde{\varphi}}_{t, \gamma},                      \\
			      \sup_{T\geqslant 0 }\abs{\mathcal{H}^{\rho, T}_{t}[\varphi]-\mathcal{H}^{\rho, T}_{t}[\tilde{\varphi} ]}
			       & \lesssim_{\rho}D^{2 K_F}\bar{\lambda}^{\ell^\ast}(1+\bar{\lambda})^{\ell^\ast-2} \langle t  \rangle ^{-\delta_{\kappa}\ell^\ast}  \norm{\varphi-\tilde{\varphi}}_{t, \gamma}.
		      \end{align}
		\item \textbf{Convergence in \(T\).}
		      For every \(D>0\), \(T_2 > T_1 \geqslant 0\), there is an \(\varepsilon=\varepsilon(\beta^{2})>0\) such that for any \(\varphi\in \mathcal{Y}_{\gamma, D}\)
		      \begin{align}\label{eq:hyp-V,F,H-T-convergence}
			      \abs{V^{\rho, T_1 }_{t}[\varphi]-V^{\rho, T_2 }_{t}[\varphi]}
			       & \lesssim_{\rho} D^{2 } \bar{\lambda}(1+\bar{\lambda})^{\ell^\ast-2} \langle T_1 \rangle^{-\varepsilon} \langle t \rangle^{1-\delta_{\kappa}}  ,                            \\
			      \norm{Q_{t}F^{\rho, T_1 }_{t}[\varphi]- Q_{t}F^{\rho, T_2 }_{t}[\varphi]}_{E}
			       & \lesssim_{\rho} D^{K_F } \bar{\lambda}(1+\bar{\lambda})^{\ell^\ast-2} \langle T_1 \rangle^{-\varepsilon} \langle t \rangle^{-\delta_{\kappa}},                             \\
			      \abs{\mathcal{H}^{\rho, T_1 }_{t}(\varphi)-\mathcal{H}^{\rho, T_2 }_{t}(\varphi)}
			       & \lesssim_{\rho} D^{2 K_F } \bar{\lambda}^{\ell^{\ast}}(1+\bar{\lambda})^{\ell^\ast-2} \langle T_1 \rangle^{-\varepsilon} \langle t \rangle^{-\delta_{\kappa} \ell^{\ast}}.
		      \end{align}
	\end{enumerate}
\end{assumption}
We always work under Assumption \ref{hyp:sublinear-est}.
Fix \(T\in [0,\infty]\).
It is convenient to isolate the drift \(Z_t^{\rho,  T} = X_{t}^{\rho,  T} - W_{t}\).
Define also the controlled ODE for \(Z\), with the noise replaced by a generic function \(h\in C([0,\infty), \mathcal{S}'(\mathbb{R}^{2}))\) such that \(\norm{h}_{\gamma}<\infty\),
\begin{align}\label{eq:ode-Zh}
	Z_t^{\rho,  T} (h)
	= \int_0^{t\wedge T} \dot{G}_s F_s^{\rho,  T} (Z_s^{\rho,  T} + h_s) \mathd s, \qquad t\in [0,\infty].
\end{align}

\begin{proposition}\label{prop:Xrho-wp}
	For every \(T\in [0,\infty]\), there is a unique strong solution to \eqref{eq:int-truncated-SDE} on \([0,T]\).
	Moreover,
	\begin{enumerate}[i)]
		\item the solution is given by
		      \begin{align}\label{eq:Xrho-fromZh}
			      X_t^{\rho,T} = Z_t^{\rho, T} + W_t = Z_t^{\rho, T} (W_{[0,t]}) + W_t,
		      \end{align}
		      where \(Z^{\rho, T} (h)\) is the unique solution to the ODE \eqref{eq:ode-Zh} and \(W_{[0,t]}\) is the restriction of \(W\) to the interval \([0,t]\).
		\item For every \(\varepsilon>0\), there exists \(c>0\) sufficiently small and independent of \(T\in [0,\infty]\) such that
		      \begin{align}\label{eq:Xrho-exp-int}
			      0< & \mathbb{E}\left[\exp\left(-\frac{c}{2} \sup_{T\geqslant0}\norm{X^{\rho,T}}_{\gamma}^{2} - \frac{c}{2}  \sup_{T\geqslant0} \norm{X^{\rho,T}_{T}}_{H^{-\varepsilon}(\chi)}^2 \right)\right] \\
			         & \leqslant
			      \mathbb{E}\left[\exp\left(\frac{c}{2} \sup_{T\geqslant0}\norm{X^{\rho, T}}_{\gamma}^{2} + \frac{c}{2}\sup_{T\geqslant0}\norm{X^{\rho, T}_T }_{H^{-\varepsilon}(\chi)}^2 \right)\right] < \infty.
		      \end{align}
		\item The solution \(X^{\rho,T}\) is Cauchy in \(T\geqslant0\), more precisely, there is an \(\varepsilon= \varepsilon(\beta^{2})>0\) such that for any \(T_2 > T_1 \geqslant0\),
		      \begin{align}\label{eq:Xrho-T-convergence}
			      \norm{X^{\rho,T_1 } - X^{\rho,T_2 }}_{\gamma} \lesssim \exp\left(C (1+ \norm {W}_{\gamma})\right) \langle T_1\rangle ^{-\varepsilon},
		      \end{align}
		      and moreover \(\norm{X^{\rho, T}_T - X^{\rho}_{\infty}}_{H^{-\varepsilon}(\chi)} \to 0\) almost surely and in every \(L^p(\mathd \mathbb{P})\).
	\end{enumerate}
\end{proposition}
The proof of Proposition \ref{prop:Xrho-wp} follows from a deterministic analysis of \eqref{eq:ode-Zh},
starting with the following a priori estimates, using the sublinear estimates on the force \(\dot{G} F^{\rho, T}\).
\begin{lemma}\label{lem:Zh-aprioi}
	Let \(Z^{\rho, T}(h)\) be a solution to \eqref{eq:ode-Zh}. For any \(\gamma \in (0,\kappa)\) and \(\alpha < 2\delta_{\kappa} = 2-\frac{\beta^2}{4\pi}-2\kappa\in (0,2)\), the following estimates hold uniformly in \(T\),
	\begin{align}\label{eq:nablaZ-6/7}
		\| Z^{\rho,  T}(h)  \|_{\gamma}\lesssim (1 + \norm {h}_{\gamma}),
		\qquad
		\sup_{T}\norm{Z_T^{\rho,T} }_{H^{\alpha}(\chi)}\lesssim_\alpha (1+\norm{h}_{\gamma}).
	\end{align}
\end{lemma}
\begin{proof}
	Let \(Z^{\rho, T}(h)\) be a solution to \eqref{eq:ode-Zh}.
	Set \(\Lambda_{\lambda}\assign \bar{\lambda}(1+\bar{\lambda})^{\ell^\ast-2}\).
	We prove the bound for each derivative level \(\abs{k}\in \{0,1\}\).
	Starting from \eqref{eq:ode-Zh} and inserting Lemma \ref{lem:hk-Lp-bounds} and the estimate \eqref{eq:hyp-V,F,H-uniform-bounds} on \(\dot{G}_t F_t^{\rho, T}\), we obtain
	\begin{align}\label{eq:a-priori-Zh-1}
		\norm{\nabla^{k} Z_{t}^{\rho, T}}_{L^\infty(\chi)}
		\leqslant & \int_0^t \norm{\nabla^{k} Q_s}_{L^1(\chi^{-1})} \norm{Q_s F_s^{\rho, T}(Z_s^{\rho,T}+h_s)}_{L^\infty(\chi)} \mathd {s}                              \\
		\lesssim  & \int_0^t \langle s\rangle^{-1+\frac{\abs{k}}{2}-\delta_{\kappa}} \Lambda_{\lambda}\left(1+[Z_s^{\rho,T}+h_s]_{\gamma,1,s}\right)^{K_F}  \mathd {s}.
	\end{align}
	Using that \(K_F<1\) by assumption, for every \(\varepsilon>0\) and with \(C_{K_F} = (1-K_F)K_F^{\frac{K_F}{1-K_F}} \), we have
	\begin{align}
		\Lambda_{\lambda}\left(1+[Z_s^{\rho,T}+h_s]_{\gamma,1,s}\right)^{K_F}
		\leqslant \varepsilon \left(1+\norm{Z^{\rho,T}}_{\gamma} + \norm{h}_{\gamma}\right) + C_{K_F}\varepsilon^{-\frac{K_F}{1-K_F}} \Lambda_{\lambda}^{\frac{1}{1-K_F}}.
	\end{align}
	Multiplying both sides by \(\langle t\rangle^{-\frac{\abs{k}}{2}-\gamma}\), we find
	\begin{align}\label{eq:a-priori-Zh-2}
		\norm{Z^{\rho, T}}_{\gamma}
		=        & \sup_{\abs{k}\leqslant1}\sup_{t>0} \langle t\rangle^{-\frac{\abs{k}}{2}-\gamma}  \norm{\nabla^k Z_t^{\rho,T}}_{L^\infty(\chi)}                                                                                             \\
		\lesssim & \left(\varepsilon\left(1 + \norm{Z^{\rho,T}}_{\gamma} + \norm{h}_{\gamma}\right) + \Lambda_{\lambda}^{\frac{1}{1-K_F}} C_{K_{F}(\beta^{2}),\varepsilon}\right)  \int_0^T \langle s\rangle^{-1-\delta_{\kappa}}  \mathd {s} \\
		\lesssim & \varepsilon\left(1 + \norm{Z^{\rho,T}}_{\gamma} + \norm{h}_{\gamma}\right) + \Lambda_{\lambda}^{\frac{1}{1-K_F}} C_{K_{F}(\beta^{2}), \varepsilon}.
	\end{align}
	Choosing \(\varepsilon \) sufficiently small, the first estimate follows after rearranging,
	\begin{align}
		\norm{Z^{\rho, T}}_{\gamma} \leqslant C_{K_{F}(\beta^{2})}\left(\norm{h}_{\gamma} + 1 +  2\Lambda_{\lambda}^{\frac{1}{1-K_F}}\right).
	\end{align}
	For the second estimate, we combine the estimate just derived with the compact-support assumption \(\rho \prec 1\), the estimate \eqref{eq:hyp-V,F,H-uniform-bounds}, and Lemma \ref{lem:hk-regularity} to find that, for any \(\alpha < 2\delta_\kappa = 2-\frac{\beta^{2}}{4\pi}- 2\kappa\),
	\begin{align}\label{eq:a-priori-Zh-regularity}
		\norm{Z^{\rho,T}_{T}}_{H^{\alpha}(\chi)}^{2}
		\lesssim & \int_{0}^{T} \langle s\rangle^{\alpha-1} \norm{Q_s (F_s^{\rho,T}(Z_{s}^{\rho,T}+h_s))}_{L^{2}(\chi)}^{2} \mathd {s} \\
		\lesssim & (1+\norm{h}_{\gamma})^{2K_F} \int_{0}^{\infty} \langle s\rangle^{-2\delta_{\kappa}+\alpha-1} \mathd {s} < \infty.
	\end{align}
\end{proof}
\begin{remark}\label{rem:local-apriori}
	The same statement applies in a localized version.
	Let \(S\in[0,T]\) and suppose that
	\(Y\) satisfies
	\begin{align}\label{eq:local-apriori-equation}
		Y_t
		=
		\int_0^t
		\dot G_s F_s^{\rho,T}(Y_s+h_s)\,\mathd s,
		\qquad t\in[0,S].
	\end{align}
	Then, repeating the proof of Lemma \ref{lem:Zh-aprioi}, but taking
	the supremum only over \(t\in[0,S]\), we obtain
	\begin{align}\label{eq:local-apriori-bound}
		\norm{Y}_{\gamma,[0,S]}
		\leqslant
		C_{\bar\lambda}(1+\norm h_\gamma),
	\end{align}
	where the constant is independent of \(S\) and \(T\).
\end{remark}
\begin{remark}\label{rem:general-Besov}
	Repeating the arguments used in \cite{gubinelliFBDSEApproachSine2026}, we could obtain regularity of \(Z^{\rho}\) in \(B^\alpha_{p, p}(\chi)\)  for \(p\in [1,\infty]\) and \(\alpha<2\delta_{\kappa}\).
\end{remark}

\begin{lemma}\label{lem:Zrho-w-p}
	For every \(T\in[0,\infty]\), there is a unique solution \(Z^{\rho,  T}(h)\) to the ODE \eqref{eq:ode-Zh}.
	This solution is non-anticipating
	and there are constants \(\varepsilon = \varepsilon(\beta^{2}), C=C(\beta^{2})>0\) such that
	\begin{align}\label{eq:Zrho-T-dependence}
		\norm{Z^{\rho,T_1 }-Z^{\rho, T_2}}_{\gamma} \lesssim \exp(C (1+\norm{h}_{\gamma})) \langle T_2 \wedge T_1\rangle^{-\varepsilon(\beta^{2})}.
	\end{align}
\end{lemma}
\begin{proof}
	The proof is in three steps.
	Let \(T<\infty\). The case \(T=\infty\) is recovered later.
	We first show local existence and uniqueness using a fixed point argument.
	We then use the good a priori bound in Lemma \ref{lem:Zh-aprioi} to continue the solution globally.

	\emph{Local existence and uniqueness.}
	Define the Banach space
	\(E_{\chi} \assign \{f\in W_{\tmop{loc}}^{1,\infty}(\mathbb{R}^{2}); \; \chi f, \chi \nabla f\in L^{\infty}(\mathbb{R}^{2})\}\) with the norm
	\begin{align}
		\norm{f}_{E_{\chi}} = \norm{f}_{L^\infty(\chi)} + \norm{\nabla f}_{L^\infty(\chi)},
	\end{align}

	For any \(t_0, \tau\), define \(\mathcal{X}_{\gamma}(t_0, \tau):= C([t_0, t_0 + \tau]; E_{\chi})\) with the equivalent norm \(\norm{\cdot}_{\gamma, I}\) as defined in \eqref{eq:nrm-gamma-restr}.
	Given \(D>1\), denote by \(B^D({t_0,\tau})\) the closed ball of radius \(D\) in \(\mathcal{X}_{\gamma}(t_0, \tau)\), that is
	\begin{align}
		B^D(t_0, \tau) = \{ z\in \mathcal{X}_{\gamma}(t_0, \tau); \norm{z}_{\gamma, [t_0, t_0+\tau]} \leqslant D\}.
	\end{align}
	For every \(z\in B^D({t_0,\tau})\), define (the non-anticipating) solution map
	\begin{align}\label{eq:def-Tcal-Zh}
		Z^z_{t}(h) \assign (\mathcal{T}z)_{t} \assign \int_{t_0}^{t} \dot{G}_{s} F_{s}^{\rho, T}(z_s+h_s) \mathd {s}, \qquad t\in [t_0, t_0+\tau].
	\end{align}

	Following the same steps as in the proof of Lemma \ref{lem:Zh-aprioi}, we see that for \(D>\norm{h}_{\gamma}\vee 1\)
	\begin{align}\label{eq:B-to-B}
		\norm{\mathcal{T}z}_{\gamma, [t_0, t_0+\tau]} \leqslant C_{ \bar{\lambda} } \tau D.
	\end{align}
	In other words, there is a \(\tau_0=\tau_0(D)\), independent of \(t_0, T\), such that for every \(\tau<\tau_0\), we have \(\mathcal{T}(B^{D} (t_0, \tau))\subset B^{D}(t_0, \tau)\).

	To see that \(\mathcal{T}\) is also a contraction on \(B^{D}(t_0, \tau) \) for \(\tau\) sufficiently small,
	let \(z^1, z^2\in B^{D}(t_0, \tau)\).
	By the local Lipschitz condition \eqref{eq:hyp-V,F,H-local-Lipschitz} from Assumption \ref{hyp:sublinear-est},
	\begin{align}\label{eq:contraction-Lipschitz}
		\norm{Q_s(F_{s}^{\rho,T}(z_s^1 + h_s) - F_{s}^{\rho,T}(z_s^2 + h_s))}_{L^\infty(\chi)}
		\lesssim \langle s\rangle^{-\delta_{\kappa}}  (D \vee \norm{h}_{\gamma})^{K_F} \norm{z_s ^{1}- z_s^2 }_{\gamma, s}.
	\end{align}
	Thus, for \(D>\norm{h}_{\gamma}\vee 1\),
	\begin{align}\label{eq:contraction-ins}
		\norm{\mathcal{T}z^1- \mathcal{T}z^2}_{\gamma, [t_0, t_0+\tau]} \leqslant C_{\bar{\lambda}} \tau D \norm{z^1-z^2 }_{\gamma,  [t_0, t_0+\tau]}.
	\end{align}
	By possibly choosing \(\tau_0(D)\) smaller, it holds that \(C_{\bar{\lambda}}D \tau<1\) whenever \(\tau<\tau_0(D)\),
	and \(\mathcal{T}\) is a strict contraction.
	This shows local existence and uniqueness on \([t_0, t_0+\tau]\).

	\emph{Global solutions.}
	We combine a version of the a priori estimate in Lemma \ref{lem:Zh-aprioi} with the local existence and uniqueness just derived.
	To this end, fix \(D>((C_{\bar{\lambda}})\vee 1) (1+\norm{h}_{\gamma})\) and choose \(\tau < \tau_0(2D)\).
	If \(T\leqslant\tau\), then there is nothing to do and the local construction already gives a global solution.
	Let \(T>\tau\), let \(N_T = \ceil{\frac{T}{\tau}}\), and define \(t_k=(k \tau )\wedge T\) for \(k\in \{1,\dots, N_T\}\).
	Then \(0=t_0<t_1<\dots < t_{N_T}= T\) and \(t_k-t_{k-1}\leqslant \tau\).
	Define \(\bar{Z}_{t}^{\rho, T, 0} = 0\) and, for \(k\geqslant1\), define \(\bar Z^{\rho,T,k}\) as the unique solution to the ODE
	\begin{align}\label{eq:extension-Zh}
		\bar Z_t^{\rho,T,k}
		=
		\begin{cases}
			\bar Z_{t_{k-1}}^{\rho,T,k-1}
			+
			\displaystyle
			\int_{t_{k-1}}^t
			\dot G_s
			F_s^{\rho,T}
			\bigl(\bar Z_s^{\rho,T,k}+h_s\bigr)\,\mathd s,
			 &
			t\in[t_{k-1},t_k],
			\\[1.2ex]
			\bar Z_{t_k}^{\rho,T,k},
			 &
			t\geqslant t_k.
		\end{cases}
	\end{align}
	Let \(Z^{\rho, T, 0} \equiv 0 \)  and define
	\begin{align}\label{eq:extension-glue}
		Z^{\rho, T, N}_t   \assign \sum_{k=1}^{N} \mathbb{1}_{(t_{k-1}, t_{k}]}(t ) \bar{Z}_{t}^{\rho, T, k}.
	\end{align}
	We set \(Z_0^{\rho,T,N}=0\).
	By definition, if the family \((\bar{Z}^{\rho, T, k} )_{k\leqslant N}\) satisfies \eqref{eq:extension-Zh}, then \(({Z}^{\rho, T, N}_{t\wedge T})_{t\geqslant0}\)
	is a solution to the ODE
	\begin{align}\label{eq:partial-ode}
		Z_{t}^{\rho,T,k} = \int_0^{t\wedge t_k} \dot{G}_s F_s^{\rho,T}(Z_s^{\rho,T,k} + h_s )\mathd s, \qquad t\geqslant0,
	\end{align}
	which is exactly \eqref{eq:ode-Zh} when \(k=N_T\) since \(t_{N_T} = T\).
	We will proceed by induction.
	By the local existence and uniqueness result, there is a unique solution \(\bar{Z}^{\rho, T, 1}\) to \eqref{eq:extension-Zh},
	and by the localized version of Lemma \ref{lem:Zh-aprioi} (Remark \ref{rem:local-apriori}), \(\norm{\bar{Z}^{\rho, T, 1}}_{\gamma}\leqslant D\).
	Suppose now that \(Z^{\rho,T,k-1}\) has been constructed and satisfies
	\(\norm{Z^{\rho,T,k-1}}_\gamma\leqslant D\).
	We have for every \(t\in (t_{k-1}, t_{k}]\),
	\begin{align}\label{eq:delta-Zbar}
		\delta_k \bar{Z}^{\rho, T}_t
		= \bar{Z}_{t}^{\rho, T, k} - \bar{Z}_{t}^{\rho, T, k-1}
		= \int_{t_{k-1}}^{t} \dot{G}_s F_s^{\rho, T} (\delta_k\bar{Z}_s^{\rho, T} + \bar{Z}_{s}^{\rho, T, k-1} + h_s)  \mathd s.
	\end{align}
	Using that \(\norm{\bar Z^{\rho,T,k-1}+h}_{\gamma,[t_{k-1},t_k]}\leqslant2D\), the definition of \(\tau\) and the contraction argument on \(B^{2D}(t_{k-1}, t_k)\) from the first step,
	we see that there is a unique solution to \eqref{eq:delta-Zbar} and therefore also to \eqref{eq:extension-Zh} and \eqref{eq:partial-ode}.
	Therefore, the localized a priori estimate again yields \( \norm{Z^{\rho,T,k}}_\gamma\leqslant D\).
	Uniqueness follows from the uniqueness of the solution on each interval \((t_{k-1}, t_k]\).

	\emph{Convergence in \(T\).}
	Let \(Z^{\rho, T_1}(h), Z^{\rho, T_2 }(h)\)  be the unique solutions to \eqref{eq:ode-Zh} with \(T_1, T_2 \) respectively.
	Assume that \(T_1 < T_2\).
	From Lemma \ref{lem:Zh-aprioi} we know that \(\norm{Z^{\rho, T_1}(h)}_\gamma, \norm{Z^{\rho, T_2}(h)}_{\gamma} \lesssim 1+\norm{h}_{\gamma}\).
	We decompose the difference as follows:
	\begin{align}\label{eq:Z-split}
		Z_t^{\rho, T_1} - Z_t^{\rho, T_2}
		= (Z_{t\wedge T_1 }^{\rho, T_1} - Z_{t\wedge T_1 }^{\rho, T_2 } )
		- \int_{t \wedge T_1}^{t\wedge T_2} \dot{G}_sF_s^{\rho,T_2} (Z_s^{\rho,T_2} + h_s )\mathd s.
	\end{align}
	The second term, the tail, is controlled as
	\begin{align}
		\norm{\int_{t\wedge T_1}^{t\wedge T_2} \dot{G}_sF_s^{\rho, T_2} (Z_s^{\rho,T_2} + h_s )\mathd s}_{\gamma, t}
		\leqslant (1 + \norm{h}_{\gamma}) \int_{t \wedge T_1}^{t\wedge T_2} \langle s \rangle^{-1-\delta_{\kappa}}\mathd s
		\lesssim    (1 + \norm{h}_{\gamma}) \langle t\wedge T_1 \rangle^{-\delta_{\kappa}}.
	\end{align}
	The first term in \eqref{eq:Z-split} is controlled using \eqref{eq:hyp-V,F,H-T-convergence} from Assumption \ref{hyp:sublinear-est},
	\begin{align}\label{eq:T-conv}
		          & \norm{ Q_s F^{\rho, T_1}_{s}(Z^{\rho, T_1 }_s + h_s) - Q_s F^{\rho, T_2}_{s}(Z^{\rho, T_2 }_s + h_s)}_{L^\infty(\chi)}                                                                   \\
		\leqslant & \norm{ Q_s F^{\rho, T_1}_{s}(Z^{\rho, T_1 }_s + h_s) - Q_s F^{\rho, T_1}_{s}(Z^{\rho, T_2 }_s + h_s)}_{L^\infty(\chi)}
		+ \norm{ Q_s( F^{\rho, T_1}_{s}- F^{\rho, T_2}_{s})(Z^{\rho, T_2 }_s + h_s)}_{L^\infty(\chi)}                                                                                                        \\
		\lesssim  & \langle s\rangle^{-\delta_{\kappa}} (1+\norm{h}_{\gamma}) \norm{Z^{\rho, T_1 }_s-Z^{\rho, T_2 }_s}_{\gamma, s}
		+ \langle s\rangle^{- \delta_{\kappa}} \langle T_2\wedge T_1 \rangle^{-\varepsilon(\beta^{2})}  (1+ \norm{h}_{\gamma})                                                                               \\
		\lesssim  & \langle s\rangle^{ - \delta_{\kappa}}(1+\norm{h}_{\gamma}) \left(\norm{Z_s^{\rho, T_1 }-Z_s^{\rho, T_2 } }_{\gamma, s} + \langle T_2\wedge T_1 \rangle^{-\varepsilon(\beta^{2})}\right).
	\end{align}
	Thus, choosing \(\varepsilon < \varepsilon (\beta ^{2}) \wedge \delta_{\kappa }\) and using \(t\geqslant s\), \(T_1<T_2\), we obtain
	\begin{align}\label{eq:T-conv-2}
		         & \norm{Z_t^{\rho,T_1 } - Z_t^{\rho,T_2}}_{\gamma, t}                                                                                                                                                                                                                                                  \\
		\lesssim & (1+\norm{h}_{\gamma}) \sum_{\abs{\alpha}\leqslant1}\langle t \rangle^{-\gamma-\frac{\abs{\alpha}}{2}}  \int_0^{t} \langle s\rangle^{-1-\delta_{\kappa}+\frac{\abs{\alpha}}{2}} \left(\norm{Z^{\rho, T_1 } _s - Z_s^{\rho, T_2 } }_{\gamma, s} + \langle T_1 \rangle^{-\varepsilon}\right) \mathd {s} \\
		\lesssim & (1+\norm{h}_{\gamma})  \left( \langle T_1 \rangle^{-\varepsilon} + \int_0^{t} \langle s\rangle^{-1-\delta_{\kappa}} \norm{Z^{\rho, T_1 } _s - Z_s^{\rho, T_2 } }_{\gamma, s}\mathd {s} \right).
	\end{align}
	The claim \eqref{eq:Zrho-T-dependence} now follows from Grönwall's inequality and \((Z^{\rho,T}(h))_{T\geqslant0}\) is Cauchy in \(\norm{\cdot}_{\gamma}\).
	Denote the unique limit by \(Z^{\rho}(h) = Z^{\rho, \infty}(h) \).
	Using the continuity of \(F^{\rho,T}\) and the arguments used above, we find that \(Z^{\rho}(h)\) is a solution to \eqref{eq:ode-Zh} for \(T=\infty\).
	Uniqueness follows from uniqueness on every small sub-interval thanks to the Lipschitz continuity of \(F^{\rho}\).
\end{proof}
\begin{lemma}\label{lem:Zh-infty-existence}
	Under the assumptions of Lemma \ref{lem:Zrho-w-p}, the terminal point \(Z_{\infty}^{\rho, T}(h) \) exists for every \(T\in[0,\infty]\) and for \(\alpha < 2\delta_{\kappa}\),
	\begin{align}
		\norm{Z_{t}^{\rho,T}(h) - Z_{\infty}^{\rho,T}(h)}^{2}_{H^{\alpha}(\chi)} \lesssim \langle t \rangle^{\alpha-2\delta_{\kappa}} (1+\norm{h}_{\gamma})^{2K_F}.
	\end{align}
\end{lemma}
\begin{proof}
	Let \(t_1,t_2\in(0,T]\) when \(T<\infty\), and let \(t_1,t_2\in(0,\infty)\) when \(T=\infty\).
	By Lemma \ref{lem:hk-regularity},
	\begin{align}
		\norm{Z^{\rho,T}_{t_1}(h) - Z^{\rho,T}_{t_2}(h) }_{H^{\alpha}(\chi)}^{2}
		\leqslant \int_{t_1 }^{t_2} \norm{Q_{s} F_{s}^{\rho, T}(Z_s^{\rho,T}(h)+h_s)}_{L^{2}(\chi)}^{2}\langle s \rangle^{\alpha-1} \mathd s
		\lesssim & \int_{t_1 }^{t_2 } \langle s \rangle^{\alpha-1- 2\delta_{\kappa}} \mathd s (1+\norm{h}_{\gamma})^{2K_F} \\
		\lesssim & \langle t_1 \rangle^{\alpha-2\delta_{\kappa}} (1+\norm{h}_{\gamma})^{2K_F}.
	\end{align}
	For finite \(T\), take \(t_2=T\); for \(T=\infty\), let \(t_2\to\infty\), which readily implies the claim.
\end{proof}

\begin{proof}[Proof of Proposition \ref{prop:Xrho-wp}]
	\begin{enumerate}[i)]
		\item Note that \(\norm{W}_{\gamma} < \infty\), so that we may apply Lemma \ref{lem:Zrho-w-p} pathwise.
		      By definition of \(Z^{\rho,  T}\), the process \eqref{eq:Xrho-fromZh} is an (adapted) solution to \eqref{eq:int-truncated-SDE}.
		      Uniqueness for \eqref{eq:ode-Zh} from Lemma \ref{lem:Zrho-w-p} implies uniqueness for \(X^{\rho, T}\).
		\item From Lemma \ref{lem:Zh-aprioi}, we have the pathwise estimate
		      \begin{align}\label{eq:Xrho-exp-prep}
			      \sup_{T\geqslant0}\norm{X^{\rho,T} }_{\gamma} + \sup_{T\geqslant0}\norm{X^{\rho, T}_{T}}_{H^{-\varepsilon}(\chi) } \lesssim 1+\norm{W}_{\gamma} + \sup_{T\geqslant0}\norm{W_{T}}_{H^{-\varepsilon}(\chi)}.
		      \end{align}
		      As a result, by Lemma \ref{lem:W-gamma-nrm-finite}, for \(c, \tilde{c}\) sufficiently small
		      \begin{align}\label{eq:Xrho-exp-int-prf}
			       & \mathbb{E}\left[\exp\left(-\frac{c}{2} \sup_{T\geqslant0}\norm{X^{\rho,T}}_{\gamma}^{2} - \frac{c}{2} \sup_{T\geqslant0}\norm{X^{\rho,T}_{T}}_{H^{-\varepsilon}(\chi)}^2 \right)\right] \\
			       & \gtrsim \mathbb{E}\left[\exp\left( -\tilde{ c}\norm{W}_{\gamma}^{2} -\tilde{ c}\sup_{T\geqslant0}\norm{W_{T}}_{H^{-\varepsilon}(\chi)}^{2} \right)\right]
			      > 0,
		      \end{align}
		      and
		      \begin{align}\label{eq:exo-int-prf-2}
			       & \mathbb{E}\left[\exp\left(\frac{c}{2} \sup_{T\geqslant 0 }\norm{X^{\rho, T}}_{\gamma}^{2} + \sup_{T\geqslant0}\frac{c}{2}\norm{X^{\rho, T}_T }_{H^{-\varepsilon}(\chi)}^2 \right)\right] \\
			       & \lesssim \mathbb{E}\left[\exp\left(\tilde{ c}\norm{W}_{\gamma}^{2} + \tilde{c}\sup_{T\geqslant0}\norm{W_{T}}_{H^{-\varepsilon}(\chi)}^{2} \right)\right]
			      < \infty.
		      \end{align}
		\item The first statement is a direct consequence of the same property for \(Z^{\rho, T}\) in Lemma \ref{lem:Zrho-w-p}.
		      Regarding the convergence in \(H^{-\varepsilon}(\chi) \), write
		      \begin{align}
			      \norm{X_{T}^{\rho, T} - X_{\infty}^{\rho}}_{H^{-\varepsilon}(\chi)}
			      \leqslant \norm{Z_{T}^{\rho, T} - Z_{T}^{\rho}}_{H^{-\varepsilon}(\chi)}
			      + \norm{Z_{T}^{\rho} - Z_{\infty}^{\rho} }_{H^{-\varepsilon}(\chi)}
			      + \norm{W_T- W_\infty}_{H^{-\varepsilon}(\chi)}.
		      \end{align}
		      The first two terms actually converge in a stronger norm.
		      For any \(\alpha\in (0, 2\delta_{\kappa})\), the difference \(Z^{\rho, T}_T-Z^{\rho}_T\) is estimated as
		      \begin{align}
			      \norm{Z^{\rho, T}_T-Z^{\rho}_T}_{H^{\alpha}(\chi)}^{2}
			      \lesssim \int_0^T\norm{Q_{s}(F_{s}^{\rho,T}(X_{s}^{\rho,T}) - F_{s}^{\rho}(X_s^{\rho}))}_{L^{2}(\chi)}^{2}\langle s \rangle^{\alpha-1} \mathd s.
		      \end{align}
		      Now, combining Lemma \ref{lem:hk-regularity} with the fact that \(F^{\rho,T}\) is compactly supported thanks to \(\rho\prec1\) and Assumption \ref{hyp:sublinear-est}, we obtain
		      \begin{align}
			      \norm{Q_{s}(F_{s}^{\rho,T}(X_{s}^{\rho,T}) - F_{s}^{\rho}(X_s^{\rho}))}_{L^2(\chi)}
			      \leqslant &
			      \norm{Q_{s}(F_{s}^{\rho,T}(X_{s}^{\rho,T}) - F_{s}^{\rho, T}(X_s^{\rho}))}_{L^2(\chi)}                                                              \\
			                & +
			      \norm{Q_{s}(F_{s}^{\rho,T}(X_{s}^{\rho}) - F_{s}^{\rho}(X_s^{\rho}))}_{L^2(\chi)}                                                                   \\
			      \lesssim  &
			      \langle s \rangle^{-\delta_{\kappa}} (1+ \norm{W}_{\gamma}) \left(\norm{Z^{\rho, T}- Z^{\rho}}_{\gamma} + \langle T \rangle^{-\varepsilon}  \right) \\
			      \lesssim  & \langle s \rangle^{-\delta_{\kappa}} (1+ \norm{W}_{\gamma}) \exp(C(1+\norm{W}_{\gamma})) \langle T \rangle^{-\varepsilon}.
		      \end{align}
		      Here, we also used the convergence of \(Z^{\rho,T} \) from Lemma \ref{lem:Zrho-w-p}.
		      The tail of \(Z^{\rho}\) is controlled by Lemma \ref{lem:Zh-infty-existence}, while the last term, the difference in \(W\), converges by Lemma \ref{lem:W-regularity}.
	\end{enumerate}
\end{proof}

\begin{lemma}\label{lem:H-convergence}
	It holds that
	\begin{align}
		\int_{0}^T \mathcal{H}_s^{\rho, T}(X_s^{\rho, T} )\mathd s
		\xrightarrow{\mathbb{P}}
		\int_{0}^\infty \mathcal{H}_s^{\rho, \infty}(X_s^{\rho, \infty} )\mathd s, \quad (T \to \infty).
	\end{align}
	Moreover, for any \(\varepsilon>0\), there are constants \(C_{\varepsilon, 1}, C_{\varepsilon, 2} \in (0,\infty)\) independent of \(T\) such that almost surely
	\begin{align}\label{eq:Hcal-bounded}
		C_{\varepsilon, 1}\exp(-\varepsilon \sup_{T\geqslant0}\norm{X^{\rho,T} }_{\gamma}^{2})
		\leqslant \exp\left\{-\int_0^T \mathcal{H}_s^{\rho, T}(X_s^{\rho,T})\mathd s\right\}
		\leqslant C_{\varepsilon, 2}\exp(\varepsilon \sup_{T\geqslant0}\norm{X^{\rho, T}}_{\gamma}^{2}),
	\end{align}
	and the family \(\left(\exp(-\int_0^T \mathcal{H}_s^{\rho, T}(X_s^{\rho, T} )\mathd s)\right)_{T\geqslant0}\) is uniformly integrable.
\end{lemma}
\begin{proof}
	\emph{Convergence in probability.}
	By the assumption \eqref{eq:hyp-V,F,H-T-convergence}, there is an \(\varepsilon= \varepsilon(\beta^{2})>0\) such that
	\begin{align}\label{eq:HT-conv}
		          & \abs{\mathcal{H}_{s}^{\rho,T} (X_s^{\rho,T}) - \mathcal{H}_{s}^{\rho} (X_s^{\rho})}                        \\
		\leqslant & \abs{\mathcal{H}_{s}^{\rho,T} (X_s^{\rho,T}) - \mathcal{H}_{s}^{\rho, T} (X_s^{\rho})}
		+ \abs{(\mathcal{H}_{s}^{\rho,T}-\mathcal{H}_{s}^{\rho}) (X_s^{\rho})}                                                 \\
		\lesssim  & \langle s\rangle^{-\delta_{\kappa}\ell^{\ast}} (1+\sup_{T>0}\norm{X^{\rho,T} }_{\gamma})^{2K_F(\beta^{2})}
		\left(\norm{X^{\rho,T}-X^{\rho}}_{\gamma} +  T^{-\varepsilon}  \right)                                                 \\
		\lesssim  & \langle s\rangle^{-\delta_{\kappa}\ell^{\ast}} (1+\norm{W}_{\gamma})^{2K_F(\beta^{2})}
		\left(\norm{X^{\rho,T}-X^{\rho}}_{\gamma} +  T^{-\varepsilon}  \right).
	\end{align}
	Thus, by the convergence of \(X^{\rho,T} \to X^{\rho}\) in Proposition \ref{prop:Xrho-wp} for any \(T>1\) and since \(\langle s\rangle ^{-\delta_{\kappa} \ell^{\ast}} \in L^1(\mathd {s})\), there is an \(\varepsilon'>0\) such that
	\begin{align}\label{eq:HT-L1}
		\mathbb{E} \left[\int_0^T \abs{\mathcal{H}_s^{\rho,T}(X_s^{\rho,T}) - \mathcal{H}_s^{\rho}(X_s^{\rho})} \mathd {s}\right]
		\lesssim T^{-\varepsilon'},
		\qquad
		\mathbb{E} \left[ \int_{T}^{\infty} \abs{\mathcal{H}^{\rho} _s(X^{\rho}_s)} \mathd {s}\right]
		\lesssim T^{1 - \delta_\kappa \ell^{\ast}  } \lesssim T^{-\varepsilon'}.
	\end{align}
	Therefore,
	\begin{align}\label{eq:HT-conv-prob}
		\int_0^T \mathcal{H}_s^{\rho,T}(X_s^{\rho,T}) \mathd {s} \overset{\mathbb{P}}{\longrightarrow} \int_0^\infty \mathcal{H}_s^{\rho}(X_s^{\rho}) \mathd {s}, \quad (T \to \infty).
	\end{align}

	\emph{Uniform bounds.}
	Thanks to Assumption \ref{hyp:sublinear-est}, in particular the bounds \eqref{eq:hyp-V,F,H-uniform-bounds}, we know that, for some \(C\) independent of \(T\) and any \(\varepsilon>0\),
	\begin{align}\label{eq:HT-dominated}
		\exp \left(-\int_0^{T} {\mathcal{H}_{s}^{\rho,T} (X_s^{\rho,T}) }\mathd {s}\right)
		\leqslant & \exp \left(\int_0^{T} \abs{\mathcal{H}^{\rho,T}_s(X_s^{\rho, T} )}\mathd s \right)    \\
		\lesssim  & \exp \left( C \sup_{T\geqslant0}\norm{X^{\rho,T}}_{\gamma}^{2K_F(\beta^{2})}  \right)
		\leqslant C_{\varepsilon} \exp \left( \varepsilon \sup_{T\geqslant0}\norm{X^{\rho,T} }_{\gamma}^{2}  \right),
	\end{align}
	where we used Young's inequality and \(2K_F(\beta^{2})<2\).
	Choosing \(\varepsilon>0\) sufficiently small, integrability of the right-hand side follows from the estimate \eqref{eq:Xrho-exp-int}.
	Similarly, for the lower bound, for any \(\varepsilon>0\)
	\begin{align}\label{eq:HT-exp-non-0}
		\exp \left(-\int_0^T \mathcal{H}_{s}^{\rho, T} (X_s^{\rho, T}) \mathd {s}\right)
		\  & \geqslant C_{\varepsilon }\exp \left(-\varepsilon \sup_{T\geqslant0}\norm{X^{\rho, T} }_{\gamma}^{2}  \right) > 0.
	\end{align}
\end{proof}

\begin{proof}[Proof of Theorem \ref{thm:sG-6/7} assuming Assumption \ref{hyp:sublinear-est}]
	We will show that the scale interpolation from Assumption \ref{hyp:sublinear-est} satisfies the required conditions.
	Existence and uniqueness for \eqref{eq:Xrho-6/7} were shown in Proposition \ref{prop:Xrho-wp}.
	The same proposition also shows convergence of the marginals \(X^{\rho,T}_{T} \to X^{\rho}_{\infty} \in H^{-\varepsilon}(\chi)\) while the full paths converge in the \(\norm{\cdot}_{\gamma}\)-norm almost surely and in every \(L^p(\mathd \mathbb{P})\).
	It remains to show \eqref{eq:sG-Girsanov}.
	We know from Proposition \ref{prop:Girsanov} that the finite \(T\) version of \eqref{eq:sG-Girsanov} holds for every \(T<\infty\).
	By Lemma \ref{lem:H-convergence}, continuity of the exponential, and continuity of the observable \(\mathcal{O}: H^{-\varepsilon}(\chi) \to \mathbb{R}\), we have
	\begin{align}\label{eq:OH-conv-prob}
		\mathcal{O}(X_T^{\rho, T}) \exp \left(-\int_0^T \mathcal{H}^{\rho, T} _s(X_s^{\rho, T} )\mathd s\right)
		\overset{\mathbb{P}}{\longrightarrow}
		\mathcal{O}(X_\infty^{\rho, \infty}) \exp \left(-\int_0^\infty \mathcal{H}^{\rho, \infty} _s(X_s^{\rho, \infty} )\mathd s\right),
	\end{align}
	as \(T \to \infty\).
	To upgrade this convergence to convergence in \(L^1(\mathd \mathbb{P})\), we combine the boundedness of \(\mathcal{O}\)
	with the bounds in \eqref{eq:Hcal-bounded}.
	The bounds in \eqref{eq:Hcal-bounded}, together with the boundedness of \(\mathcal O\) and \eqref{eq:Xrho-exp-int}, give uniform integrability; convergence in probability therefore implies convergence in \(L^1(\mathd\mathbb P)\).
	Finally, choosing \(\mathcal{O} \equiv 1\) and using the lower bound in \eqref{eq:Hcal-bounded} and \eqref{eq:Xrho-exp-int},
	\begin{align}\label{eq:partition-non-zero}
		\mathbb{E}\left[
			\exp\left( -\int_{0}^{\infty} \mathcal{H}_s^{\rho}(X_s^{\rho})\mathd s\right)
			\right]
		\gtrsim
		\mathbb{E}\left[
			\exp\left( -c\sup_{T\geqslant 0}\norm{X^{\rho, T}}_{\gamma}^{2}\right)
			\right]
		>0,
	\end{align}
	ensuring that the normalization in \eqref{eq:sG-Girsanov} is stable as \(T \to \infty\).
\end{proof}
\begin{remark}\label{rem:large-fields}
	The estimates in this section rely heavily on the sublinear control of the effective force.
	Without these estimates, the proof of well-posedness fails already at the level of long-time existence for solutions to the SDE \eqref{eq:Xrho-6/7}, because we rely heavily on sublinear absorption via Young's inequality.
\end{remark}

\subsection{The variational problem}\label{sec:var-problem}
We now turn to the proof of Proposition \ref{prop:fin-vol-var}.
\begin{proof}[Proof of Proposition \ref{prop:fin-vol-var}]
	Fix \(T\in [0,\infty]\) and define for any \(t<s\),
	\begin{align}\label{eq:Rcal-def}
		\mathcal{R}_{t,s}^{\rho, T}(g) \assign \exp\left( -g(X_T^{\rho, T}) - \int_{t}^{s} \mathcal{H}^{\rho, T}_r(X_r^{\rho, T})\mathd r \right).
	\end{align}
	A direct computation using \eqref{eq:Hcal-bounded} and \eqref{eq:Xrho-exp-int} shows that almost surely \(\mathbb{E}_t \left[ \mathcal{R}_{t,T}^{\rho, T} (g)\right] \in (0,\infty)\).
	Let
	\begin{align}\label{eq:qBSDE-solution}
		\bar{Y}_t^{\rho, T}(g) & = - \log  \left(\mathbb{E}_t\left[ \mathcal{R}^{\rho,T}_{t,T}(g)\right] \right),
	\end{align}
	and define the martingale associated to \eqref{eq:def-M-density} as
	\begin{align}\label{eq:def-MtT}
		M_t^{\rho, T} (g)
		= M_{0,t}^{\rho, T} (g) \assign
		\frac{\mathbb{E}_t\left[\mathcal{R}_{0,T}^{\rho,T} (g)\right]}
		{\mathbb{E}\left[\mathcal{R}_{0,T}^{\rho, T}(g) \right]}
		= \exp\left(-\int_0^t \mathcal{H}^{\rho, T} _s(X^{\rho, T} _s) \mathd s - (\bar{Y}_t^{\rho, T}(g) - \bar{Y}_0^{\rho, T}(g)) \right).
	\end{align}
	First, by the strong uniqueness for \eqref{eq:Xrho-6/7}, the filtrations generated by \(X^{\rho,T}\) and \(W\) are equivalent.
	Applying the martingale representation theorem to the square-integrable martingale \(M_t^{\rho,T}(g)\),
	there is a predictable process \(\bar{U}_{s}^{\rho, T}(g) \) such that
	\begin{align}\label{eq:def-Ubar}
		M_{t}^{\rho,T}(g) = \mathcal{E}_{0,t}(-\bar{U}^{\rho, T}(g)).
	\end{align}
	Moreover, by Itô's formula
	\begin{align}
		\bar{Y}_t^{\rho,T}(g) = g(X_{T}^{\rho,T}) + \int_t^T \mathcal{H}^{\rho,T}_s(X_s^{\rho, T}) \mathd s - \frac{1}{2} \int_t^T \norm{Q_s \bar{U}^{\rho, T}_s(g)}_{L^{2}}^{2} \mathd s - \int_t^T \bar{U}_{s}^{\rho, T}(g) \mathd W_s.
	\end{align}
	We will now prove that \(-\bar{U}^{\rho,T}\in \mathcal{A}\) with
	\begin{align}\label{eq:Wg-Y0}
		\mathcal{W}^{\rho,T} (g) = \mathcal{J}^{\rho, T}(g;-\bar{U}^{\rho, T}(g)) = \bar{Y}^{\rho,T} _{0}(g).
	\end{align}
	For \(T=\infty\), the claim will then follow from
	\begin{align}
		\bar{Y}_{0}^{\rho, T}(g) - \bar{Y}_{0}^{\rho, T}(0)
		= & -\log\frac{
			\mathbb{E}\left[ \mathe^{-g(X_{T}^{\rho,T}) - \int_0^T \mathcal{H}^{\rho, T} _{s}(X_s^{\rho, T})\mathd s } \right]
		}{
			\mathbb{E}\left[ \mathe^{-\int_0^T \mathcal{H}_s^{\rho,T}(X_s^{\rho,T})\mathd s} \right]
		}                                                                                    \\
		= & - \log \mathbb{E} \left[ \mathe^{-g(X_{T}^{\rho,T})} M_{0,T}^{\rho,T}(0) \right] \\
		= & - \log \Lambda_{\tmop{sG}}^{\rho, T}(g).
	\end{align}
	Since \(\mathcal{E}(-\bar{U}^{\rho}(0))=M_{0,\infty}^{\rho}(0)\), combining this identity with Lemma \ref{lem:g-measure-determining} below also gives \(\tmop{Law}^{\bar{\mathbb{P}}}(X_\infty^{\rho}) = \nu_{\tmop{sG}}^{\rho}\).

	\emph{Finite cost.}
	Uniform integrability of \(\mathcal{E}(-\bar{U}^{\rho,T}(g))\) follows directly from the fact that \(M_{t}^{\rho, T}(g)\) is closed.
	Denote the tilted probability measure by \(\bar{\mathbb{P}}= \mathbb{P}^{-\bar{U}^{\rho, T}(g)}\) and denote expectations under this measure by \(\bar{\mathbb{E}}\).
	By Lemma \ref{lem:fin-cost-entropy} below, it is therefore sufficient to show that
	\begin{align}\label{eq:Ubar-fin-cost}
		I := \bar{\mathbb{E}} \left[\int_0^T \norm{Q_s\bar{U}_s^{\rho,T}}_{L^{2}}^{2}\mathd s \right] < \infty.
	\end{align}
	We compute using that \(M_{T}^{\rho, T}(g)\in L^{2}(\mathd \mathbb{P}) \) by the boundedness of \(g\) and Lemma \ref{lem:H-convergence} combined with Proposition \ref{prop:Xrho-wp},
	\begin{align}
		\frac{1}{2} I = \mathbb{E} \left[\log ( \frac{\mathd \bar{\mathbb{P}}}{\mathd \mathbb{P}}) \frac{\mathd \bar{\mathbb{P}}}{\mathd \mathbb{P}}\right]
		=  \mathbb{E} \left[\log ( M_{T}^{\rho, T}(g)) M_{T}^{\rho, T}(g)\right] < \infty.
	\end{align}
	Using that \(\bar{W}= W^{-\bar{U}^{\rho, T}(g)}_t = W_t + \int_0^t\dot{G}_s \bar{U}^{\rho, T}_s \mathd s  \) is a \(\bar{\mathbb{P}}\)-martingale, we compute
	\begin{align}
		\bar{Y}^{\rho, T}_{0}
		= & g(X_{T}^{\rho, T}) + \int_0^T \mathcal{H}_s^{\rho, T}(X_s^{\rho, T})\mathd s
		+ \frac{1}{2}\int_0^T \norm{Q_s \bar{U}^{\rho, T}_s (g)}_{L^{2}}^{2} \mathd s
		- \int_0^T \bar{U}^{\rho, T}_s \mathd \bar{W}_s                                                         \\
		= & \bar{\mathbb{E}}\left[ g(X_{T}^{\rho, T}) + \int_0^T \mathcal{H}_s^{\rho, T}(X_s^{\rho, T})\mathd s
		+ \frac{1}{2}\int_0^T \norm{Q_s \bar{U}^{\rho, T}_s (g)}_{L^{2}}^{2} \mathd s  \right]                  \\
		= & \mathcal{J}^{\rho, T}(g; -\bar{U}^{\rho, T}).
	\end{align}
	In the last step we used that the stochastic integral is a martingale by \eqref{eq:Ubar-fin-cost}.

	\emph{Optimality.}
	Let \(u\in \mathcal{A}\) with \(\mathcal{J}^{\rho, T}(g; u) < \infty\).
	By definition of \(\mathcal{J}^{\rho, T}(g; u) \), the associated stochastic exponential \(\mathcal{E}(u)\) exists and, moreover, by Lemma \ref{lem:fin-cost-entropy}
	\begin{align}
		\mathbb{H}(\mathbb{P}^u \vert \mathbb{P}) = \mathbb{E}^u\left[\log(\mathcal{E}_{0,T}(u))\right]
		= \frac{1}{2} \mathbb{E}^{u} \left[\int_0^T \norm{Q_s u_s}_{L^{2}}^{2}\mathd s\right] < \infty.
	\end{align}
	The estimates on \(\mathcal{H}\) and the boundedness assumptions of \(g\) ensure that \(\bar{\mathbb{P}}\sim \mathbb{P} \).
	As a result, \(\mathbb{P}^u \ll \bar{\mathbb{P}}\) and by the chain rule for the Radon--Nikodym derivatives \(\frac{\mathd \mathbb{P}^u}{ \mathd \bar{\mathbb{P}}} = \frac{\mathcal{E}_{0,T}(u)}{M_{0,T}^{\rho, T}(g)}\).
	We compute using that \(\bar{Y}_{T}^{\rho, T}(g)=g(X_{T}^{\rho, T})\)
	\begin{align}
		\mathcal{J}^{\rho, T}(g; u) - \bar{Y}_{0}^{\rho, T} (g)
		= & \mathbb{E}^u \left[ g(X_{T}^{\rho, T})  + \int_0^T \mathcal{H}^{\rho, T}_{s}(X_s^{\rho, T})\mathd s   - \bar{Y}_{0}^{\rho, T}(g) \right] + \mathbb{E}^u \left[\log\left(\mathcal{E}_{0,T}(u)\right)\right]                  \\
		= & \mathbb{E}^u \left[ \int_0^T \mathcal{H}^{\rho, T}_{s}(X_s^{\rho, T})\mathd s  +\left( \bar{Y}_{T}^{\rho, T}(g) - \bar{Y}_{0}^{\rho, T}(g)\right) \right] + \mathbb{E}^u \left[\log\left(\mathcal{E}_{0,T}(u)\right)\right] \\
		= & \mathbb{E}^u \left[\log\left( \frac{\mathcal{E}_{0,T}(u)}{M_{0,T}^{\rho, T}(g)}\right)\right]                                                                                                                               \\
		= & \mathbb{H} (\mathbb{P}^u \vert \bar{\mathbb{P}})
		\geqslant 0,
	\end{align}
	with equality if and only if \(\mathbb{H} (\mathbb{P}^u \vert \bar{\mathbb{P}}) = 0 \iff \mathbb{P}^u = \bar{\mathbb{P}}\).
	This proves that \(\bar{u}^g_s \assign -\bar{U}_s^{\rho, T}(g)\) is optimal.

	\emph{Uniqueness.}
	Assume that \(u\) is optimal, so that \(\mathbb{P}^u=\bar{\mathbb{P}}\), as above.
	Since \(\mathcal{E}(u)\) and \(\mathcal{E}(\bar{u}^g)\) are uniformly integrable with coinciding terminal values, we know that
	\begin{align}
		\mathcal{E}_t  \assign \mathcal{E}_t(u) = \mathbb{E}_t \left[ \mathcal{E}_{T}(u)\right] = \mathbb{E}_t \left[ \mathcal{E}_{T}(\bar{u}^g)\right] = \mathcal{E}_t(\bar{u}^g).
	\end{align}
	We compute using Itô's formula for the stochastic exponential,
	\begin{align}
		0 = \mathcal{E}_t(u) - \mathcal{E}_t(\bar{u}^g) = \int_0^t \mathcal{E}_s (u_s - \bar{u}_s^g) \mathd W_s = N_t.
	\end{align}
	In particular, \(\langle N \rangle_t = 0 \) and therefore
	\begin{align}
		0 = \langle N \rangle_t = \int_0^t \mathcal{E}_s^2 \norm{Q_s (u_s-\bar{u}^g_s)}^{2}_{L^{2}} \mathd s,\qquad t\in[0,T].
	\end{align}
	Since \(\mathcal{E}_t > 0\) and \(\norm{Q_s (u_s-\bar{u}^g_s)}^{2}_{L^{2}}\geqslant0\), this implies
	\begin{align}
		\norm{Q_s (u_s-\bar{u}^g_s)}^{2}_{L^{2}} = 0, \qquad \mathd \mathbb{P} \otimes \mathd s\text{-almost surely,}
	\end{align}
	and thus \(Q_s(u_s-\bar u_s^g)=0\), \(\mathd\mathbb P\otimes\mathd s\)-almost surely.
\end{proof}

It remains to prove the following statements.
\begin{lemma}\label{lem:fin-cost-entropy}
	Let \(u\in \mathcal{A}\) such that \(\mathcal{E}(u)\) is a uniformly integrable martingale. Then
	\begin{align}
		\mathcal{J}^{\rho,T}(g; u)<\infty \iff \mathbb{E}^{u}\left[ \int_0^T \norm{Q_s u_s}^{2}_{L^{2}}\mathd s\right] < \infty.
	\end{align}
\end{lemma}
\begin{proof}
	Suppose that \( \mathbb{E}^{u}\left[ \frac{1}{2}\int_0^T \norm{Q_s u_s}^{2}_{L^{2}} \mathd s\right] < \infty\).
	Then, by the boundedness of \(g\) and the estimates on \(\mathcal{H}\),
	\begin{align}
		\abs{g(X_{T}^{\rho, T}) + \int_0^T \mathcal{H}^{\rho, T}_s(X_{s}^{\rho, T})\mathd s}
		\lesssim 1 + \norm{X^{\rho, T}}_{\gamma}^{2K_F}
		\lesssim 1 + \norm{W}_{\gamma}^{2K_F}.
	\end{align}
	Now note that for any \(\abs{\alpha}< 1\) and \(t\geqslant0 \),
	\begin{align}
		\norm{\nabla^\alpha \int_0^t \dot{G}_s u_s \mathd s}_{L^\infty(\chi)}
		\leqslant \int_0^t \norm{\nabla^\alpha Q_s}_{L^{2}} \norm{Q_s u_s}_{L^{2}} \mathd s
		\leqslant \left(\int_0^t \norm{\nabla^\alpha Q_s}_{L^{2}}^{2}\mathd s\right)^{\frac{1}{2}}
		\left(\int_0^t \norm{Q_s u_s}^{2}_{L^{2}} \mathd s\right)^{\frac{1}{2}}.
	\end{align}
	By the heat-kernel estimates for \(Q_s\) in Lemma \ref{lem:hk-Lp-bounds}, for any \(\gamma>0\),
	\begin{align}
		\sup_{t\geqslant0} \langle t \rangle^{-\frac{\abs{\alpha}}{2} - \gamma}   \left(\int_0^t \norm{\nabla^\alpha Q_s}_{L^{2}}^{2}\mathd s\right)^{\frac{1}{2}}
		< \infty.
	\end{align}
	Thus,
	\begin{align}\label{eq:W-leq-Wu}
		\norm{W}_{\gamma}
		\leqslant \norm{W^u}_{\gamma} +  \norm{\int_0^\cdot \dot{G}_s u_s \mathd s}_{\gamma}
		\lesssim  \norm{W^u}_{\gamma} + \left(\int_0^T\norm{Q_su_s}^{2}_{L^{2}} \mathd s \right)^{\frac{1}{2}}.
	\end{align}
	Since \(\mathbb{E}^u{\norm{W^u}_{\gamma}^2}<\infty\), combining these estimates and using \(2K_F<2\), we find
	\begin{align}
		\abs{\mathcal{J}^{\rho, T}(g; u)}
		\leqslant &
		C\left(1 + \mathbb{E}^u \left[\norm{W}_{\gamma}^{2K_F} \right]\right) + \mathbb{E}^{u}\left[ \int_0^T \norm{Q_s u_s}^{2}_{L^{2}}\mathd s\right] \\
		\lesssim  &
		1 + \mathbb{E}^{u} \left[\norm{W^u}_{\gamma}^{2K_F} \right] + \mathbb{E}^{u}\left[ \int_0^T \norm{Q_s u_s}^{2}_{L^{2}}\mathd s\right]
		< \infty.
	\end{align}

	Conversely, suppose now that \(\mathcal{J}^{\rho, T}(g; u) <\infty\).
	We will prove the pathwise estimate
	\begin{align}\label{eq:J-coercive}
		          & g(X_{T}^{\rho, T}) + \int_0^T \mathcal{H}_s^{\rho, T} (X_s^{\rho, T}) \mathd s + \frac{1}{2}\int_0^T \norm{Q_s u_s}^{2}_{L^{2}} \mathd s \\
		\geqslant & \frac{1}{4}  \int_0^T \norm{Q_s u_s}^{2}_{L^{2}} \mathd s
		- C(1 + \norm{W^u}_{\gamma}^{2K_F} ),
	\end{align}
	where all quantities are almost surely finite.
	This directly implies the claim, since taking expectations in \eqref{eq:J-coercive} yields
	\begin{align}
		\frac{1}{4}  \mathbb{E}^u\left[\int_0^T \norm{Q_s u_s}^{2}_{L^{2}}\mathd s\right]
		\leqslant \mathcal{J}^{\rho, T}(g; u) + C\left(1+ \mathbb{E}^{u}\left[\norm{W^u}_{\gamma}^{2K_F}\right]\right) <\infty.
	\end{align}
	It remains to prove \eqref{eq:J-coercive}.
	To this end, note first that \(\int_0^T \norm{Q_s u_s}^{2}_{L^{2}}\mathd s< \infty \) holds \(\mathbb{P}^u\)-almost surely since \(\mathcal{E}(u)\) is uniformly integrable.
	From \eqref{eq:W-leq-Wu} and the estimates on \(g, \mathcal{H}^{\rho, T} \) already used in the first part of this proof,
	\begin{align}
		\abs{g(X_{T}^{\rho, T}) + \int_0^T \mathcal{H}^{\rho, T}_s (X_s^{\rho, T}) \mathd s }
		\lesssim & 1 + \norm{X^{\rho, T}}_{\gamma}^{2K_F}                                                             \\
		\lesssim & 1 + \norm{W^u}_{\gamma}^{2K_F} + \left( \int_0^T \norm{Q_s u_s}^{2}_{L^{2}} \mathd s\right)^{K_F}.
	\end{align}
	By Young's inequality using again \(K_F < 1\), this implies
	\begin{align}
		g(X_{T}^{\rho, T}) + \int_0^T \mathcal{H}^{\rho, T}_s (X_s^{\rho, T}) \mathd s
		\geqslant -\frac{1}{4} \left( \int_0^T \norm{Q_s u_s}^{2}_{L^{2}} \mathd s\right) -  C (1 + \norm{W^u}_{\gamma}^{2K_F}),
	\end{align}
	which is \eqref{eq:J-coercive} once we add \(\frac{1}{2}\int_0^T \norm{Q_s u_s}^{2}_{L^{2}}\mathd s \) to both sides.
\end{proof}

\begin{lemma}\label{lem:g-measure-determining}
	Denote the class of functions \(g\) satisfying Assumption \ref{hyp:g-allowed} by \(\mathcal{G}\).
	The class \(\{\mathe^{-g}\}\) is measure-determining on \(H^{-\varepsilon}(\chi)\).
\end{lemma}
\begin{proof}
	Let \(\nu^1, \nu^2\) be two measures on \(H^{-\varepsilon}(\chi) \) such that \(\nu^1(\mathe^{-g}) = \nu^2 (\mathe^{-g})\) for all \(g\in \mathcal{G}\).
	We will show that then also
	\begin{align}\label{eq:g-measure-det}
		\nu^1(-g) = \nu^2(-g) \qquad \text{for all \(g\in \mathcal{G}\) }.
	\end{align}
	Note that \(\mathcal{G}\) contains the cylindrical functionals \(g(\varphi) = f(\langle \varphi, h_1 \rangle,\dots, \langle \varphi, h_n \rangle  )\) where \(f\in C^1_b(\mathbb{R}^{n}) \) and \(h_1,\dots,h_n\in C^{\infty}_c(\mathbb{R}^{2})\) for \(n\in \mathbb{N}\).
	Since \(H^{-\varepsilon}(\chi) \) is separable, these functionals form a measure-determining class and the claim will follow from \eqref{eq:g-measure-det}.
	To see \eqref{eq:g-measure-det}, note that for \(g\in \mathcal{G}\), also \(t g\in \mathcal{G}\) for every \(t\in \mathbb{R}\).
	Moreover, thanks to the boundedness of \(g\), \(\lim_{t \to 0} t^{-1}\left(\mathe^{-t g}-1\right)=-g\) uniformly.
	Hence,
	\begin{align}
		\nu^1(-g) = \lim_{t\to 0} t^{-1}\left(\nu^1(\mathe^{-tg}) -1\right) = \lim_{t\to 0} t^{-1}\left(\nu^2(\mathe^{-tg}) -1\right) =  \nu^2(-g).
	\end{align}
\end{proof}

\section{The flow equation for the potential in the full subcritical regime}\label{sec:Flow}

In this section, we study approximate solutions to the renormalization flow equation
\begin{align}\label{eq:V-ell-full-flow}
	\partial_t V_{t} + \frac{1}{2} \Delta_{\dot{G}_t} V_{t} = \frac{1}{2} \langle \nabla V_t, \dot{G}_t \nabla V_t \rangle_{L^2},
\end{align}
subject to the (formal) terminal condition \eqref{eq:int-V-cos}
and prove Theorem \ref{thm:int-Va-estimates} (see Proposition \ref{prop:Va-general-estimates}) as well as the existence of a scale interpolation for \((V_t^{\rho,T})_{t\in [0,T]}\) satisfying Assumption \ref{hyp:sublinear-est}.

More precisely, we show that for every \(\beta^{2}<8\pi\) and \(\rho\prec1\), there is a choice of counterterms as defined in \eqref{eq:Vrho-T} and a consistent family of scale interpolations \((V_{t}^{\rho, T})_{t\in [0,T]}\),
such that the interpolation \(V_{t}^{ \rho, T}\) can be estimated uniformly in \(T\geqslant0\).
Before we state the result more carefully, recall that \(\delta=1-\frac{\beta^{2}}{8\pi}\) denotes the distance to the critical value \(\beta^{2}=8\pi\) of the sine-Gordon model in our normalization.
Recall also that we want to find \(V_{t}^{\rho,T}\)  such that the remainder
\begin{align}\label{eq:flow-def-Hcal}
	   & \mathcal{H}_{t}^{\rho, T}[\varphi]
	= \mathcal{H}_{t}^{V^{\rho,T} }[\varphi]  \\
	:= & \partial_{t}V_{t}^{\rho, T}[\varphi]
	+ \frac{1}{2} \Delta_{\dot{G}_{t}} V_{t}^{\rho, T}[\varphi]
	- \frac{1}{2} \langle\nabla  V_{t}^{\rho,T}[\varphi], \dot{G}_{t} \nabla V_{t}^{\rho,T}[\varphi]\rangle_{L^{2}},
\end{align}
as defined in \eqref{eq:int-def-Hcal} is integrable in the UV.

\subsection{Strategy for the proof of Theorem \ref{thm:int-Va-estimates}}\label{sec:Va-outline}
Let us fix \(T\in (0,\infty)\) and \(\rho\prec1\) and suppress the dependence on \(T>0, \rho\prec1\) whenever possible to simplify the notation.
Theorem \ref{thm:int-Va-estimates} requires a systematic procedure for the renormalization flow equation that can generate arbitrarily good approximations to the flow equation.
The starting point for this is the following Picard iteration of the full renormalization flow equation \eqref{eq:V-ell-full-flow}.
If we formally expand \(V_{t}\) as a power series in the coupling constant \(\lambda\)
\begin{align}\label{eq:V-perturbative-expansion}
	{V}_{t}[\varphi] \assign \sum_{\ell\in \mathbb{N}} V_{t}^{[\ell]}[\varphi],\qquad V_t^{[\ell]} \in \mathcal{O}(\lambda^{\ell}),
\end{align}
and insert this ansatz in \eqref{eq:V-ell-full-flow},
then, matching the coefficients of the formal power series, we obtain equations for each order of the perturbative expansion \eqref{eq:V-perturbative-expansion},
\begin{align}\label{eq:V-ell}
	(\partial_{t}+\frac{1}{2}\Delta_{\dot{G}_{t}})V_t^{[\ell]}[\varphi]
	= \frac{1}{2}\sum_{\ell'+\ell''=\ell} \langle \nabla V_{t}^{[\ell']}[\varphi],\dot{G}_{t}\nabla V_{t}^{[\ell'']}[\varphi]\rangle_{L^{2}}, \quad \ell > 1.
\end{align}
It is not difficult to check that if we start the induction from
\begin{align}\label{eq:V-initial}
	V_{t}^{\rho, T, [1]}[\varphi] = \int_{\mathbb{R}^{2}} \lambda_t \cos (\beta \varphi(x))\rho(\mathd {x}),
\end{align}
subject to the terminal condition
\begin{align}
	V_{T}^{\rho , T, [\ell]} [\varphi] = c^{\rho, T, [\ell]},\qquad \ell>1,
\end{align}
and if \eqref{eq:V-perturbative-expansion} converges and \eqref{eq:V-ell} is satisfied at each order, then \eqref{eq:V-perturbative-expansion} is a solution to \eqref{eq:V-ell-full-flow}.
We are, however, not concerned with questions of convergence of \eqref{eq:V-perturbative-expansion}:
Since our ansatz requires only sufficiently good approximate solutions to \eqref{eq:V-ell-full-flow}, we will always work with truncations
\begin{align}\label{eq:V-truncated}
	V_{t} [\varphi] = V_{t}^{[< \ell^{\ast} ]} [\varphi] \assign \sum_{\ell< \ell^{\ast} } V_{t}^{[\ell]}[\varphi], \ell^\ast\in \mathbb{N}.
\end{align}
With this ansatz, we can verify that \(\mathcal{H}\) as defined in \eqref{eq:flow-def-Hcal} is given by
\begin{align}\label{eq:H-truncated}
	\mathcal{H}_{t}^{[<\ell^{\ast}]}[\varphi]
	\assign - \frac{1}{2} \sum_{{\substack{\ell',\ell''< \ell^{\ast} \\ \ell'+\ell'' \geqslant \ell^{\ast} }}} \langle \nabla  V_{t}^{[\ell' ]}[\varphi], \dot{G}_{t} \nabla  V_{t}^{[\ell'']}[\varphi] \rangle_{L^{2}}.
\end{align}
We expect that this remainder \(\mathcal{H}\) is small not only in terms of \(\bar{\lambda}\), but also in the UV:
A heuristic argument (see \cite[Section 3.1]{gubinelliFBDSEApproachSine2026} for a more detailed discussion of this) suggests that
\begin{align}\label{eq:V,H-scaling}
	\abs{V_{t}^{[\ell]}[\varphi]}  \sim \langle t  \rangle ^{1-\ell \delta}, \qquad \abs{\mathcal{H}_{t}^{[<\ell^{\ast} ]}[\varphi] }  \sim \langle t  \rangle ^{-\ell^{\ast}\delta}.
\end{align}
In the subcritical regime \(\beta^{2}<8\pi\), we therefore expect that choosing \(\ell^{\ast}>\frac{1}{\delta}\) is sufficient to ensure integrability of \(\mathcal{H}\) in the UV.
Since we are not studying questions of convergence at this point, we drop all the combinatorial constants, but they could be tracked with some additional notational effort.

To simplify \eqref{eq:V-ell}, it is convenient to pass to a generalized Fourier transform, also known as the Mayer expansion.
It is given by
\begin{align}\label{eq:V-Mayer-exp}
	V_{t}^{[\ell], \rho, T}[\varphi]
	= & \int \rho(\mathd {\xi_{1:\ell}}) \tilde{v}_{t}^{[\ell], \rho, T}(\xi_{1:\ell}) \psi(\xi_{1:\ell})                                                     \\
	= & \int \rho(\mathd {\xi_{1:\ell}}) \tilde{v}_{t}^{[\ell], \rho, T}(\xi_{1:\ell})\frac{1}{2}\left( \psi(\xi_{1:\ell}) + \psi(\bar{\xi}_{1:\ell })\right) \\
	= & \int \rho(\mathd {\xi_{1:\ell}}) \tilde{v}_{t}^{[\ell], \rho, T}(\xi_{1:\ell}) \cos\left(\beta\sum_{j=1}^{\ell} \sigma_j\varphi(x_j)\right),
\end{align}
where
\begin{itemize}
	\item \(\xi=(\sigma, x)\in \{\pm1\}\times \tmop{supp(\rho)}=:\mathcal{X}\) is a {point charge} and \(\xi_{1:n}= (\xi_{1},\dots, \xi_{n})\in \mathcal{X}^{n}\), and \(\bar{\xi} \assign (-\sigma, x)\),
	\item the shorthand integral notation for the volume cut-off \(\rho\prec 1\),
	      \begin{align}\label{eq:def-rho-integral}
		      \int \rho(\mathd {\xi}_{1:n})
		      := \sum_{\sigma_{1:n}\in \{\pm1\}^{n}} \int_{\mathbb{R}^{2n}} \prod_{k=1}^{n}\rho(\mathd {x}_{k})
		      = \sum_{\sigma_{1:n}\in \{\pm1\}^{n}} \int_{\mathbb{R}^{2n}} \prod_{k=1}^{n}\rho({x}_{k}) \mathd {x}_{k},
	      \end{align}
	\item the vertex fields \(\psi(\xi)= \exp(i \beta \sigma \varphi(x))\), and \(\psi(\xi_{1:n})=\exp(i \beta \varphi(\xi_{1:n}))\),  where we used the shorthand
	      \begin{align}\label{eq:def-S-xi}
		      \varphi(\xi_{1:n}) := S_{\xi_{1:n}}(\varphi) := \sum_{k=1}^{n} \sigma_{k} \varphi(x_k).
	      \end{align}
	      For a neutral configuration \(\xi_{1:n}\), define
	      \begin{align}\label{eq:def-delta-cluster}
		      \delta(\xi_{1:n};y) = \sum_{j=1}^{n}\abs{x_j-y},
		      \qquad
		      \delta(\xi_{1:n}) = \min_{k\in[n]}\delta(\xi_{1:n};x_k).
	      \end{align}
	      For \(h\in C^\infty(\mathbb{R}^{2})\), we also define the shorthand
	      \begin{align}\label{eq:S-frac}
		      S_{\xi_{1:n}}^{\nu} (h) \assign \frac{S_{\xi_{1:n}}(h)}{(\delta(\xi_{1:n}))^\nu}.
	      \end{align}

	\item \(\tilde v^{[\ell], \rho, T}_{t}: \mathcal{X}^{\ell}\to \mathbb{R}\) are suitable kernels that are invariant under translations, permutations of the point charges \(\xi\), charge conjugation \(\xi_{1:\ell}\mapsto \bar{\xi}_{1:\ell}\), and have finite norm
	      \begin{align}\label{eq:v-ell-norm}
		      \tnorm{\tilde{v}_{t}^{[\ell]}}_{t}
		       & \assign \sup_{n\in [\ell]} \sup_{(\sigma_n, x_n)\in \{\pm 1\}\times\mathbb{R}^{2}}
		      \int \mathd \xi_{[\ell]\setminus\{n\}} \omega^{\ell}_{t} (\xi_{1:\ell})\abs{\tilde{v}_{t}^{[\ell]}(\xi_{[\ell]})} \\
		       & =  \sup_{x_1\in \mathbb{R}^{2}}
		      \int \mathd \xi_{2:\ell} \omega^{\ell}_{t} (\xi_{1:\ell})\abs{\tilde{v}_{t}^{[\ell]}(\xi_{[\ell]})},
	      \end{align}
	      where for some \(r(\ell)>0\), the weight \(\omega_t^{\ell}(\xi_{1:\ell}) = \exp( r(\ell) \sqrt{t}\tmop{St}(\xi_{1:\ell}) )\) quantifies the exponential concentration near the diagonal (see \eqref{eq:def-Steiner-weight} for the definition of \( \tmop{St}\), and \eqref{eq:def-r-steiner-rate} for the concrete definition of the rate \(r\) that is used later).
	      If \(\omega\equiv 1\) or if \(\ell\) is clear from the context, we may also drop the dependence on \(t\) or \(\ell\) in the subscript.
\end{itemize}
Now fix \(T<\infty\) and \(\rho\prec1\) for the rest of this section and suppress the dependence of \(V,  F, \mathcal{H}, H\) on these parameters in the notation.
We also define the \emph{charge} of a configuration \(\xi_{1:\ell}\) as
\begin{align}\label{eq:v2-charge}
	q(\xi_{1:\ell}) = \sum_{k=1}^{\ell} \sigma_k\in \{-\ell,\dots,\ell\}.
\end{align}
We say that a configuration \(\xi_{1:\ell}\) is \emph{neutral} if \(q(\xi_{1:\ell}) = 0\) and call it \emph{charged} otherwise.
It is often convenient to separate the kernels according to their charge by defining
\begin{align}\label{eq:v2-def-v-q}
	\tilde{v}_{t}^{[\ell](q)}(\xi_{1:\ell}) = \mathbb{1}_{q(\xi_{1:\ell})=q} \tilde{v}_{t}^{[\ell]}(\xi_{1:\ell}),  \qquad
	\tilde{v}_{t}^{[\ell]} = \sum_{q\in \mathbb{Z}} \tilde{v}_{t}^{[\ell](q)}.
\end{align}

In terms of the kernels \(v\), the initial condition \eqref{eq:V-initial} implies
\begin{align}\label{eq:v1-explicit}
	\tilde{v}_{t}^{[1]} (\xi_{1}) = \frac{\lambda_t}{2}.
\end{align}
In particular, this also implies that
\begin{align}\label{eq:v1-bound}
	\tnorm{\tilde{v}^{[1]}_t} = \abs{\tilde{v}^{[1]}_t} = \frac{\bar{\lambda}_t}{2} \lesssim \bar\lambda\langle t  \rangle ^{1-\delta},
\end{align}
which is compatible with the heuristic bound \eqref{eq:V,H-scaling}.
Instead of estimating \(\tilde{V}_{t}^{[\ell]}\) directly, we want to control the finite-dimensional kernels \(\tilde{v}_{t}^{[\ell ]}\).

A simple computation using the fact that the functions \(\psi(\xi_{1:n})\) are eigenfunctions of the Laplacian (see \cite[Lemma 3.3]{gubinelliFBDSEApproachSine2026}) verifies that we can now replace the flow equation \eqref{eq:V-ell} by the following flow equation on the level of the coefficients
\begin{align}\label{eq:v-ell}
	\partial_t \tilde{v}_{t}^{[\ell]}(\xi_{1:\ell}) + \dot{\mathbb{G}}_{t}(\xi_{1:\ell}) \tilde{v}_{t}^{[\ell]}(\xi_{1:\ell})
	= \frac{1}{2}\sum_{I'\dot{\cup} I''=[\ell]} \tilde{v}_{t}^{[\abs{I'}]}(\xi_{I'}) \dot{\mathbb{G}}_{t}(\xi_{I'}, \xi_{I''}) \tilde{v}_{t}^{[\abs{I ''}]}(\xi_{I''}),
\end{align}
where \(\dot{\cup}\) denotes the disjoint union, we use the convention \(v^{[0]} = 0\) and define the shorthand
\begin{align}\label{eq:def-Gbb}
	\dot{\mathbb{G}}_t(\xi_J,\xi_K)
	 & := -\beta^2 \sum_{k \in K} \sum_{j \in J} \sigma_k \sigma_j \, \dot{G}_t(x_k - x_j), \\
	\dot{\mathbb{G}}_t(\xi_{1:n})
	 & := \frac{1}{2}\,\dot{\mathbb{G}}_t(\xi_{1:n},\xi_{1:n}),                             \\
	\mathbb{G}_t(\xi_J,\xi_K)
	 & := \int_0^t \dot{\mathbb{G}}_s(\xi_J,\xi_K)\,\mathd s,                               \\
	\mathbb{G}_{t}(\xi_{1:n})
	 & := \frac{1}{2}\mathbb{G}_{t}(\xi_{1:n}, \xi_{1:n})                                   \\
	\mathbb{G}_{s,t}(\cdot,\cdot)
	 & := \mathbb{G}_s(\cdot,\cdot) - \mathbb{G}_t(\cdot,\cdot)
	= \int_t^s \dot{\mathbb{G}}_u(\cdot,\cdot)\,\mathd u.
\end{align}
The functions \(\mathbb{G}\) can also be written in terms of the sums \(S_{\xi_{J}}\) as defined in \eqref{eq:def-S-xi}.
Notice that the system \eqref{eq:v-ell} is triangular in \(\ell \in \mathbb{N}\):
The bilinear term
\begin{align}\label{eq:B-ell-Mayer}
	B^{I', I''}_{[\ell]}(\tilde{v}_{s}^{[\ell']}, \tilde{v}_{s}^{[\ell'']}, \dot{G}_{s})
	\assign
	\mathbb{1}_{\{\abs{I'}= \ell'\}}\mathbb{1}_{\{\abs{I''}= \ell''\}}\mathbb{1}_{\{I'\dot{\cup} I''= [\ell]\}} \tilde{v}_{s}^{[\abs{I'}]}(\xi_{I'}) \dot{\mathbb{G}}_{s}(\xi_{I'}, \xi_{I''}) \tilde{v}_{s}^{[\abs{I''}]}(\xi_{I''}),
\end{align}
in \eqref{eq:v-ell} is non-zero only when \(\ell', \ell'' < \ell \).
The linear part of the equation can be solved exactly using the propagator of the linear equation
\begin{align}\label{eq:def-Mayer-Gamma}
	\Gamma_{t,s}(\xi_{1:\ell})= \exp(\mathbb{G}_{s,t}(\xi_{1:\ell})),
\end{align}
as an integrating factor.
Note that no operation changes the configuration \(\xi_{1:\ell}\), so that the charge is preserved along the flow.
Combining these observations, we obtain the triangular, integrated version of \eqref{eq:v-ell} for every \(\ell >1\),
\begin{align}\label{eq:v-ell-integrated}
	\tilde{v}_{t}^{[\ell](q)}(\xi_{1:\ell})
	= & \Gamma_{t,T}(\xi_{1:\ell})\tilde{v}_{T}^{[\ell](q)}(\xi_{1:\ell})                                                                                                                   \\
	- & \sum_{I'\dot\cup I''=[\ell]}\int_t^{T} \Gamma_{t,s}(\xi_{1:\ell })  B^{I', I''}_{[\ell]}(\tilde{v}_{s}^{[\abs{I'}](q')}, \tilde{v}_{s}^{[\abs{I''}](q'')}, \dot{G}_{s}) \mathd {s},
\end{align}
where \(q'=q(\xi_{I'})\) and \(q''=q(\xi_{I''})\).

The natural norm for \eqref{eq:v-ell-integrated} is the norm \eqref{eq:v-ell-norm} which is compatible with convolution in one of the input variables.
The representation of \(V^{[\ell]} \) as a multilinear functional of the vertex fields in \eqref{eq:V-Mayer-exp} is very convenient: It reduces the infinite-dimensional PDE to a system of ODEs for the kernels using the eigenfunctions of the Laplacian \(\psi(\xi_{1:\ell })\).
This representation makes the action of the Laplacian diagonal, which allows us to solve \eqref{eq:v-ell-integrated}
inductively using only the triangular structure given by \(\ell\).
Given \(\tilde{v}_{t}^{[1]}\), all subsequent kernels \(v^{[\ell]} \) for \(\ell>1 \) are determined by \eqref{eq:v-ell-integrated}.
As \(\beta^{2}\) passes the thresholds \(\beta^{2}_{\ell }\) defined in \eqref{eq:int-beta-thresholds},
we will need to extract additional renormalizations corresponding to the counterterms \(c^{\rho,T}\) in \eqref{eq:Vrho-T}.
Allowing for this renormalization systematically requires some tweaking and will lead us to a refined version of this ansatz.
Modifying \eqref{eq:def-Mayer-Gamma} in a non-trivial way means giving up the eigenbasis \(\psi\),
and, as a result, tends to destroy the triangularity in \(\ell \) alone.

Before introducing the general notation and setup, it is instructive to work through the
first case requiring renormalization.
This computation will serve as a motivating example for the modifications we introduce for the ansatz \eqref{eq:V-Mayer-exp} later.
\subsection{The second-order contribution as a prototype}\label{sec:ell=2}

In this section, we apply the required modifications at order \(\ell=2\) for \(\beta^{2}\in[4\pi,8\pi)\), starting from the ansatz \eqref{eq:V-Mayer-exp}.
The dipole has to be renormalized as soon as \(\beta^{2}\geqslant4\pi\), and as such, previous works on the Mayer expansion, such as
\cite{brydgesMayerExpansionsHamiltonJacobi1987, bauerschmidtLogSobolevInequalityContinuum2021,gubinelliFBDSEApproachSine2026},
all perform this renormalization in one way or another.
We still believe it is useful to include the concrete computation here,
as it serves as a minimal example and as a blueprint for the general procedure in Section \ref{sec:gen-Mayer} below.
The detailed analysis in this section allows us to introduce the relevant ideas without the bookkeeping that will be required to generalize the ansatz to the full regime \(\beta^{2}<8\pi\) and \(\ell \leqslant \frac{1}{\delta}\).

At second order, all kernels are exponentially localized by inspection.
In general, this localization has to be tracked systematically using the Steiner weights \(\omega_t^r\).
For simplicity, we will always set \(\omega=1\) in the definition \eqref{eq:v-ell-norm} throughout this section.

\subsubsection{Uniform estimates on the \(2\)-point kernel}

With the initial condition fixed as in \eqref{eq:v1-explicit}, the first problem already occurs
when trying to propagate the bound to \(\ell = 2\) along \eqref{eq:v-ell-integrated}.

A key player in \eqref{eq:v-ell-integrated} is the integrating factor \(\Gamma_{t,T}(\xi_{1:\ell})\) which satisfies different estimates depending on the charge of the configuration \(\xi_{1:\ell}\).
Since \(G_t\) is positive definite, \(\Gamma_{t,s}\leqslant1\) is true whenever \(s>t\) irrespective of the particular configuration \(\xi_{1:\ell}\) and its charge.
For \(\ell = 2\),
\begin{align}\label{eq:v2-gamma-charge-preq}
	\Gamma_{t,s}(\xi_{1:2}) = \mathe^{-\beta^{2}G_{s,t}(0)} \mathe^{-\beta^{2}\sigma_1 \sigma_2 G_{s,t}(x_1-x_2)}
	\leqslant \mathe^{-\frac{\beta^2}{2}\abs{q(\xi_{1:2})}G_{s,t}(0)}\wedge1.
\end{align}
If \(q(\xi_{1:2})\neq0\), the \(L^\infty(\mathd(x_1,x_2))\) estimate on \(\Gamma\) in \eqref{eq:v2-gamma-charge-preq}, together with the estimate on \(\tilde v^{[1]}\) in \eqref{eq:v1-bound}, is sufficient to obtain bounds uniform in \(T>0\) from the terminal condition \(v_T^{[2]}=0\)
\begin{align}\label{eq:v2-general-est}
	\tnorm{\tilde{v}_{t}^{[2](q)}(\xi_{1:2})}
	\lesssim \int_t^\infty \tnorm{{\tilde{v}}_{s}^{[1]}}^{2} \norm{\dot{G}_{s}}_{L^1} \norm{\Gamma_{t,s}(\xi_{1:2})}_{L^\infty( \mathd {(x_1 , x_2 )})}\mathd {s} \lesssim \bar{\lambda}^{2}\langle t \rangle^{1 - 2\delta}.
\end{align}
For the neutral part of \(\tilde{v}^{[2]}\) the \(L^\infty\)-bounds on the linear propagator do not help in the same way, and in fact these terms will require renormalization as soon as \(\beta^{2} \geqslant 4\pi\) .

\subsubsection{Ansatz for the dipole}

For a regular field \(\varphi\),
\begin{align}\label{eq:cos-localization}
	\abs{\cos(\beta(\varphi(x_1)-\varphi(x_2)))-1} \lesssim \beta^{2}\abs{x_1-x_2}^{2} \norm{\nabla\varphi}^{2}_{L^{\infty}}.
\end{align}
Since the kernels \(\tilde{v}_{t}^{[\ell]}\) are exponentially concentrated near the diagonal \(\{\abs{x_1-x_2}^{2}\sim t^{-1}\}\), the gain in spatial regularity should improve the scaling in \(t\) enough to obtain uniform bounds.
We replace the original ansatz for \(V_{t}^{[2](0)}\) given in \eqref{eq:V-Mayer-exp} by
\begin{align}\label{eq:new-V2-ansatz}
	V_{t}^{[2](0)}[\varphi]
	= & c_{t}^{[2], \rho, T} + \int \rho(\mathd \xi_{1:2}) v_{t}^{[2](0, \alpha)}(\xi_{1:2}) \frac{\cos(\beta (\varphi(x_1)-\varphi(x_2)))-1}{\abs{x_1-x_2}^{\alpha} } \\
	= & c_{t}^{[2], \rho, T} + V_{t}^{[2](0, \alpha)} [\varphi].
\end{align}
To be compatible with the renormalization condition \eqref{eq:Vrho-T}, the constant \(c^{[2], \rho, T} = -\tilde{c}^{\rho, T}\norm{\rho}_{L^{1}}\) may diverge as \(T \to \infty \), but the kernel \(v_t^{[2](0,\alpha)}\) has to be controlled uniformly in \(T>0\) by choosing \(\alpha > 0\) sufficiently large.
Including the modifications discussed above in the neutral part of the potential leads to the system
\begin{align}\label{eq:V2-new}
	\begin{cases}
		V_{t}^{[2](\pm 2)}[\varphi]     & =  \int \rho(\mathd \xi_{1:2}) {v}_{t}^{[2](\pm 2)}(\xi_{1:2}) \psi(\xi_{1:2}),                \\
		V_{t}^{[2](0, \alpha)}[\varphi] & = \int \rho(\mathd \xi_{1:2}) v_{t}^{[2](0,\alpha)}(\xi_{1:2}) \psi^{\alpha}_{0}(\xi_{1:2}),   \\
		V_{t}^{[2](0,0)} [\varphi]      & = c_{t}^{\rho, T},                                                                             \\
		V_{t}^{[2]}[\varphi]            & =  V_{t}^{[2](\pm 2)}[\varphi] + V_{t}^{[2](0, \alpha)}[\varphi] + V_{t}^{[2](0, 0)}[\varphi], \\
	\end{cases}
\end{align}
where for \(q(\xi_{1:2})=0\) we define \emph{the connected dipole with localization degree \(\alpha\)} as
\begin{align}\label{eq:def-dipole}
	\psi^{\alpha}_{0}(\xi_{1:2}) \assign\frac{\psi(\xi_{1:2})-1}{\abs{x_1-x_2}^{\alpha} }.
\end{align}
In this modified ansatz, we require that \eqref{eq:V-ell} holds, and we now derive the corresponding flow equation for the new coefficients \(v_{t}^{[2](0,\alpha )}, v_{t}^{[2](0,0)}\).
The source term of the flow equation is still given by the bilinear term on the right-hand side of \eqref{eq:V-ell} with \(\ell'=\ell''=1 \) and \(q'=-q''=\pm 1\).
In the new basis, this source term can be represented as
\begin{align}\label{eq:V2-flow-modified}
	   & \langle \nabla V_{t}^{[1](\pm 1)}, \dot{G}_t \nabla V_{t}^{[1](\mp 1)}\rangle_{L^{2}}                                                                                                   \\
	=  & \int \rho(\mathd \xi_{1:2}) \mathbb{1}_{\{q(\xi_{1:2})=0\}}v_{t}^{[1] (-1)}(\xi_1) \beta^{2}\dot{G}_{t}(x_1-x_2)v_{t}^{[1](+1)}(\xi_2)\psi(\xi_{1:2})                                   \\
	=  & \int\rho(\mathd \xi_{1:2})\mathbb{1}_{\{q(\xi_{1:2})=0\}}v_{t}^{[1] (-1)}(\xi_1) \beta^{2}\dot{G}_{t}(x_1-x_2)v_{t}^{[1](+1)}(\xi_2)\abs{x_1-x_2}^{\alpha} \psi^{\alpha}_{0}(\xi_{1:2}) \\
	   & + \int\rho(\mathd \xi_{1:2})\mathbb{1}_{\{q(\xi_{1:2})=0\}}v_{t}^{[1] (-1)}(\xi_1) \beta^{2}\dot{G}_{t}(x_1-x_2)v_{t}^{[1](+1)}(\xi_2)                                                  \\
	=: & \int \rho (\mathd {\xi_{1:2}}) B_{[2](0,\alpha)}^{[1],[1]}  (v_{t}^{[1]}, v_{t}^{[1]}, \dot{G}_{t}) (\xi_{1:2})\psi_{0}^{\alpha} (\xi_{1:2})
	+ B_{[2](0,0)}^{[1],[1]}  (v_{t}^{[1]}, v_{t}^{[1]}, \dot{G}_{t}) .
\end{align}

To translate the full equation \eqref{eq:V2-flow-modified} to the level of the kernels \(v^{[2]}\),
we have to compute the action of the Laplacian on the dipole \(\psi^{\alpha}_{0}\), which is no longer an eigenfunction.
We compute,
\begin{align}\label{eq:Laplacian-dipole}
	\frac{1}{2}\Delta_{\dot{G}_{t}} \psi_0^{\alpha}(\xi_{1:2})
	= & [-\beta^{2}\dot{G}_{t}(0) + \beta^{2}\dot{G}_{t}(x_1-x_{2})]\abs{x_1-x_2}^{-\alpha} \psi(\xi_{1:2})                         \\
	= & \dot{\mathbb{G}}_t(\xi_{1:2}) \psi^{\alpha}_{0}(\xi_{1:2}) + \frac{\dot{\mathbb{G}}_t(\xi_{1:2})}{\abs{x_1-x_2}^ {\alpha}}.
\end{align}
Compared to the standard vertex fields, the action is thus no longer constrained to the diagonal:
In the same way as for the vertex fields \(\psi(\xi_{1:2})\), applying the functional derivative leads to a term that reproduces the same field \(\psi^{\alpha}_{0}(\xi_{1:2})\).
However, since the constant \(1\) is removed by the functional derivative, there is an additional contribution.
It corresponds to replacing the neutral vertex field \(\psi(\xi_{1:2})\) by its localization, the constant \(1\).

We define
\begin{align}\label{eq:A-dipole}
	A_{[2](0,0)}^{[2](0, \alpha)} (v_{t}^{[2](0, \alpha)}; \dot{G}_t) :=  \int\rho(\mathd \xi_{1:2}) v_{t}^{[2](0,\alpha)}(\xi_{1:2}) \frac{\dot{\mathbb{G}}_t(\xi_{1:2})}{\abs{x_1-x_2}^{\alpha}}.
\end{align}
Since the basis does not depend on \(t\), the definitions of \(A\) and \(B\) imply that if \(v_t^{[2](0,\alpha)}, v_t^{[2](0,0)}\) satisfy the system
\begin{align}\label{eq:v2-kernel-system}
	\begin{cases}
		\partial_{t} v_{t}^{[2](0,\alpha)}(\xi_{1:2}) + \dot{\mathbb{G}}_{t}(\xi_{1:2}) v_{t}^{[2](0,\alpha)}(\xi_{1:2})  =  B_{[2](0,\alpha)}^{[1],[1]}  (v_{t}^{[1]}, v_{t}^{[1]}, \dot{G}_{t}) (\xi_{1:2}), \\
		\partial_{t} v_{t}^{[2](0,0)} = - A_{[2](0,0)}^{[2](0,\alpha)}(v_{t}^{[2](0,\alpha)}; \dot{G}_t) + B_{[2](0,0)}^{[1],[1]}  (v_{t}^{[1]}, v_{t}^{[1]}, \dot{G}_{t}),
	\end{cases}
\end{align}
then \eqref{eq:V-ell} also holds for \(\ell=2\).

\subsubsection{Estimates on the modified flow equation for \(\ell=2\)}\label{sec:ell=2-flow-est}

We now want to use the differential equations for the kernels in \eqref{eq:v2-kernel-system} to obtain bounds uniform in \(T>0\) on \(v^{[2](0,\alpha)}, v_{t}^{[2](0,0)}\) for any \(\beta^{2}<8\pi\).
The modifications to the ansatz mean that \eqref{eq:v2-kernel-system} is no longer triangular in \(\ell\) alone.
However, triangularity can be restored if we include also the number of dipoles:
Since \(\frac{1}{2}\Delta_{\dot{G}_{t}}V_{t}^{[2](0,0)} = 0\), the kernel \(v^{[2](0,0)} \) does not contribute to the equation for \(v^{[2](0,\alpha)} \).
As a result, we can propagate the bounds from the first level inductively.
Starting with \(v^{[2](0,\alpha)}\), the source term depends only on the lower order contribution \(v^{[1]}\),
and for \(\alpha>0\) sufficiently large, we obtain bounds uniform in \(T>0\).
We then use these bounds on \(v^{[2](0,\alpha)}\) to control \(v_{t}^{[2](0,0)} \).
This phenomenon is a general structural feature of the flow equation.
The bilinear part may create neutral clusters, but the linear part coming from the functional Laplacian cannot.

The renormalization condition \eqref{eq:Vrho-T} forces the terminal conditions
\begin{align}\label{eq:v2-terminal-condition}
	v_{T}^{[2](0,\alpha)}(\xi_{1:2}) = 0,
	\qquad
	v_{T}^{[2](0,0)} = c^{\rho, T}_{T} = c^{\rho, T}.
\end{align}
Integrating \eqref{eq:v2-kernel-system} (recall the definition of \(\Gamma\) in \eqref{eq:def-Mayer-Gamma}), we find
\begin{align}\label{eq:v2-kernel-system-integrated}
	\begin{cases}
		v_{t}^{[2](0,\alpha)}(\xi_{1:2}) = -\int_{t}^{T}  \Gamma_{t,s}(\xi_{1:2})  B_{[2](0,\alpha)}^{[1],[1]}  (v_{s}^{[1]}, v_{s}^{[1]}, \dot{G}_{s}) (\xi_{1:2}) \mathd s, \\
		v_{t}^{[2](0,0)} = c^{T,\rho, [2]} - \int_{t}^{T}   \left\{
		B_{[2](0,0)}^{[1],[1]}  (v_{s}^{[1]}, v_{s}^{[1]}, \dot{G}_{s})
		- A_{[2](0,0)}^{[2](0,\alpha)}(v_{s}^{[2](0,\alpha)}; \dot{G}_{s})
		\right\}\mathd s.
	\end{cases}
\end{align}
\begin{lemma}\label{lem:v2-A,B-estimates}
	The following estimates apply for any \(s\geqslant t\):
	\begin{align}\label{eq:v2-A,B-estimates}
		\begin{split}
			\tnorm{B_{[2](0,\alpha)}^{[1],[1]}  (v_{s}^{[1]}, v_{s}^{[1]}, \dot{G}_{s})}
			 & \lesssim \langle s \rangle^{-2 - \frac{\alpha}{2}} \tnorm{v_{s}^{[1]}}^{2},                         \\
			\tnorm{B_{[2](0,0)}^{[1],[1]}  (v_{s}^{[1]}, v_{s}^{[1]}, \dot{G}_{s})}
			 & \lesssim \norm{\rho}_{L^{1}} \langle s \rangle^{-2} \tnorm{v_{s}^{[1]}}^{2},                        \\
			\tnorm{ A_{[2](0,0)}^{[2](0,\alpha)}(v_{s}^{[2](0,\alpha)}; \dot{G}_{s})}
			 & \lesssim \norm{\rho}_{L^{1}}\langle s \rangle^{-1+\frac{\alpha}{2}} \tnorm{v_{s}^{[2](0,\alpha)} }.
		\end{split}
	\end{align}
	In particular, if \(\alpha>\alpha_{2}\), where
	\begin{align}\label{eq:v2-alpha-condition}
		\alpha_{2} \assign 2(1-2\delta),
	\end{align}
	then for every \(\rho\prec 1\) and \(T\in (0,\infty)\), there is a choice of \(c^{T,\rho}\) such that, uniformly in \(T>0\),
	\begin{align}\label{eq:v2-alpha-estimates}
		\tnorm{v^{[2](0,\alpha)}_{t} }\lesssim \bar{\lambda}^2 \beta^2 \langle t  \rangle ^{1-2\delta - \frac{\alpha}{2}},
		\qquad
		\tnorm{v^{[2](0,0)}_{t} }\lesssim \bar{\lambda}^2 \beta^2 \langle t  \rangle ^{1-2\delta}.
	\end{align}
\end{lemma}
\begin{proof}
	We first prove \eqref{eq:v2-alpha-estimates} assuming \eqref{eq:v2-A,B-estimates}.
	Combining \eqref{eq:v2-A,B-estimates} with \(\Gamma_{t,s}(\xi_{1:2})\leqslant 1\) and the estimate on \(v^{[1]}\) from \eqref{eq:v1-bound}, we can directly compute
	\begin{align}
		\tnorm{B_{[2](0,\alpha)}^{[1],[1]}  (v_{s}^{[1]}, v_{s}^{[1]}, \dot{G}_{s})}
		\lesssim \langle s\rangle^{-2-\frac{\alpha}{2}} \tnorm{v_s^{[1]}}^{2},
	\end{align}
	so that
	\begin{align}\label{eq:v2-alpha-est-proof}
		\tnorm{v_{t}^{[2](0,\alpha)}}
		\lesssim
		\int_{t}^{T} \tnorm{B_{[2](0,\alpha)}^{[1],[1]}  (v_{s}^{[1]}, v_{s}^{[1]}, \dot{G}_{s})} \mathd s
		\lesssim
		\int_{t}^{T} \langle s  \rangle ^{1-2\delta-\frac{\alpha}{2}}\frac{\mathd s}{\langle s \rangle }
		\lesssim
		\langle t  \rangle ^{1-2\delta-\frac{\alpha}{2}},
	\end{align}
	provided \(\alpha>\alpha_{2}\).
	Moving on to \(v_{t}^{[2](0,0)}\), we do not have any help from \(\alpha\):
	inserting \eqref{eq:v2-A,B-estimates} with the estimate just obtained for \(v_{t}^{[2](0,\alpha)}\), we can only hope for the bound
	\begin{align}\label{eq:v2-constant-divergent}
		\abs{B_{[2](0,0)}^{[1],[1]}  (v_{s}^{[1]}, v_{s}^{[1]}, \dot{G}_{s})}
		+
		\abs{A_{[2](0,0)}^{[2](0,\alpha)}(v_{s}^{[2](0,\alpha)}; \dot{G}_{s})}
		\lesssim
		\langle s  \rangle ^{-2\delta},
	\end{align}
	which is not integrable for \(\beta^{2}\geqslant4\pi\) and would destroy the induction.
	However, by choosing
	\begin{align}\label{eq:v2-fixing-c2}
		c^{T,\rho, [2]} \assign \int_{0}^{T} \mathd s  \left\{
		B_{[2](0,0)}^{[1],[1]}  (v_{s}^{[1]}, v_{s}^{[1]}, \dot{G}_{s})
		- A_{[2](0,0)}^{[2](0,\alpha)}(v_{s}^{[2](0,\alpha)}; \dot{G}_{s})
		\right\},
	\end{align}
	the second equation in \eqref{eq:v2-kernel-system-integrated} becomes
	\begin{align}\label{eq:v2-flow-forward}
		v_{t}^{[2](0,0)} = \int_{0}^{t}  \mathd s \left\{
		B_{[2](0,0)}^{[1],[1]}  (v_{s}^{[1]}, v_{s}^{[1]}, \dot{G}_{s})
		- A_{[2](0,0)}^{[2](0,\alpha)}(v_{s}^{[2](0,\alpha)}; \dot{G}_{s})
		\right\}.
	\end{align}
	Inserting \eqref{eq:v2-constant-divergent} into \eqref{eq:v2-flow-forward} yields the required estimate in \eqref{eq:v2-alpha-estimates}.

	It now remains to show \eqref{eq:v2-A,B-estimates}.
	Starting with the first estimate, using Lemma \ref{lem:hk-taylor-remainder}, we obtain directly
	\begin{align}\label{eq:v2-B2-alpha-est}
		 & \tnorm{B_{[2](0,\alpha)}^{[1],[1]}  (v_{s}^{[1]}, v_{s}^{[1]}, \dot{G}_{s})}                                                                                                \\
		 & \leqslant
		\sup_{x_{1}\in \mathbb{R}^{2}} \abs{\rho (x_1)}\int \rho(\mathd x_{2}) \abs{v_{s}^{[1]}(\xi_{1})} \abs{\dot{G}_{s}(x_1-x_2)}\abs{x_1-x_2}^{\alpha} \abs{v_{s}^{[1]}(\xi_{2}) } \\
		 & \leqslant \tnorm{v_{s}^{[1]}}^{2} \int\mathd x \abs{x}^{\alpha} \abs{\dot{G}_s (x)}                                                                                         \\
		 & \lesssim \langle s  \rangle ^{-2-\frac{\alpha}{2}}  \tnorm{v_{s}^{[1]}}^{2}.
	\end{align}
	The estimates on \(B_{[2](0,0)}^{[1],[1]}\) follow similarly.
	Indeed, the only difference is that the estimate will now depend on the cut-off \(\rho\),
	\begin{align}\label{eq:v2-B2-0-est}
		\tnorm{B_{[2](0,0)}^{[1],[1]}(v_{s}^{[1]}, v_{s}^{[1]}, \dot{G}_{s})}
		\leqslant & \abs{B_{[2](0,0)}^{[1],[1]}(v_{s}^{[1]}, v_{s}^{[1]}, \dot{G}_{s})}                                     \\
		\leqslant & \beta^{2}\int \rho(\mathd \xi_{1:2}) \abs{v_{s}^{[1]}(\xi_1)} \dot{G}_s(x_1-x_2) \abs{v_s^{[1]}(\xi_2)} \\
		\lesssim  & \beta^{2}\tnorm{v_{s}^{[1]}}^{2} \norm{\rho}_{L^1} \norm{\dot{G}_{s}}_{L^1}                             \\
		\lesssim  & \bar{\lambda}^{2} \beta^{2}\langle s \rangle^{-2\delta} \norm{\rho}_{L^1}.
	\end{align}

	Regarding the estimate on \(A_{[2](0,0)}^{[2](0,\alpha)}\), we have to control
	\begin{align}
		\abs{A_{[2](0,0)}^{[2](0,\alpha)}(v_{t}^{[2](0,\alpha)}; \dot{G}_{t})}
		= \beta^{2} \abs{\int \rho(\mathd \xi_{1:2}) {v_{t}^{[2](0,\alpha)}(\xi_{1:2})}
			\frac{\dot{{G}}_t(0)-\dot{{G}}_t(x_1-x_2)}{\abs{x_1-x_2}^{\alpha}}},
	\end{align}
	which in particular requires us to absorb a divergence of degree \(\alpha > 1\) as soon as \(\beta^{2} \geqslant 6\pi\).
	We accomplish this using the neutrality of \(\xi_{1:2}\) and the resulting second-order heat-kernel gain.
	Indeed, the definition \eqref{eq:def-Gbb} gives
	\begin{align}\label{eq:G-2pt-taylor}
		\dot{\mathbb G}_t(\xi_{1:2})
		= -\frac{\beta^{2}}{2}\,(S_{\xi_{1:2}}\otimes S_{\xi_{1:2}})[\dot G_t]
		= -\beta^{2}\bigl(\dot G_t(0)-\dot G_t(x_1-x_2)\bigr).
	\end{align}
	Thus, using the heat-kernel estimate in Lemma \ref{lem:hk-multipoint},
	\begin{align}\label{eq:v2-A-estimate}
		                  & \abs{A_{[2](0,0)}^{[2](0,\alpha)}(v_{t}^{[2](0,\alpha)}; \dot{G}_{t})}                       \\
		=       \beta^{2} & \abs{\int \rho(\mathd \xi_{1:2}) {v_{t}^{[2](0,\alpha)}(\xi_{1:2})}
		\frac{\dot{{G}}_t(0)-\dot{{G}}_t(x_1-x_2)}{\abs{x_1-x_2}^{\alpha}}}                                              \\
		\leqslant         & \norm{\rho}_{L^1} \tnorm{v_{t}^{[2](0, \alpha)} } \langle t \rangle^{-1+\frac{\alpha}{2}}    \\
		\lesssim          & \langle t  \rangle ^{-1+\frac{\alpha}{2}} \norm{\rho}_{L^1} \tnorm{v_{t}^{[2](0, \alpha)} }.
	\end{align}
\end{proof}

\subsubsection{Bounds on the potential and force}
With the bounds on the kernels \(v_t^{[\ell]}\) complete, we now transfer them to the potential and the force.
For simplicity, we assume \(\alpha_{2}>1 \iff \delta<\frac{1}{4} \iff \beta^{2}> 6\pi\).
This is exclusively to ensure that \(\alpha_{2}-1>0\) and to avoid writing \((\alpha_{2}-1)\vee 0\).
The same argument still holds if \(\beta^{2}\leqslant6\pi\), with \(\alpha_2-1\) replaced by \(0\).
\begin{lemma}\label{lem:V2-F2-estimates}
	Let \(\alpha\in (\alpha_2, 2)\).
	With \((V_{t}^{[2]})_{t\in[0,T]}\) defined as in \eqref{eq:V2-new} and \(F^{[2]}_{t}[\varphi] \assign -\nabla V^{[2]}_{t}[\varphi]\), we have
	\begin{align}\label{eq:V2-F2-estimates}
		\abs{V_{t}^{[2]}[\varphi]}                   & \lesssim \bar{\lambda}^{2} \beta^{2}\norm{\rho}_{L^{1}}\langle t  \rangle ^{1-2\delta}(1 + (\langle t  \rangle ^{-\frac{1}{2}} \norm{\nabla\varphi}_{L^{\infty} })^{2}), \\
		\norm{Q_{t}F^{[2]} _t[\varphi]}_{L^{\infty}} & \lesssim \bar{\lambda}^{2} \beta^{2}\langle t  \rangle ^{-2\delta}(1 + (\langle t  \rangle ^{-\frac{1}{2}} \norm{\nabla\varphi}_{L^{\infty}})^{\alpha-1}).
	\end{align}
\end{lemma}
\begin{proof}
	In the charged case, the estimates are immediate from the estimates on \(v_{t}^{[2](\pm 2)}\) and the trivial bound \(\abs{\psi(\xi_{1:2})}\leqslant1\).
	Similarly, \(V_{t}^{[2](0,0)}\) does not depend on the field \(\varphi\) and satisfies the bound \eqref{eq:V2-F2-estimates} thanks to Lemma \ref{lem:v2-A,B-estimates}.

	The only remaining case corresponds to the contribution \(V^{[2](0,\alpha)}\), where we need to control the connected dipole with localization degree \(\alpha\in (1,2)\).
	Here, we will have to use the symmetry under charge conjugation.
	Since \(v^{[2](0,\alpha)}_t(\xi_{1:2})= v^{[2](0,\alpha)}_{t}(\bar\xi_{1:2})\), and \(S_{\xi_{1:2}} = -S_{\bar\xi_{1:2}}\),
	\begin{align}\label{eq:v2-symmetry}
		\int \rho (\mathd \xi_{1:2}) v_{t}^{[2](0, \alpha)}(\xi_{1:2}) \abs{x_1-x_2}^{-\alpha} S_{\xi_{1:2}} (\varphi) = 0.
	\end{align}
	Consequently,
	\begin{align}\label{eq:V2-potential-estimate}
		          & \abs{V_{t}^{[2](0, \alpha)} [\varphi] }                                                                                                                                        \\
		\leqslant & \int \rho(\mathd \xi_{1:2}) \abs{v_{t}^{[2](0,\alpha)}(\xi_{1:2}) \abs{x_1-x_2}^{2-\alpha} }\abs{\frac{\psi(\xi_{1:2})-i \beta S_{\xi_{1:2}}(\varphi) - 1}{\abs{x_1-x_2}^{2}}} \\
		\lesssim  & \norm{\rho}_{L^{1}} \langle t \rangle^{1-2\delta}   (\langle t \rangle^{-\frac{1}{2}}  \norm{\nabla \varphi}_{L^{\infty} })^{2},
	\end{align}
	where we used the exponential concentration of \(v^{[2]}\) and
	\(\abs{\mathe^{i u} - 1 - iu} \lesssim \abs{u}^{2}\) in the last step.
	Regarding the force, the estimates for the charged part are immediate from \eqref{eq:v2-general-est} and the \(L^1\)-estimate for \(Q_t\) in Lemma \ref{lem:hk-Lp-bounds}.
	In the neutral case, we only control the combination \(Q_{t}F_{t}^{[2](0,\alpha)}[\varphi] \), which also benefits from improved estimates.
	We compute
	\begin{align}\label{eq:F2-good-estimate}
		  & Q_{t}F_{t}^{[2](0,\alpha)}[\varphi]
		= - Q_{t}\tmop{D}V_{t}^{[2](0,\alpha)}[\varphi]                                                                            \\
		= & - \int \rho(\mathd \xi_{1:2}) v_{t}^{[2](0,\alpha)} (\xi_{1:2}) \tmop{D}\psi^{\alpha}_{0}(\xi_{1:2}) [Q_{t}(\cdot-x)],
	\end{align}
	where, using \(\sigma_1 = - \sigma_2 \),
	\begin{align}\label{eq:dipole-G-gain}
		  & \tmop{D}\psi^{\alpha}_{0}(\xi_{1:2}) [Q_{t}(\cdot-x)]                                                                                                                                         \\
		= & \sigma_1i \beta\frac{[Q_{t}(x - x_1) - Q_{t}(x - x_2)]}{\abs{x_1-x_2}} \psi_{0}^{\alpha -1 } (\xi_{1:2}) + \sigma_1 i \beta\frac{[Q_{t}(x - x_1) - Q_{t}(x - x_2)] }{\abs{x_1-x_2}^{\alpha}}.
	\end{align}
	In other words, the functional derivative reduces the localization degree of the dipole by \(1\) at the cost of introducing a finite difference in \(Q_{t}\).
	For the first field-dependent part, we estimate using \( \abs{\mathe^{i \beta \varphi(x)}- 1} \leqslant 2\) and \(\alpha-1 \leqslant1\),
	\begin{align}\label{eq:D-dipole-leq1}
		\norm{\psi_{0}^{\alpha -1 } (\xi_{1:2})}_{L^{\infty}}
		=        & \sup_{x_1,x_2\in \mathbb{R}^{2}} \abs{\frac{\mathe^{i \beta \sigma_1 (\varphi(x_1)- \varphi(x_2))}-1}{\abs{x_1-x_2}^{\alpha-1}}} \\
		\lesssim & \sup_{x_1,x_2\in \mathbb{R}^{2}} \abs{\frac{\mathe^{i \beta \sigma_1 (\varphi(x_1)- \varphi(x_2))}-1}{\abs{x_1-x_2}}}^{\alpha-1} \\
		\lesssim & \norm{ \nabla \varphi}_{L^\infty}^{\alpha-1}.
	\end{align}
	Combining this with Young's inequality, the heat-kernel estimate in Lemma \ref{lem:hk-taylor-remainder}, and the estimates on the kernels from Lemma \ref{lem:v2-A,B-estimates},
	\begin{align}\label{eq:F2-field-dependent-part}
		          & \sup_{x\in \mathbb{R}^{2}}\abs{\int \rho(\mathd \xi_{1:2}) v_{t}^{[2](0,\alpha)} (\xi_{1:2})\sigma_1 i\beta\frac{Q_t(x-x_1)-Q_t(x-x_2)}{\abs{x_1-x_2}}\psi_{0}^{\alpha -1 } (\xi_{1:2})}     \\
		\leqslant & \beta \sup_{x\in \mathbb{R}^{2}}\int \rho(\mathd \xi_{1:2}) \abs{v_{t}^{[2](0,\alpha)} (\xi_{1:2})}\frac{\abs{Q_t(x-x_1)-Q_t(x-x_2)}}{\abs{x_1-x_2}} \abs{\psi_{0}^{\alpha -1 } (\xi_{1:2})} \\
		\lesssim  & \langle t  \rangle ^{-2+\frac{1}{2}} \tnorm{v_{t}^{[2](0,\alpha)} } \norm{\nabla \varphi}_{L^{\infty } }^{\alpha-1}                                                                          \\
		\lesssim  & \langle t  \rangle ^{-2\delta} \langle t \rangle^{-\frac{\alpha-1}{2}}  \norm{\nabla \varphi}_{L^{\infty } }^{\alpha-1}.
	\end{align}
	Regarding the second, field-independent part of \eqref{eq:dipole-G-gain},
	the denominator and the kernel are invariant under charge conjugation, while \(S_{\bar\xi_{1:2}}=-S_{\xi_{1:2}}\).
	Therefore, for any \(x\in \mathbb{R}^{2}\),
	\begin{align}\label{eq:F2-field-independent-part}
		\int \rho(\mathd \xi_{1:2}) v_{t}^{[2](0,\alpha)} (\xi_{1:2})\frac{S_{\xi_{1:2}}(Q_t(x-\cdot))}{\abs{x_1-x_2}^{\alpha}}=0.
	\end{align}
	Combining the estimate \eqref{eq:D-dipole-leq1} with \eqref{eq:F2-field-dependent-part} yields the claimed estimate \eqref{eq:V2-F2-estimates} on \(F^{[2](0, \alpha)}_{t}\).
\end{proof}
Before we move on to the general case, let us collect some remarks for future reference.
\begin{remark}\label{rem:cancelation}
	The identity \eqref{eq:G-2pt-taylor} is the two-point instance of the cancellation mechanism used for higher-order multipoles.
	We record the relevant general observations.
	\begin{itemize}
		\item If both \(\xi_{I}, \xi_{J}\) are neutral, both the constant and the first-order terms in the computation analogous to \eqref{eq:G-2pt-taylor} for \(\dot{\mathbb{G}}(\xi_I, \xi_J)\) vanish pointwise as \(S_{\xi_{I}}(1) = S_{\xi_{J}}(1) = 0\).
		      In particular, in the self-interacting case generated in \(A\), this always happens for neutral configurations.
		\item If at least one of \(\xi_I, \xi_J \) is neutral, say \(\xi_I\) is neutral,  only the constant Taylor term vanishes pointwise.
		      Divergences \(\alpha > 1\) can only be absorbed when integrating against a kernel \(v(\xi_{I}) = v(\bar{\xi}_{I}) \) (see \eqref{eq:v2-symmetry} and \eqref{eq:F2-field-independent-part}).
	\end{itemize}
\end{remark}
\begin{remark}\label{rem:weighted-estimates}
	For simplicity, all estimates on \(V, F\) in the previous section are in terms of \(\norm{\nabla\varphi}_{L^\infty}\).
	In practice, since the Gaussian free field diverges logarithmically at spatial infinity,
	we are forced to use weighted estimates and replace \(\norm{\nabla \varphi}_{L^\infty}\) by \(\norm{\nabla\varphi}_{L^\infty(\chi)}\).
\end{remark}
\begin{remark}\label{rem:v2-choosing-alpha}
	If we are only interested in an estimate for the potential, the choice to work with fractional finite differences, namely the localized dipoles for some \(\alpha\in (\alpha_{2},2)\), seems unnecessary and one might be tempted to use the integer grading \(\alpha\in \{0, 1, 2\}\)  for the localization degree.
	However, in the estimate for \(Q_tF_t^{[2](0,\alpha)} \), the advantage of allowing fractional localization degrees via the finite differences becomes clear.
	For instance, the force corresponding to the dipole has sublinear estimates in the full subcritical regime as shown in Lemma \ref{lem:V2-F2-estimates}.
	Note also that we do not have \(L^\infty\) control of \(F_t^{[2](0, \alpha)}\) uniformly in \(T>0\) when \(\beta^{2}\geqslant 6 \pi\), which is why we always consider the combination \(Q_tF^{[2](0, \alpha)} _t\).
\end{remark}
\begin{remark}\label{rem:v2-cosine-estimate}
	Equivalently, we could use the representation of \(V^{[2](0,\alpha)}_{t}\) in terms of the cosine and
	\begin{align}\label{eq:v2-cosine}
		\abs{\cos(\beta S_{\xi_{1:2}}(\varphi))-1}
		\lesssim \abs{x_1-x_2}^{2} \norm{\nabla \varphi}_{L^{\infty}}^{2},
	\end{align}
	to arrive at the same conclusion.
	Then, the first two functional derivatives \(\tmop{D}\) act on the dipole \eqref{eq:v2-cosine} as
	\begin{align}\label{eq:v2-D-cosine}
		\tmop{D} \frac{\cos(\beta (S_{\xi_{1:2}}(\varphi)))-1}{\abs{x_1-x_2}^{2}}  [h]
		 & = -\beta S_{\xi_{1:2}}^{1}(h)   \frac{\sin(\beta S_{\xi_{1:2}}(\varphi))}{\abs{x_1 - x_2}}.      \\
		\tmop{D^2} \frac{\cos(\beta S_{\xi_{1:2}}(\varphi)) - 1}{\abs{x_1-x_2}^{2}}[h_1, h_2 ]
		 & = -\beta^{2} \cos(\beta S_{\xi_{1:2}}(\varphi)) S_{\xi_{1:2}}^{1} (h_1) S_{\xi_{1:2}}^{1} (h_2),
	\end{align}
	where we used the notation introduced in \eqref{eq:S-frac},
	\begin{align}\label{eq:S-dipole}
		S_{\xi_{1:2}}^{\alpha}(h)= \frac{\sigma_1 h(x_1) + \sigma_2h(x_2 )}{\abs{x_1-x_2}^{\alpha}} \qquad S_{\xi_{1:2}} = S_{\xi_{1:2}}^{0}.
	\end{align}
	In this way the functional derivative allows one to transfer some of the divergence away from the dipole and onto the test functions.
	In particular, after two applications, if \(h_1 , h_2 \) are smooth enough,
	all the required regularity can be obtained from the test functions instead of the field \(\varphi\).
\end{remark}

\begin{remark}
	The same effect can be achieved by changing the norms \(\tnorm{\cdot}_{t}\) for the dipole instead of including a localization inside the ansatz, as is done in \cite{gubinelliFBDSEApproachSine2026}.
	See also Appendix \ref{app:T-dependent} for the definition of the general weight.
\end{remark}

\subsection{The multipole Mayer expansion}\label{sec:gen-Mayer}

The goal is now to automate the steps in Section \ref{sec:ell=2} and obtain bounds on \(V^{[\ell]}\) for any finite \(\ell\in \mathbb{N}\) and \(\beta^{2}<8\pi \iff \delta>0\).
As we have seen, it is not sufficient to track only the loop order; the number of point charges, the overall charge, and the localization degree of the fields also need to be tracked.
We will collect all of this information using a suitable index set \(\mathfrak{A}\) allowing us to generalize \eqref{eq:v2-kernel-system-integrated} more compactly.
With this in mind, our expansion generalizes \eqref{eq:V2-new} and takes the form
\begin{align}\label{eq:def-Va}
	V_{t}^{\rho,T}[\varphi] = \sum_{\mathfrak{a}\in \mathfrak{A}} \int \rho(\mathd {\xi^{\mathfrak{a}}}) v_{t}^{\mathfrak{a},\rho, T}(\xi^{\mathfrak{a}}) \Psi^{\mathfrak{a}}(\xi^{\mathfrak{a}}),\qquad \rho\prec 1, T\in [0,\infty).
\end{align}
The precise meanings and definitions of \(\mathfrak a\), \(\Psi^{\mathfrak a}\), and the kernels \(v^{\mathfrak a,\rho,T}\) will be given below.
We fix \(\rho\prec 1\) and \(T>0\).
Let us state the structural constraints (observed already in Section \ref{sec:ell=2}) that \eqref{eq:def-Va} has to satisfy to guide the definitions and constructions that follow.

\subsubsection[Requirements on the modified ansatz]{Requirements on the ansatz {(\eqref{eq:def-Va})}}

\begin{enumerate}[a)]
	\item \textbf{Higher-order multipoles.} The basis \(\Psi^{\mathfrak{a}}\) includes neutral clusters (such as the dipole).
	\item \textbf{Triangularity.} The system of equations (corresponding to the generalization of \eqref{eq:V2-flow-modified} and \eqref{eq:v2-kernel-system-integrated}) is triangular in the combined partial order induced by the loop order and the multi-index tracking the number of clusters.
	\item  \textbf{Uniform estimates.} If \(\Psi^{\mathfrak{a}}\) contains more than one point, we have to choose \(v^{\mathfrak{a}, \rho, T}_{T}=0\) and control \(\tnorm{v^{\mathfrak{a}, \rho, T}}\) uniformly in \(T>0\).
	      That is, the only terms requiring a non-zero terminal condition correspond either to a field-independent constant or a cosine, as required by \eqref{eq:Vrho-T}.
\end{enumerate}

To achieve this, we will translate the action of the differential operators in the flow equation to maps acting on the indices and the kernels.
This will require some notational setup but in turn gives a systematic procedure and avoids distinguishing cases.

\subsubsection{Outline}
We now define the different parts required to make sense of \eqref{eq:def-Va} while respecting the requirements stated above.
First, we define the general neutral clusters of size \(2n\) in Section \ref{sec:def-clusters},
before using them to define a generic index set \(\mathfrak{A}\) and the basic fields \(\Psi^{\mathfrak{a}}\) in Section \ref{sec:def-index}.
We identify a suitable set of admissible indices and classify them according to their scaling dimension in Section \ref{sec:admissible}.
To model the action of the flow equation on functionals of the form \eqref{eq:def-Va},
we introduce maps to modify the indices in Section \ref{sec:r-l-def}.
With these definitions, we can define the coefficient kernels \(v^{\mathfrak{a}}\) with the dual maps in Section \ref{sec:kernels} and state the flow equation for the coefficients with the main estimates (see Proposition \ref{prop:A-B-estimates}) in Section \ref{sec:flow-coeff}.
We prove Theorem \ref{thm:int-Va-estimates} and the corresponding sublinear result in Section \ref{sec:pot-for-estimates}.
The technical proof of the main proposition, Proposition \ref{prop:A-B-estimates}, is delayed to Section \ref{sec:details-flow}.

\subsubsection{The \(2n\)-point clusters}\label{sec:def-clusters}
In general, the same problem as in the case \(\ell = 2\) can appear for neutral contributions of any size \(2n \leqslant \frac{1}{\delta}\).
We therefore introduce \(2n\)-point clusters.
For a set of charges \(\xi_{1:2n}\) such that \(q(\xi_{1:2n})=0\), recall the spread \(\delta(\xi_{1:2n})\) from \eqref{eq:def-delta-cluster}.
The choice of base point is arbitrary. We take the infimum only to preserve the symmetry among the points \(x_k\); another possible choice would be, for example, the center of mass.
We define the fractional \(2n\)-point clusters as
\begin{align}\label{eq:def-2n-pt-cluster}
	\psi_{0}^{\nu}(\xi_{1:2n}) = \frac{\psi(\xi_{1:2n})-1}{\delta(\xi_{1:2n})^{\nu}},\qquad \text{whenever \(\nu\in  [0,\infty)\)},
\end{align}
with the convention that \(\psi^{\nu}_0(\xi_{1:2n}) = 0 \) whenever \(\delta(\xi_{1:2n}) = 0\).
To define the index set and the operations on the indices it is more convenient to work with the full range \([0,\infty)\) that is closed under addition of positive numbers.
We will restrict the value of \(\nu\) to concrete values in \([0, 2)\) later.

\begin{lemma}\label{lem:2n-estimate}
	Let \(q(\xi_{1:2n})=0\). It holds that
	\begin{align}\label{eq:2n-nu-estimate-nu}
		\norm{\psi^{\nu}_{0}(\xi_{1:2n})}_{L^\infty(\mathd x_{1:2n})}
		 & \lesssim \norm{\nabla \varphi}_{L^\infty(\mathd x_{1:2n})}^{\nu}, \qquad \nu\leqslant1, \\
		\label{eq:2n-2-estimate}
		\norm{\frac{\cos(\beta S_{\xi_{1:2n}}(\varphi))-1}{\delta(\xi_{1:2n})^{2}}}_{L^\infty(\mathd x_{1:2n})}
		 & \lesssim \norm{\nabla \varphi}_{L^{\infty}(\mathd x_{1:2n})}^{2}.
	\end{align}
\end{lemma}
\begin{proof}
	The first estimate is the multipoint version of \eqref{eq:D-dipole-leq1}, and it follows in the same way as the two-point case.
	Since \(S_{\xi_{1:2n}}(1)=0\) by neutrality, a Taylor expansion allows estimating
	\begin{align}\label{eq:S-nabla-estimate}
		\abs{S_{\xi_{1:2n}}(\varphi)} \leqslant \norm{\nabla \varphi}_{L^{\infty } } \delta(\xi_{1:2n}; x_{k}).
	\end{align}
	If \(\nu\leqslant 1\) interpolating between \(\abs{\psi(\xi_{1:2n})-1}\leqslant2\) and \(\abs{\psi(\xi_{1:2n})-1}\leqslant \abs{S_{\xi_{1:2n}}(\varphi)}\) yields \eqref{eq:2n-nu-estimate-nu}.
	The second part follows in the same way; see also, e.g., \eqref{eq:v2-cosine}.
\end{proof}

\subsubsection{The index set and vertex fields} \label{sec:def-index}
We gather all the information we need to track (such as loop order, charge, and the connected fractional derivatives of the \(2n\)-point clusters) in a single symbol \(\mathfrak a\) belonging to an index set
\begin{align}
	\mathfrak{A} \assign \{\mathfrak{a} = (L(\mathfrak{a}),Q(\mathfrak{a}),  N(\mathfrak{a}), \nu(\mathfrak{a}))\}.
\end{align}
We define the following quantities.
\begin{itemize}
	\item \(L(\mathfrak{a})\in \mathbb{N}\) is the \emph{loop order}.
	\item \(Q(\mathfrak{a})\in \{-L(\mathfrak{a}),\dots, L(\mathfrak{a})\}\) is the \emph{charge}.
	\item \(N(\mathfrak{a})=(N_r(\mathfrak{a}))_{r\geqslant1}\in\mathbb{N}_{0}^{\mathbb{N}}\) is finitely supported, where \(N_1(\mathfrak{a})\) is the number of monopoles and \(N_{2n}(\mathfrak{a})\) is the number of neutral \(2n\)-point clusters.
	      We further set \(N_r(\mathfrak{a}) = 0\) if \(r\not\in \{1\}\cup 2 \mathbb{N}\).
	\item For each \(n\geqslant1\), define the set of labels \(I_{2n}(\mathfrak{a})
	      :=\{(2n,k):1\leqslant k\leqslant N_{2n}(\mathfrak{a})\}\) and \(I_{\mathrm{n}}(\mathfrak{a}) :=\bigcup_{n\geqslant1}I_{2n}(\mathfrak{a})\). 
		  The \emph{localization degrees} are the size-indexed family 
		  \(\nu(\mathfrak{a}) = \bigl(\nu_{2n,k}(\mathfrak{a})\bigr)_ {(2n,k)\in I_{\mathrm{n}}(\mathfrak{a})} \in \prod_{n\geqslant1} [0,\infty)^{N_{2n}(\mathfrak{a})}\).
	      Monopoles carry no localization but whenever a uniform notation is useful,
	      we set \(\nu_{1,k}(\mathfrak{a})=0\).
\end{itemize}
This index set is (intentionally) too large, and we will restrict to a suitable subset of indices compatible with our construction shortly.
For \(\mathcal{X} = \{\pm1\}\times \mathbb{R}^{2}\), \(n\in \mathbb{N}\), and \(q\in \mathbb{Z}\), we define the fixed-charge configuration spaces
\begin{align}
	\mathcal{X}^{(q)}_n
	:= \left\{ \xi_{1:n} \in \mathcal{X}^n:\: q(\xi_{1:n})
	= \sum_{j} \sigma_j= q\right\}.
\end{align}
We use the convention that \(\mathcal{X}_0^{(0)} = \{\varnothing\}\) and \(\mathcal{X}_0^{(q)} = \varnothing\) for \(q\neq 0 \) to ensure products and integrals behave as expected.
The space \(\mathcal{X}^{(q)}_{n}\) is non-empty only when \(\abs{q}\leqslant n\) and \(n \equiv q \pmod 2\).
To each index \(\mathfrak{a}\in \mathfrak{A}\) we associate the restricted configuration space
\begin{align}\label{eq:def-Xcal}
	\mathcal{X}^{\mathfrak{a}}
	:= \mathcal{X}^{\mathfrak{a}, \mathtt{c}}\times \mathcal{X}^{\mathfrak{a}, \mathtt{n}}
	:= \mathcal{X}_{N_1(\mathfrak{a})}^{(Q(\mathfrak{a}))} \times \prod_{n\geqslant1} \left(
	\mathcal{X}^{(0)}_{2n}
	\right)^{N_{2n}(\mathfrak{a})}.
\end{align}
For \(\mathcal{X}^{\mathfrak{a}} \) to be non-empty, it must hold that \(\abs{Q(\mathfrak{a})} \leqslant N_1(\mathfrak{a})\) and \(N_1(\mathfrak{a}) \equiv Q(\mathfrak{a})\pmod 2\).
We define the charge-conjugate index
\begin{align}
	\bar{\mathfrak{a}} = (L(\mathfrak{a}), -Q(\mathfrak{a}), N(\mathfrak{a}), \nu (\mathfrak{a})).
\end{align}
With this definition, an element \(\xi^{\mathfrak{a}} = (\xi^{\mathfrak{a}, \mathtt{c}}, \xi^{\mathfrak{a}, \mathtt{n}})\) consists of a charged part \(\xi^{\mathfrak{a}, \mathtt{c}}\in\mathcal{X}_{N_1(\mathfrak{a})}^{(Q(\mathfrak{a}))}\) and a neutral part \(\xi^{\mathfrak{a}, \mathtt{n}}\in \mathcal{X}^{\mathfrak{a}, \mathtt{n}}\).
Under charge conjugation, \(\bar{\xi^{\mathfrak{a}}}\in \mathcal{X}^{\bar{\mathfrak{a}}}\).

We define the modified vertex field
\begin{align}\label{eq:def-psi-a}
	\Psi^{\mathfrak{a}}(\xi^{\mathfrak{a}})
	\assign \Psi^{\mathfrak{a}, \mathtt{c}}(\xi^{\mathfrak{a}, \mathtt{c}}) \Psi^{\mathfrak{a}, \mathtt{n}} (\xi^{\mathfrak{a}, \mathtt{n}}),
\end{align}
where we define the charged and neutral parts of \(\Psi^{\mathfrak{a}}\) as
\begin{align}
	\Psi^{\mathfrak{a}, \mathtt{c}} \assign \prod_{k=1}^{N_{1}(\mathfrak{a})}\psi(\xi^{\mathfrak{a}, 1,  k}),
	\qquad
	\Psi^{\mathfrak{a}, \mathtt{n}} :=
	\prod_{(2n,k)\in I_{\mathrm n}(\mathfrak a)}
	\psi^{\nu_{2n,k}(\mathfrak{a})}_{0}(\xi^{\mathfrak{a}, 2n, k}_{1:2n}).
\end{align}
The vertex field associated to the localized cluster \(\psi^{\nu}_{0}(\xi_{1:2n})\) was defined in \eqref{eq:def-2n-pt-cluster} above.
If \(N_{2n}(\mathfrak{a})=0\), we use the convention \(\prod_{k=1}^{0} a_k = 1\) for the empty product.
In particular,
\begin{align}\label{eq:def-constant-psi}
	\Psi^{\mathfrak{a}}(\xi^{\mathfrak{a}}) \equiv 1\qquad \text{if \(N_{n}(\mathfrak{a})=0\) for all \(n\in \mathbb{N}\).}
\end{align}
Define also the following derived quantities.
\begin{itemize}
	\item The total degree of \(\mathfrak{a}\),
	      \begin{align}\label{eq:def-K}
		      K(\mathfrak{a}) \assign
		      \sum_{(2n,k)\in I_{\mathrm n}(\mathfrak a)} \nu_{2n,k}(\mathfrak{a}).
	      \end{align}
	\item The {scaling dimension} of \(\mathfrak{a}\),
	      \begin{align}\label{eq:def-scale-mfra}
		      [\mathfrak{a}] \assign \delta L(\mathfrak{a}) + \frac{1}{2} K(\mathfrak{a}) -1,
		      \qquad
		      [\mathfrak{a}]_{Q} \assign \delta L(\mathfrak{a}) + \frac{1}{2} K(\mathfrak{a}) + (1-\delta)\abs{Q(\mathfrak{a}) }- 1,
	      \end{align}
	      in terms of the heat-kernel regularization \(\dot{G}_t(x)\) which is exponentially concentrated near the set \(\{\abs{x}^{-2}\sim t\}\).
	\item The number of clusters \(c(\mathfrak{a})\), and the total number of point charges \(m(\mathfrak{a})\),
	      \begin{align}\label{eq:def-m}
		      c(\mathfrak{a}) = \sum_{n\geqslant1} N_{2n}(\mathfrak{a}),
		      \qquad
		      m(\mathfrak{a}) = \sum_{r\geqslant1} rN_{r}(\mathfrak{a})
		      = N_1(\mathfrak{a}) + \sum_{n\geqslant1}2nN_{2n}(\mathfrak{a}).
	      \end{align}
	      Upon flattening the cluster variables, \(\mathcal{X}^{\mathfrak{a}}\) is naturally identified with a subset of \(\mathcal{X}^{m(\mathfrak{a})}\).
\end{itemize}
The flow equation, as well as the additional localization and renormalization processing, will require manipulations of the indices \(\mathfrak{a}\in \mathfrak{A}\).
We define the following operations on the set of all indices.
\begin{itemize}
	\item For two indices \(\mathfrak{b}, \mathfrak{c}\in \mathfrak{A}\), define  \(\mathfrak{a}:=\mathfrak{b}+ \mathfrak{c} \in \mathfrak{A}\) as the index with
	      \begin{align}\label{eq:def-b+c}
		       & L(\mathfrak{a})   = L(\mathfrak{b}) + L(\mathfrak{c}),
		       & Q(\mathfrak{a})   = Q(\mathfrak{b}) + Q(\mathfrak{c}),          \\
		       & N_n(\mathfrak{a})   = N_n(\mathfrak{b}) + N_n(\mathfrak{c})
		       & \nu(\mathfrak{a}) = \nu(\mathfrak{b}) \sqcup \nu(\mathfrak{c}),
	      \end{align}
	      where \(\nu(\mathfrak{a}) = \nu(\mathfrak{b}) \sqcup \nu(\mathfrak{c})\) denotes the size-wise concatenation,
	      \begin{align}
		      \left(\nu_{2n,k}(\mathfrak a)\right)_{k=1}^{N_{2n}(\mathfrak a)}
		      = \left(\nu_{2n,k}(\mathfrak{b})\right)_{k=1}^{N_{2n}(\mathfrak{b})} \sqcup
		      \left(\nu_{2n,k}(\mathfrak{c})\right)_{k=1}^{N_{2n}(\mathfrak{c})}.
	      \end{align}
	\item For any \(n\in \mathbb{N}\) and \(k\in [N_{2n}(\mathfrak{a})]\), define \(\mathfrak{a}{\setminus(2n,k)}\) as \(\mathfrak{a}\) with the \(k\)-th \(2n\)-point cluster removed. More precisely,
	      for \((2n,k)\in I_{\mathrm n}(\mathfrak a)\), we define
	      \(\mathfrak b=\mathfrak a\setminus(2n,k)\) by
	      \begin{align}
		      L(\mathfrak b)   & =L(\mathfrak a),                         &
		      Q(\mathfrak b)   & =Q(\mathfrak a),                           \\
		      N_r(\mathfrak b) & =N_r(\mathfrak a)-\mathbbm 1_{\{r=2n\}}.
	      \end{align}
		  The \(k\)-th coordinate is deleted from the \(2n\)-list and the remaining \(2n\)-clusters are relabelled consecutively.
		  The localization degrees of sizes \(2r\neq2n\) remain unchanged. 
	      \begin{align}
		      \left(\nu_{2n,j}(\mathfrak b)\right)_{j=1}^{N_{2n}(\mathfrak b)}
		      :=
		      \left(
		      \nu_{2n,1}(\mathfrak a),\ldots,
		      \nu_{2n,k-1}(\mathfrak a),
		      \nu_{2n,k+1}(\mathfrak a),\ldots,
		      \nu_{2n,N_{2n}(\mathfrak a)}(\mathfrak a)
		      \right).
	      \end{align}
	\item For any multi-index \(\alpha=(\alpha_{2n,k})_{2n,k}\in [0,\infty)^{c(\mathfrak{a})}\), we define the index \(\mathfrak{a}^{\pm \alpha}\in \mathfrak{A}\) differing from \(\mathfrak{a}\) only through the localization degree via
	      \begin{align}\label{eq:def-mfra-alpha}
		      \nu_{2n,k}(\mathfrak{a}^{\pm\alpha}) = (\nu_{2n,k}(\mathfrak{a}) \pm \alpha_{2n,k})\vee 0.
	      \end{align}
\end{itemize}
For convenience, we also define the subsets of a given loop order
\begin{align}\label{eq:def-A-ell}
	\mathfrak{A}^{\ell }         := \{\mathfrak{a}\in \mathfrak{A}; \; L(\mathfrak{a})=\ell\}, \qquad
	\mathfrak{A}^{\lessgtr\ell}  := \bigcup_{\ell' \lessgtr\ell} \mathfrak{A}^{\ell'}.
\end{align}

\subsubsection{The action of the flow}\label{sec:action-flow-informal}
The flow equation acts on the basis \(\Psi^{\mathfrak{a}} \) through the functional Laplacian \(\frac{1}{2} \Delta_{\dot{G}_t}\)
and the functional derivatives in the bilinear form \(\frac{1}{2} \langle \nabla (\cdot), \dot{G}_t \nabla (\cdot) \rangle_{L^{2}} \).
They act on one of the generalized vertex fields containing neutral clusters in the following way.
\begin{itemize}
	\item The bilinear operator \(\frac{1}{2}\langle \nabla  V_{t}^{\mathfrak{b}}, \dot{G}_{t} \nabla  V_{t}^{\mathfrak{c}}\rangle\) combines two lower-loop-order indices \(\mathfrak{b},  \mathfrak{c}\) to form a new index \(\mathfrak{a}\) with  \(L(\mathfrak{a}) = L(\mathfrak{b}) + L(\mathfrak{c})\).
	      It is therefore always triangular in the loop order.
	      However, contracting between two charged contributions might lead to the creation of new neutral clusters, as observed in the simplest case in Section \ref{sec:ell=2}.
	\item The linear operator \(\frac{1}{2}\Delta_{\dot{G}_{t}}V_{t}^{\mathfrak{b}}\) acts on a single index \(\mathfrak{b}\) and produces a new index \(\mathfrak{a}\) {without changing its loop order}, i.e. \(L(\mathfrak{a})=L(\mathfrak{b})\).
	      It is therefore diagonal in the loop order.
	      Since it does not combine different configurations, it can only preserve or remove neutral clusters, and will never create them.
\end{itemize}
This suggests that the system will be triangular with respect to the partial order
\begin{align}\label{eq:partial-order}
	(L(\mathfrak{b}),N(\mathfrak{b}))\prec(L(\mathfrak{a}),N(\mathfrak{a})) :\iff \text{ \(L(\mathfrak{b}) < L(\mathfrak{a})\) or \((L(\mathfrak{b})=L(\mathfrak{a})\) and \(N(\mathfrak{b})\prec N(\mathfrak{a}))\)},
\end{align}
where we use the convention that \(N(\mathfrak{b})\preceq N(\mathfrak{a})\) if \(N_{k}(\mathfrak{b})\geqslant N_{k}(\mathfrak{a})\) for every \(k\) and \(N(\mathfrak{b})\prec N(\mathfrak{a})\) if \(N(\mathfrak{b})\preceq N(\mathfrak{a})\) but \(N(\mathfrak{a}) \neq  N(\mathfrak{b})\).
We also define
\begin{align}\label{eq:a-partial-order}
	\mathfrak{b} \prec \mathfrak{a} \iff   (L(\mathfrak{b}),N(\mathfrak{b}))\prec(L(\mathfrak{a}),N(\mathfrak{a})).
\end{align}
For the complete analysis of the action of the flow on \(\Psi^{\mathfrak{a}}\), see Section \ref{sec:action-flow-formal}.
To ensure that this triangular structure is preserved under the full flow, we have to restrict the index set to a suitable subclass.

\subsubsection{Classification of indices}\label{sec:admissible}
The following subclass of indices will be used for the inductive procedure.
\begin{definition}\label{def:Admissible}
	We define the set of \emph{admissible} indices \(\mathfrak{A}_{\tmop{adm}}\) as the subset of \(\mathfrak{A}\) such that
	\begin{enumerate}[a)]
		\item the number of point charges is bounded by the loop order, that is \(m(\mathfrak{a}) \leqslant L(\mathfrak{a})\),
		\item for any \(n, k\), it holds that
		      \begin{align}\label{eq:nu-conditions}
			      \nu_{2n,k}(\mathfrak{a})\in \{0, (\alpha_{2n}-1)\vee0,\alpha_{2n}\},
		      \end{align}
		      where, with \(\ell^{\ast}>\frac{1}{\delta_{\kappa}}\) and \(\delta_\kappa\) as defined in \eqref{eq:def-delta-kappa},
		      \begin{align}\label{eq:def-alpha-kappa}
			      \alpha_{\ell} := 2(1-\delta_\kappa \ell)
			      > \bar{\alpha}_{\ell} := 2(1-\delta \ell),
		      \end{align}
		      is the maximum localization degree required for the \(2n\)-cluster to be integrable,
		\item if \(\nu_{2n,k}(\mathfrak{a}) = \alpha_{2n}\) for some \(n,k\in \mathbb{N}\),
		      then \({\mathfrak{a}}\) consists of exactly one neutral cluster, that is \(c(\mathfrak{a})=1\), and \(N_{1}(\mathfrak{a})=0\),
		\item if \(Q(\mathfrak{a})=0\), then the charged part is empty, that is \(N_1(\mathfrak{a})=0\).
	\end{enumerate}
	We carry the grading by loop order over to \(\mathfrak{A}_{\tmop{adm} }\) and define
	\begin{align}\label{eq:A-ell-adm}
		\mathfrak{A}_{\tmop{adm} }^{\ell } := \mathfrak{A}_{\tmop{adm} } \cap \mathfrak{A}^{\ell }, \qquad \mathfrak{A}^{\gtrless \ell}_{\tmop{adm} } := \mathfrak{A}_{\tmop{adm} } \cap \mathfrak{A}^{\gtrless \ell}.
	\end{align}
	For convenience, we also define the subset of admissible indices containing no maximally localized cluster as
	\begin{align}\label{eq:A-circ}
		\mathfrak{A}_{\tmop{adm}}^{\circ} =
		\left\{
		\mathfrak{a}\in \mathfrak{A}_{\tmop{adm}}; \;
		\nu_{2n,k} (\mathfrak{a}) \neq \alpha_{2n}\text{ for every \((2n,k)\) }
		\right\},
	\end{align}
	with the obvious definition for \(\left(\mathfrak{A}^{\ell}_{\tmop{adm} } \right)^{\circ} \) etc.
\end{definition}
\begin{remark}\label{rem:admissible}
	Each of these assumptions is helpful or even required for the induction.
	\begin{enumerate} [a)]
		\item This condition will automatically be satisfied by any configuration produced by the flow equation \eqref{eq:V-ell-full-flow} and helps to ensure \(\abs{\mathfrak{A}^{\ell }_{\tmop{adm}}}<\infty\).
		\item Having a discrete set of values for the localization degree for each neutral cluster makes the bookkeeping clearer and, combined with (a), ensures that \(\abs{\mathfrak{A}_{\tmop{adm}}^{\ell} }<\infty\) for every \(\ell\in \mathbb{N}\).
		\item This is to ensure that we only produce renormalization constants compatible with \eqref{eq:int-V-cos}.
		\item This condition ensures that the charged part of every configuration is in fact charged.
		In our setup, it is necessary to preserve the triangular structure.
	\end{enumerate}
	Together, the assumptions also ensure that for any fixed loop order \(\ell\in \mathbb{N}\),
	the set of admissible indices with loop order \(\ell\) is finite, \( \abs{\mathfrak{A}_{\tmop{adm}}^{\ell}} < \infty\).
\end{remark}
Although we define \(\mathfrak{A}_{\tmop{adm}}\) at every loop order,
the construction below uses only the truncated class \(\mathfrak{A}_{\tmop{adm}}^{<\ell^\ast}\). Indices of loop order at least \(\ell^\ast\) are retained for the classification of terms contributing to the remainder.
We define \(V^{[\ell]}\) as the sum over the finite set \(\mathfrak{A}_{\tmop{adm}}^{\ell}\),
\begin{align}\label{eq:V-ell-from-V-a}
	V^{[\ell],\rho,T}_{t}[\varphi] = \sum_{\mathfrak{a}\in \mathfrak{A}^{\ell}_{\tmop{adm}}} V^{\mathfrak{a}, \rho, T}_{t}[\varphi].
\end{align}
In this sense, the index set \(\mathfrak{A}\) is a refinement of \eqref{eq:V-Mayer-exp}.

\paragraph{Relevant and irrelevant indices.}\label{sec:classification}
Within the set of admissible indices, additional processing steps are required for bounds that are uniform in the UV cut-off:
We need to make use of the localization degree \(\nu(\mathfrak{a})\) to improve the scaling of the kernel \(v^{\mathfrak{a}} \) in the UV at the cost of having a field-dependent bound on \(V^{\mathfrak{a}}\) later.

We now define some additional useful subsets of \(\mathfrak{A}_{\tmop{adm} }\) to classify the indices according to their behavior in the UV.
This behavior is measured by the scaling dimension \([\mathfrak{a}]_{Q}\) introduced in \eqref{eq:def-scale-mfra}.
We define the set \(\mathfrak{A}_{\tmop{rel}}\) of \emph{relevant} and the set \(\mathfrak{A}_{\tmop{irr}}\) of \emph{irrelevant} indices as
\begin{align}\label{eq:def-relevant-indices}
	\mathfrak{A}_{\tmop{rel}} := \{\mathfrak{a}\in \mathfrak{A}_{\tmop{adm}}; \; [\mathfrak{a}]_{Q}\leqslant 0\},\qquad \mathfrak{A}_{\tmop{irr}} := \mathfrak{A}_{\tmop{adm}}\setminus \mathfrak{A}_{\tmop{rel}}.
\end{align}
Irrelevant indices correspond to kernels decaying sufficiently fast as \(t\to \infty\) to propagate bounds backwards along the flow equation with terminal condition \(v_{T}^{\mathfrak{a}, \rho,T}=0\).
Examples include the charged two-point configuration \(v_t^{[2](\pm 2)}\) or \(v_{t}^{[2](0, \alpha)} \) for \(\alpha>\alpha_2\) in Section  \ref{sec:ell=2}.
Relevant contributions, however, generally grow as \(t \to \infty\) and require renormalization, i.e.
a specific non-zero choice for the terminal condition at \(t=T\), as seen for \(v_{t}^{[2](0,0)}\).
The renormalization has to be local and satisfy the terminal condition \eqref{eq:Vrho-T}.
In the case of the subcritical sine-Gordon model, the only required renormalization is the vacuum renormalization with a constant \(c^{\rho, T}\).
We therefore define the set of indices that contribute to the terminal condition \eqref{eq:Vrho-T} as
\begin{align}\label{eq:def-A-renormalization}
	\mathfrak{A}_{\tmop{ren}} := \{\mathfrak{a}\in \mathfrak{A}_{\tmop{rel}}; \; m(\mathfrak{a})=0\}\cup \mathfrak{A}_{\tmop{adm}}^{1}.
\end{align}  
\begin{remark}\label{rem:irrelevant}
	Let \(\mathfrak{a}\in \mathfrak{A}_{\tmop{adm} }\).
	It is easy to check that \(\mathfrak{a}\in \mathfrak{A}_{\tmop{irr}}\) whenever at least one of the following conditions holds:
	\begin{enumerate}[i)]
		\item \(\abs{Q(\mathfrak{a})}>0\) and \(L(\mathfrak{a})>1\),
		\item \(L(\mathfrak{a})\geqslant\ell^{\ast}\),
		\item \(K(\mathfrak{a}) > \bar{\alpha}_{L(\mathfrak{a})} := 2(1-\delta L(\mathfrak{a}))\).
	\end{enumerate}
	To verify condition (iii), we often use that \(\ell'<\ell'' \Rightarrow \alpha_{\ell'}>\alpha_{\ell'' }\) and, moreover, that
	\begin{align}\label{eq:two-cluster-irr}
		\alpha_{\ell'}+\alpha_{\ell'' }-2 = \alpha_{\ell' + \ell''}.
	\end{align}
\end{remark}
\subsubsection{Pairing, raising and lowering}\label{sec:pairing}
To ensure that the flow preserves admissibility of the indices when started from an admissible initial condition,
we require some additional processing maps.
\paragraph{Pairing}
When combining two clusters with opposite charge, condition (d) is not preserved.
At the level of the vertex fields, we can always write for \(\mathfrak{b}, \mathfrak{c}\in \mathfrak{A}_{\tmop{adm}}\) and \(\mathfrak{a}= \mathfrak{b}+\mathfrak{c}\) with \(N_{1}(\mathfrak{a})>0\) and \(Q(\mathfrak{a})=0\),
\begin{align}\label{eq:new-cluster-decomposition}
	\Psi^{\mathfrak{a}}(\xi^{\mathfrak{a}} )
	 & = \Psi^{\mathfrak{a}, \mathtt{c}}(\xi^{\mathfrak{a}, \mathtt{c}} )\Psi^{\mathfrak{a}, \mathtt{n}}(\xi^{\mathfrak{a}, \mathtt{n}})                                                                          \\
	 & = (\Psi^{\mathfrak{a}, \mathtt{c}}(\xi^{\mathfrak{a}, \mathtt{c}} ) - 1) \Psi^{\mathfrak{a}, \mathtt{n}}(\xi^{\mathfrak{a}, \mathtt{n}}) + \Psi^{\mathfrak{a}, \mathtt{n}}(\xi^{\mathfrak{a}, \mathtt{n}}) \\
	 & =: \Psi^{\mathfrak{a}_{\tmop{pair} }}(\xi^{\mathfrak{a}_{\tmop{pair} }})+ \Psi^{\mathfrak{a}_{\tmop{res}}}(\xi^{\mathfrak{a}_{\tmop{res}}}).
\end{align}
To make this more precise and well-defined at the level of the indices \(\mathfrak{a}\), we define the pairing map.

\begin{definition}\label{def:pairing}
	Let \(\mathfrak{P}_{\tmop{fin}}(A):=\{B\subset A; \; \abs{B}<\infty\}\) denote the set of finite subsets of \(A\).
	We define the map \(\mathcal{P}:\mathfrak{A}_{\tmop{adm} }^{\circ} + \mathfrak{A}_{\tmop{adm} }^{\circ} \to \mathfrak{P}_{\tmop{fin}}(\mathfrak{A}_{\tmop{adm}})\) as
	\begin{align}\label{eq:new-cluster}
		\mathcal{P}(\mathfrak{a})=
		\begin{cases}
			\{\mathfrak{a}_{\tmop{pair}},  \mathfrak{a}_{\tmop{res}}\}, & \text{if \(\mathfrak{a}\in\mathfrak{A}\) with \(N_{1}(\mathfrak{a})>0\) and \(Q(\mathfrak{a})=0\)}, \\
			\{\mathfrak{a}\},                                           & \text{otherwise.}
		\end{cases}
	\end{align}
	where for  \(\mathfrak{a}\in\mathfrak{A}\) with \(N_{1}(\mathfrak{a})>0\) and \(Q(\mathfrak{a})=0\), we define \(p=N_1(\mathfrak{a})\) to be the number of monopoles being combined into the new cluster,
	\begin{align}\label{eq:pair}
		\mathfrak{a}_{\tmop{res}} = \left(
		L(\mathfrak{a}), 0, N(\mathfrak{a}) - p e_1, \nu(\mathfrak{a})
		\right),\qquad
		\mathfrak{a}_{\tmop{pair} } = \left(
		L(\mathfrak{a}), 0, N(\mathfrak{a}) - p e_1 + e_p, \nu(\mathfrak{a}) \sqcup_p (0)
		\right).
	\end{align}
	Here,  \(\nu (\mathfrak{a})\sqcup_p(0)\) denotes the index obtained from concatenating \(0\) to the list of \(p\)-point clusters and \(e_n = (\delta_{n,k})_{k\in \mathbb{N}}\).
\end{definition}
With this definition, \eqref{eq:new-cluster-decomposition} holds and \(\xi^{\mathfrak{a}_{\tmop{res}}} = \xi^{\mathfrak{a}, \mathtt{n}}\).

\paragraph{Raising}\label{sec:r-l-def}
To ensure the induction later produces uniform bounds, we will rely on modifying the localization degree of the neutral clusters.
The next lemma shows that it is always possible to replace an index \(\mathfrak{a}\in \mathfrak{A}_{\tmop{rel}}\setminus \mathfrak{A}_{\tmop{ren}}\) by an index \(\tilde{\mathfrak{a}}\in \mathfrak{A}_{\tmop{irr}}\).
\begin{lemma}\label{lem:A=Aren+Airr}
	There is a map \(\mathcal{R}: \mathfrak{A}_{\tmop{adm}} \to \mathfrak{P}_{\tmop{fin}}(\mathfrak{A}_{\tmop{ren}}\cup \mathfrak{A}_{\tmop{irr}}) \), such that
	\begin{itemize}
		\item \(\mathcal{R}(\mathfrak{a})=\left\{\mathfrak{a}\right\}\) if \(\mathfrak{a}\in \mathfrak{A}_{\tmop{ren}}\cup \mathfrak{A}_{\tmop{irr}}\),
		\item every \(\tilde{\mathfrak{a}}\in \mathcal{R}(\mathfrak{a})\) differs from \(\mathfrak{a}\) only through \(\nu\), and it holds that
		      \begin{align}\label{eq:R-psi}
			      \Psi^{\mathfrak{a}} = \frac{1}{\abs{\mathcal{R}(\mathfrak{a})}} \sum_{\tilde{\mathfrak{a}}\in \mathcal{R}(\mathfrak{a})} \prod_{(2n,k)\in I_{\mathrm n}(\mathfrak a)} (\delta(\xi^{\mathfrak{a}, 2n, k}))^{\Delta^{\mathcal{R}(\mathfrak{a}, \tilde{\mathfrak{a}})}(2n,k)} \Psi^{\tilde{\mathfrak{a}}},
		      \end{align}
		      where \(\Delta^{\mathcal{R}(\mathfrak{a},\tilde{\mathfrak{a}})}(2n,k) \assign \nu_{2n, k }(\tilde{\mathfrak{a}})- \nu_{2n,k}(\mathfrak{a}) \geqslant 0\).
		\item If \(\mathfrak{a}\in \mathfrak{A}_{\tmop{rel}}\setminus \mathfrak{A}_{\tmop{ren}}\), then \(K(\tilde{\mathfrak{a}})\leqslant \alpha_2\) for any \(\tilde{\mathfrak{a}}\in \mathcal{R}(\mathfrak{a})\).
	\end{itemize}
\end{lemma}
\begin{proof}
	Since \(\mathcal{R}\) is defined to be the identity on \(\mathfrak{A}_{\tmop{ren}}\cup \mathfrak{A}_{\tmop{irr}}\),
	it remains only to give a definition for \(\mathfrak{a}\in \mathfrak{A}_{\tmop{rel}}\setminus \mathfrak{A}_{\tmop{ren} }\).
	Let
	\begin{align}\label{eq:admissible-raises}
		\mathfrak{R}(\mathfrak{a}) = \left\{
		\tilde{\mathfrak{a}}\in\mathfrak{A}_{\tmop{irr}}; \;
		\tilde{\mathfrak{a}} =  \mathfrak{a}^{+\alpha} \text{ for some \(\alpha\geqslant 0 \) }
		\right\},
	\end{align}
	be the set of admissible indices that can be obtained by raising \(\nu\) (see \eqref{eq:def-mfra-alpha} for the definition of \(\mathfrak{a}^{+\alpha} \)).
	Then \(\mathfrak{R}(\mathfrak{a})\) is non-empty.
	Indeed, if \(\mathfrak{a}\) contains only one cluster and this cluster has size \(2n\), then raising \(\nu_{2n, 1}(\mathfrak{a})\) to \(\alpha_{2n}\) ensures \(\tilde{\mathfrak{a}}\in \mathfrak{A}_{\tmop{irr}}\).
	If, on the other hand, \(\mathfrak a\) contains at least two clusters and these clusters have sizes \(2n'\) and \(2n''\), respectively, then raising these clusters to \(\alpha_{2n'}-1\) and \(\alpha_{2n''}-1\), respectively, ensures, thanks to \eqref{eq:two-cluster-irr}, that
	\begin{align}
		K(\tilde{\mathfrak{a}}) \geqslant \alpha_{2n'} + \alpha_{2n''} - 2 = \alpha_{2n'+ 2n''}\geqslant \alpha_{L(\mathfrak{a}) }.
	\end{align}
	By Remark \ref{rem:irrelevant}, this ensures \(\tilde{\mathfrak{a}} \in \mathfrak{A}_{\tmop{irr}}\).
	We define
	\begin{align}\label{eq:minimal-raises}
		\mathcal{R}(\mathfrak{a}) := \mathfrak{R}_{\tmop{min} }(\mathfrak{a}) =  \left\{
		\tilde{\mathfrak{a}}\in \mathfrak{R}(\mathfrak{a}); \; [\tilde{\mathfrak{a}}] = \min_{\mathfrak{b}\in \mathfrak{R}(\mathfrak{a})} [\mathfrak{b}]
		\right\}.
	\end{align}
	Since \(\tilde{\mathfrak{a}}\in \mathfrak{R}(\mathfrak{a}) \) differs from \(\mathfrak{a}\) only through \(\nu\), the representation \eqref{eq:R-psi} is immediate.
	It remains to verify that \(K(\tilde{\mathfrak{a}})\leqslant\alpha_{2}\) for every \(\tilde{\mathfrak{a}}\in \mathfrak{R}_{\tmop{min}}(\mathfrak{a})\).
	The computation in \eqref{eq:minimal-raises} ensures that having two clusters with \(\nu_{2n,k}(\tilde{\mathfrak{a}}) \neq 0 \) is sufficient to ensure \([\tilde{\mathfrak{a}}]>0\).
	Therefore, any \(\mathfrak{a}\in \mathfrak{A}_{\tmop{rel}}\setminus \mathfrak{A}_{\tmop{ren} }\) has at most one cluster with non-zero localization degree.
	If there is only one cluster, then we have no  choice but to raise this cluster so that  \(K(\tilde{\mathfrak{a}}) = \alpha_{2n} \leqslant \alpha_{2}\).
	Otherwise, increasing the localization degree of just one other neutral cluster ensures \(\tilde{\mathfrak{a}}\in \mathfrak{A}_{\tmop{irr} }\) while at the same time \(K(\tilde{\mathfrak{a}}) \leqslant \alpha_2\) by \eqref{eq:minimal-raises}.
\end{proof}

\paragraph{Lowering}
The functional derivative \(\tmop{D}\) appearing in the flow equation \eqref{eq:V-ell} can lower the localization degree (as seen in Lemma \ref{lem:V2-F2-estimates} and Remark \ref{rem:v2-cosine-estimate}).
We therefore also define a lowering map.
\begin{definition}\label{def:lowering-map}
	We define \(\mathcal{L}_{(n, k)}(\mathfrak{a})\) as the index differing from \(\mathfrak{a}\) only via \(\nu_{n,k}\) and such that
	\begin{align}\label{eq:def-lowering-map}
		\nu_{n, k}(\mathcal{L}_{(n, k)}(\mathfrak{a}) )=
		\begin{cases}
			(\alpha_{n}-1)\vee0 , & \text{if }\nu_{n,k}(\mathfrak{a})=\alpha_{n}, \\
			0,                    & \text{otherwise}.
		\end{cases}
	\end{align}
	The induced map for the vertex fields is defined as
	\begin{align}\label{eq:def-lowering-map-psi}
		\mathcal{L}_{(n,k)}\Psi^{\mathfrak{a}}(\xi^{\mathfrak{a}}) = \Psi^{\mathcal{L}_{(n,k)}(\mathfrak{a})}(\xi^{\mathfrak{a}}) \implies
		\mathcal{L}_{(n,k)}\Psi^{\mathfrak{a}}(\xi^{\mathfrak{a}})=  \delta(\xi^{\mathfrak{a}, n, k } )^{-\Delta_{\mathcal{L}(\mathfrak{a})}(n,k)} \Psi^{\mathfrak{a}}(\xi^{\mathfrak{a}}),
	\end{align}
	where \(\Delta_{\mathcal{L}(\mathfrak{a})}(n,k)\assign\nu_{n,k}(\mathcal{L}_{(n,k)}(\mathfrak{a}))-\nu_{n,k}(\mathfrak{a}) \leqslant 0\).
	For every \(\mathfrak{a}\in \mathfrak{A}_{\tmop{adm}}\), we define
	\begin{align}\label{eq:def-KL}
		K_{{\mathcal{L}}}(\mathfrak{a})
		=\max_{\substack{n\in\{1\}\cup2\mathbb N \\k\in[N_n(\mathfrak a)]}}
		K(\mathcal{L}_{(n,k)}(\mathfrak{a})),
	\end{align}
	with the convention that \(K_{\mathcal{L}}(\mathfrak{a}) = K(\mathfrak{a}) = 0\) whenever \(\mathfrak{a}\in \mathfrak{A}_{\tmop{ren}}\).
	For convenience, we also extend the lowering map to the monopoles in the trivial way, \(\mathcal{L}_{(1,k)}(\mathfrak{a})=\mathfrak{a}\) for every \(\mathfrak{a}\in \mathfrak{A}_{\tmop{adm}}\).
\end{definition}
\begin{remark}
	Here \(K_{\mathcal{L}}(\mathfrak{a}) \) is the worst possible scaling degree after applying the lowering map to \(\mathfrak{a}\).
	It is important to track \(K_{\mathcal{L}}\) for the precise dependence of the estimates for the potential and force on the field \(\varphi\) later in Section \ref{sec:pot-for-estimates}.
	Note that, since \(\mathcal{L}_{(1,k)}(\mathfrak{a})=\mathfrak{a}\), we have \(K_{\mathcal{L}}(\mathfrak{a})=K(\mathfrak{a})\) whenever \(Q(\mathfrak{a}) \neq 0 \).
\end{remark}
\begin{remark}\label{rem:lowering-req}
	The observations in Remark \ref{rem:v2-cosine-estimate} suggest that the lowering operation is not just an optimization but in fact required to obtain the correct counterterm.
	Since the clusters can be removed through the action of the flow equation \eqref{eq:V-ell-full-flow},
	the remaining singularity will have to be absorbed fully by the heat kernel \(\dot{G}\).
	Whenever the localization degree is larger than \(1\), this requires symmetry of the associated kernel that in general only holds when \(c(\mathfrak{a})=1\); see also the model computation in Section \ref{sec:ell=2}.
\end{remark}

\subsubsection{The space of kernels}\label{sec:kernels}
We define the analytic setting for the coefficient kernels \(v^{\mathfrak{a}} \) in \eqref{eq:def-Va}.
For a fixed index \(\mathfrak{a}\in \mathfrak{A}\) and \(t\geqslant0\), define
\({\mathcal{K}}_{t}^{\mathfrak{a}}\) as the Banach space of all real-valued functions \(v: \mathcal{X}^{\mathfrak{a}} \to \mathbb{R}\) that have finite norm
\begin{align}\label{eq:def-triple-norm}
	\interleave v \interleave_{\mathfrak{a}, t} \assign
	\begin{cases}
		\sup_ {k\in [m(\mathfrak{a})]} \sup_{\xi_k \in \mathcal{X}}  \int_{\mathcal{X}^{\mathfrak{a} \setminus k}}
		\mathd (\xi^{\mathfrak{a} \setminus k})  | v (\xi^{\mathfrak{a}}) |
		\omega_t^{\mathfrak{a}} (\xi^{\mathfrak{a}}), \qquad
		\quad m(\mathfrak{a}) \neq 0, \\
		\abs{v}, \qquad m(\mathfrak{a}) = 0.
	\end{cases}
\end{align}

Here, \(\mathcal{X}^{\mathfrak{a}\setminus\{k\}}\) denotes the space \(\mathcal{X}^{\mathfrak{a}}\) with the \(k\)-th point charge removed,
and, analogously, we define \(\xi^{\mathfrak{a} \setminus k} := (\xi^{\mathfrak{a}}_j)_{j \neq k}\).
The weights \(\omega_t^{\mathfrak{a}} (\xi^{\mathfrak{a}})\)
denote Steiner tree weights
\begin{align}\label{eq:def-exp-Steiner}
	\omega_t^{\mathfrak{a}} (x_{1 : m(\mathfrak{a})}) = \omega_t^{r(\mathfrak{a})} (x_{1 : m(\mathfrak{a})})= \exp (r(\mathfrak{a}) t^{\frac{1}{2}}\tmop{St} (x_{1 : m(\mathfrak{a})})),
\end{align}
where \(r: \mathfrak{A}  \to [0,\infty)\) is any strictly decreasing function with respect to the partial order \eqref{eq:partial-order} on \(\mathfrak{A}\).
For concreteness, we fix
\begin{align}\label{eq:def-r-steiner-rate}
	r(\mathfrak{a}) \assign 6^{-L(\mathfrak{a})}\left(1+\frac{1}{2}\Theta(N(\mathfrak{a})) \right),
\end{align}
with
\begin{align}\label{eq:def-N-order-function}
	\Theta(N(\mathfrak{a})) \assign \frac{\theta(N(\mathfrak{a}))}{1+\theta(N(\mathfrak{a}))}, \text{\, and \, } \theta(N(\mathfrak{a})) \assign \sum_{k\geqslant1} N_k(\mathfrak{a}) 2^{-k}.
\end{align}
We also define \(\tnorm{\cdot}_{\mathfrak{a}}\) to be the unweighted norm \eqref{eq:def-triple-norm} with \(\omega \equiv 1\). The modified kernel norm \(\tnorm{v}_{\mathfrak a,t;r}\) denotes the norm
\eqref{eq:def-triple-norm} with the Steiner weight \(\omega_t^r=\exp(r\sqrt t\,\tmop{St})\).
It is sometimes convenient to have a notation for the graded kernel space
\begin{align}
	\mathcal{K}_t \assign \bigoplus_{\mathfrak a\in\mathfrak A}\mathcal{K}_t^{\mathfrak{a}}.
\end{align}

For \(v\in\mathcal K_t^{\mathfrak a}\), define its charge conjugate
\(\mathsf {C}_{\mathfrak a}v\in\mathcal K_t^{\bar{\mathfrak a}}\) by
\begin{align}\label{eq:def-kernel-charge-conjugation}
	(\mathsf {C}_{\mathfrak a}v)(\xi^{\bar{\mathfrak a}})
	\assign v(\bar\xi^{\bar{\mathfrak a}}).
\end{align}
Charge conjugation preserves all spatial coordinates and gives a bijection
\(\mathcal X^{\bar{\mathfrak a}}\to\mathcal X^{\mathfrak a}\), so that \(\tnorm{\mathsf {C}_{\mathfrak a}v}_{\bar{\mathfrak a},t}=\tnorm{v}_{\mathfrak a,t}\).
We say a family \((v^{\mathfrak{a}})_{\mathfrak{a}\in \mathfrak{A}_{\tmop{adm}}}\)
is compatible with charge conjugation if \(v^{\bar{\mathfrak{a}}} = \mathsf{C}_{\mathfrak{a}} v^{\mathfrak{a}}\) for all \(\mathfrak{a}\in \mathfrak{A}_{\tmop{adm}}\).
If \(\xi^{\mathfrak{a}}=\varnothing\), then the norms are all equivalent 
and for any \(r_1 , r_2\), we have  \(\tnorm{\cdot}_{\mathfrak{a}, t; r_1} = \tnorm{\cdot}_{\mathfrak{a}, t; r_2}=\abs{\cdot}\).

\begin{remark}
	The norms \(\tnorm{\cdot}_{\mathfrak{a}}\) are chosen to be compatible with the convolutional structure of the flow equation \eqref{eq:v-ell}.
	The exponential Steiner weights \(\omega\) quantify the concentration of the kernels \(v^{\mathfrak{a}} \) near the diagonal.
	The specific choice in \eqref{eq:def-r-steiner-rate} is for technical convenience only and any strictly decreasing rate would work.
	The decreasing rate is convenient to absorb the gain of the raising operation using some of the exponential weight along the induction.
\end{remark}

The symmetry under charge conjugation has the following simple but important consequence.
\begin{lemma}\label{lem:charge-moment}
	Fix \(t\geqslant0\) and suppose that \((v_t^{\mathfrak{b}})_{\mathfrak{b}\in \mathfrak{A}_{\tmop{adm}}}\) is compatible with charge conjugation.
	For any index \(\mathfrak{a}\in \mathfrak{A}_{\tmop{adm}}\) with \(Q(\mathfrak{a})=0\),
	\begin{align}\label{eq:charge-moment}
		\int \rho(\mathd {\xi^{\mathfrak{a}}}) v_{t}^{\mathfrak{a}}(\xi^{\mathfrak{a}}) S_{\xi^{\mathfrak{a}}}(\tmop{id}) = 0.
	\end{align}
\end{lemma}
\begin{proof}
	Since \(Q(\mathfrak{a}) = 0\), compatibility under charge conjugation gives \(v_t^{\mathfrak{a}} (\xi^{\mathfrak{a}}) = v_t^{\mathfrak{a}} (\bar{\xi}^{\mathfrak{a}} )\), while \(S_{\xi^{\mathfrak{a}}}=-S_{\bar{\xi}^{\mathfrak{a}}}\). Hence,
	\begin{align}\label{eq:v-xi-symmetric}
		\int \rho(\mathd {\xi^{\mathfrak{a}}}) v_{t}^{\mathfrak{a}}(\xi^{\mathfrak{a}}) S_{\xi^{\mathfrak{a}}}(\tmop{id}) = - \int \rho(\mathd {\xi^{\mathfrak{a}}}) v_{t}^{\mathfrak{a}}(\xi^{\mathfrak{a}}) S_{\xi^{\mathfrak{a}}}(\tmop{id}) = 0,
	\end{align}
	which directly implies \eqref{eq:charge-moment}.
\end{proof}

Next we define the dual action of the pairing and raising maps \(\mathsf{P}, \mathsf{R}: \mathcal{K}_{t} \to \mathcal{K}_t\).

\begin{definition}\label{def:dual-P-kernels}
	Let
	\(\mathfrak a\in
	\mathfrak A_{\tmop{adm}}
	\cup(\mathfrak A_{\tmop{adm}}^{\circ}
	+\mathfrak A_{\tmop{adm}}^{\circ})\)
	and let \(v\in\mathcal K_t^{\mathfrak a}\) be a kernel.
	For \(\mathfrak{a}\in \mathfrak{A}_{\tmop{adm}}\), define \(\mathsf{P}_{\mathfrak{a} \to \mathfrak{a}} v = v\).
	For
	\(\mathfrak{a}\in
	(\mathfrak{A}_{\tmop{adm}}^{\circ}
	+\mathfrak{A}_{\tmop{adm}}^{\circ})
	\setminus\mathfrak{A}_{\tmop{adm}}\),
	we define, with
	\(\mathfrak{a}_{\tmop{pair}}\) and
	\(\mathfrak{a}_{\tmop{res}}\) as in \eqref{eq:pair},
	\begin{align}
		\bigl(\mathsf{P}_{\mathfrak{a} \to \mathfrak{a}_{\tmop{pair}}}v\bigr)
		(\xi^{\mathfrak{a}_{\tmop{pair}}})
		 & \assign v(\xi^{\mathfrak{a}}),                                                                                              \\
		\bigl(\mathsf{P}_{\mathfrak{a} \to \mathfrak{a}_{\tmop{res}}}v\bigr)
		(\xi^{\mathfrak{a}_{\tmop{res}}})
		 & \assign \int \rho(\mathd \xi^{\mathfrak{a}, \mathtt{c}}) v(\xi^{\mathfrak{a}, \mathtt{c}}, \xi^{\mathfrak{a}, \mathtt{n}}).
	\end{align}
	In the first line, we identified the spaces \(\mathcal{X}^{\mathfrak{a}} \sim \mathcal{X}^{\mathfrak{a}_{\tmop{pair}}}\) by unpacking the newly formed neutral cluster, as in \eqref{eq:new-cluster-decomposition}. 
	For all other combinations of indices, we set \(\mathsf{P}_{\mathfrak{a} \to \mathfrak{b}} v = 0\).
\end{definition}

\begin{definition}\label{def:dual-R-kernels}
	For \(\mathfrak a\in\mathfrak A_{\tmop{adm}}\), \(\tilde{\mathfrak a}\in\mathcal R(\mathfrak a)\), and \(v\in \mathcal{K}_t^{\mathfrak{a}}\), define the raising map
	\begin{align}\label{eq:def-dual-raising-kernels}
		\big(\mathsf R_{\mathfrak a\to\tilde{\mathfrak a}}v\big)(\xi^{\tilde{\mathfrak a}})
		\assign \frac{1}{\abs{\mathcal R(\mathfrak a)}}
		\prod_{(2n,k)\in I_{\mathrm n}(\mathfrak a)}
		\delta(\xi^{\tilde{\mathfrak a},2n,k})^{\Delta^{\mathcal R(\mathfrak a,\tilde{\mathfrak a})}(2n,k)} v(\xi^{\mathfrak a}).
	\end{align}
	Otherwise, \(\mathsf R_{\mathfrak b\to \mathfrak{c}}v=0\) whenever \(\mathfrak{b}\not\in \mathfrak{A}_{\tmop{adm}}\) or \(\mathfrak{c}\not\in\mathcal R(\mathfrak b)\).
\end{definition}

We note some immediate consequences of these definitions.
\begin{lemma}\label{lem:dual-P-R-reconstruction}
	For \(t\geqslant0\), the pairing operations satisfy for any \(\mathfrak{a}\in \mathfrak{A}_{\tmop{adm}}^{\circ} + \mathfrak{A}_{\tmop{adm}}^{\circ}\) and \(v\in \mathcal{K}_t^{\mathfrak{a}}\),
	\begin{align}\label{eq:dual-pairing-reconstruction}
		\int \rho(\mathd \xi^{\mathfrak{a}}) v(\xi^{\mathfrak{a}})\Psi^{\mathfrak{a}}(\xi^{\mathfrak{a}})
		= \sum_{\mathfrak{b}\in \mathfrak{A}_{\tmop{adm}}} \int \rho(\mathd \xi^{\mathfrak{b}}) (\mathsf{P}_{\mathfrak{a} \to \mathfrak{b}} v)(\xi^{\mathfrak{b}}) \Psi^{\mathfrak{b}}(\xi^{\mathfrak{b}}).
	\end{align}
	Similarly, for every \(\mathfrak{a}\in \mathfrak{A}_{\tmop{adm}}\) and \(v\in \mathcal{K}_t^{\mathfrak{a}} \),
	\begin{align}\label{eq:dual-raising-reconstruction}
		\int \rho(\mathd \xi^{\mathfrak{a}}) v(\xi^{\mathfrak{a}})\Psi^{\mathfrak{a}}(\xi^{\mathfrak{a}})
		= \sum_{\mathfrak{b}\in \mathfrak{A}_{\tmop{adm}}} \int \rho(\mathd \xi^{\mathfrak{b}}) (\mathsf{R}_{\mathfrak{a} \to \mathfrak{b}} v)(\xi^{\mathfrak{b}}) \Psi^{\mathfrak{b}}(\xi^{\mathfrak{b}}).
	\end{align}
\end{lemma}
\begin{proof}
	Note that in both cases there are only finitely many non-zero contributions to the sum.
	The identities follow readily from the definition of \(\Psi^{\mathfrak{a}}\) and the maps \(\mathsf{P}, \mathsf{R}\).
\end{proof}

We note the following estimates for future reference.
\begin{lemma}\label{lem:dual-R-kernel-estimate}
	Let \(\mathfrak a\in\mathfrak A_{\tmop{adm}}\), \(\tilde{\mathfrak a}\in\mathcal R(\mathfrak a)\),
	and \(r, \varepsilon > 0\).
	For any \(t>0 \) and \(v\in\mathcal K_t^{\mathfrak a}\),
	\begin{align}\label{eq:dual-R-kernel-estimate}
		\tnorm{\mathsf{R}_{\mathfrak{a} \to \tilde{\mathfrak{a}}} v}_{\tilde{\mathfrak{a}}, t; r}
		\lesssim (\varepsilon \sqrt{t})^{-(K(\tilde{\mathfrak{a}})- K(\mathfrak{a}))} \tnorm{v}_{\mathfrak{a}, t; r+\varepsilon}.
	\end{align}
\end{lemma}
\begin{proof}
	Since
	\begin{align}\label{eq:est-steiner}
		\delta(\xi_I)^\gamma
		\lesssim_{\abs {I},\gamma}\tmop{St}(\xi_I)^\gamma
		\lesssim_{\gamma,r}t^{-\gamma/2}\omega_t^r(\xi_I),
		\qquad \gamma,r>0,
	\end{align}
	we have
	\begin{align}
		\omega_t^r(\xi^{\mathfrak a})
		\prod_{(2n,k)\in I_{\mathrm n}(\mathfrak a)}
		\delta(\xi^{\mathfrak a,2n,k})^{
				\Delta^{\mathcal R(\mathfrak a,\tilde{\mathfrak a})}(2n,k)}
		\lesssim
		(\varepsilon\sqrt t)^{-(K(\tilde{\mathfrak a})-K(\mathfrak a))}
		\omega_t^{r+\varepsilon}(\xi^{\mathfrak a}).
	\end{align}
	Inserting this into the definition \eqref{eq:def-dual-raising-kernels} yields \eqref{eq:dual-R-kernel-estimate} after taking the norm.
\end{proof}
\begin{lemma}\label{lem:dual-pairing-estimate}
	Let \(\mathfrak{a}\in \mathfrak{A}^{\circ}_{\tmop{adm}} + \mathfrak{A}_{\tmop{adm}}^{\circ}\) and \(\mathfrak b\in\mathfrak A_{\tmop{adm}}\).
	For any \(t, r\geqslant0\) and \(v\in\mathcal K_t^{\mathfrak a}\),
	\begin{align}
		\tnorm{\mathsf{P}_{\mathfrak{a} \to \mathfrak{b}}v}_{\mathfrak{b}, t; r} \lesssim \tnorm{v}_{\mathfrak{a}, t; r}.
	\end{align}
\end{lemma}
\begin{proof}
	Immediate from the definition.
\end{proof}

\subsection{The flow equation for the coefficients}\label{sec:flow-coeff}
With the notation fixed, the definition \eqref{eq:def-Va} is precise.
We now turn to the analysis of the flow equation in this refined basis.

We first derive the flow equation for the coefficients \(v^{\mathfrak{a}} \) of \(V^{\mathfrak{a}}\) from \eqref{eq:V-ell} as defined by \eqref{eq:V-ell-from-V-a} and then state the bounds required for Theorem \ref{thm:int-Va-estimates} on the coefficients \(v^{\mathfrak{a}}\) in Proposition \ref{prop:A-B-estimates} and Corollary \ref{cor:va-scaling}.
The bulk of the estimates is deferred to Section \ref{sec:details-flow}.

Translating the initial condition \(V^{[1], \rho, T}\) from \eqref{eq:V-initial} fixes the initial condition for the indices \(\mathfrak{a}\in \mathfrak{A}^{1}_{\tmop{adm}}\) to be
\begin{align}\label{eq:va-init}
	v^{\mathfrak{a}}_{t}(\xi^{\mathfrak{a}})=
	\begin{cases}
		\frac{\lambda_t}{2}, & \abs{Q(\mathfrak{a})} = 1, \\
		0,                   & \text{otherwise.}
	\end{cases}
	\qquad t\in  [0, T].
\end{align}
Since the right-hand side depends only on \(\abs{Q(\mathfrak a)}\), the initial coefficient family \eqref{eq:va-init} is compatible with charge conjugation.

We now want to find functions \(\mathbb{A}^{\mathfrak{b}}_{\mathfrak{a}}\) and \(\mathbb{B}^{\mathfrak{b}, \mathfrak{c}}_{\mathfrak{a}}\)
such that if \((v^{\mathfrak{a}})_{\mathfrak{a}\in \mathfrak{A}_{\tmop{adm}}^{1}}\) is given by the initial condition \eqref{eq:va-init} and for \(L(\mathfrak{a})>1\),
the kernels \((v^{\mathfrak{a}})_{\mathfrak{a}\in \mathfrak{A}_{\tmop{adm}}^{<\ell^\ast} }\) satisfy the system
\begin{equation}\label{eq:sG-va-full-flow}
	v_t^{\mathfrak{a}} (\xi^{\mathfrak{a}}) = \left\{
	\begin{array}{lc}
		\int_t^{T} \Gamma_{t, s}^{\mathfrak a}
		(\xi^{\mathfrak{a}}) (\sum_{\mathfrak{b}}\mathbb{A}^{\mathfrak{b}}_{\mathfrak{a}}
		(v_s^{\mathfrak{b}}, \dot{G}_s) - \sum_{\mathfrak{b}, \mathfrak{c}}\mathbb{B}^{\mathfrak{b},
				\mathfrak{c}}_{\mathfrak{a}} (v_s^{\mathfrak{b}}, v_s^{\mathfrak{c}},
		\dot{G}_s)) \mathd s,                                          & \mathfrak{a} \in \mathfrak{A}_{\tmop{irr}}, \\
		- \int_0^t
		(\sum_{\mathfrak{b}} \mathbb{A}^{\mathfrak{b}}_{\mathfrak{a}} (v_s^{\mathfrak{b}},
		\dot{G}_s) - \sum_{\mathfrak{b}, \mathfrak{c}}\mathbb{B}^{\mathfrak{b}, \mathfrak{c}}_{\mathfrak{a}}
		(v_s^{\mathfrak{b}}, v_s^{\mathfrak{c}}, \dot{G}_s)) \mathd s, &
		\mathfrak{a} \in \mathfrak{A}_{\tmop{ren}},                                                                  \\
		0,                                                             & \text{otherwise}
	\end{array} \right.
	,
\end{equation}
then the functions \((V^{[\ell]})_{\ell < \ell^\ast}\) as defined in \eqref{eq:V-ell-from-V-a} satisfy the flow equations \eqref{eq:V-ell}.
To ensure that the system can be solved inductively,
we need \(\mathbb{A}\) and \(\mathbb{B}\) to respect the partial order \eqref{eq:partial-order}.
More precisely, our definitions allow us to define \(\mathbb{A}, \mathbb{B}\) such that the only contributions in \eqref{eq:sG-va-full-flow} come from
\begin{align}\label{eq:A,B-triangular}
	\mathbb{A}^{\mathfrak{b}}_{\mathfrak{a}} \neq 0 \Rightarrow
	\begin{cases}
		L(\mathfrak{a}) = L(\mathfrak{b}),     \\
		Q(\mathfrak{a}) = Q(\mathfrak{b}),     \\
		N(\mathfrak{a}) \succ N(\mathfrak{b}), \\
		\mathfrak{a} \succ \mathfrak{b},
	\end{cases}
	\qquad
	\mathbb{B}^{\mathfrak{b},\mathfrak{c}}_{\mathfrak{a}}  \neq 0
	\Rightarrow
	\begin{cases}
		L(\mathfrak{a}) = L(\mathfrak{b}) + L(\mathfrak{c}), \\
		Q(\mathfrak{a}) = Q(\mathfrak{b}) + Q(\mathfrak{c}),
		\\
		\mathfrak{a} \succ \mathfrak{b}, \mathfrak{c}.
	\end{cases}
\end{align}
The main result of this section is the following.
\begin{proposition}\label{prop:A-B-estimates}
	Fix \(\ell^\ast \in \mathbb{N}\) such that \(\ell^{\ast}\delta_\kappa>1\).
	For any \(\mathfrak{a},\mathfrak{b}, \mathfrak{c}\in \mathfrak{A}_{\tmop{adm}}^{<\ell^\ast}\)  there are functions
	\begin{align}
		\mathbb{A}^{\mathfrak{b}}_{\mathfrak{a}} (\cdot, \dot{G}_t) :
		\mathcal{K}^{\mathfrak{b}}_t \rightarrow \mathcal{K}_t^{\mathfrak{a}}, \qquad
		\mathbb{B}^{\mathfrak{b}, \mathfrak{c}}_{\mathfrak{a}} (\cdot, \cdot,
		\dot{G}_t) : \mathcal{K}_t^{\mathfrak{b}} \times \mathcal{K}_t^{\mathfrak{c}}
		\rightarrow \mathcal{K}_t^{\mathfrak{a}}, \qquad \Gamma_{t,
			s}^{\mathfrak{a}} : \mathcal{X}^{\mathfrak{a}} \rightarrow (0, \infty).
	\end{align}
	satisfying the following properties.
	\begin{enumerateroman}
		\item If \((v^{\mathfrak{a}}_{t})_{\mathfrak{a}\in \mathfrak{A}}\) satisfies \eqref{eq:va-init} and, for \(L(\mathfrak a)>1\), \eqref{eq:sG-va-full-flow},
		then \((V^{[\ell]})_{\ell<\ell^{\ast}}\) defined according to \eqref{eq:V-ell-from-V-a} satisfies \eqref{eq:V-ell}.
		\item The maps \(\mathbb{A}^{\mathfrak{b}}_{\mathfrak{a}}\) are linear, while the maps \(\mathbb{B}^{\mathfrak{b}, \mathfrak{c}}_{\mathfrak{a}}\) are bilinear; condition \eqref{eq:A,B-triangular} holds; and for each \(\mathfrak{a}\in \mathfrak{A}_{\tmop{adm} }\) there are only finitely many non-zero \(\mathfrak{b}, \mathfrak{c}\in \mathfrak{A}_{\tmop{adm}}\) such that \(\mathbb{B}^{\mathfrak{b}, \mathfrak{c}}_{\mathfrak{a}}\) or \(\mathbb{A}^{\mathfrak{b}}_{\mathfrak{a}}\) is non-zero.
		Moreover, if \(\mathbb{A}^{\mathfrak{b}}_{\mathfrak{a}}\) or \(\mathbb{B}^{\mathfrak{b}, \mathfrak{c}}_{\mathfrak{a}} \) are non-zero, then \(\mathfrak{a}\in \mathfrak{A}_{\tmop{irr}}\cup \mathfrak{A}_{\tmop{ren}}\).
		\item Charge conjugation is preserved. More precisely, whenever
		\((v^{\mathfrak d})_{\mathfrak d\in\mathfrak A_{\tmop{adm}}^{<\ell^\ast}}\)
		is compatible with charge conjugation, then
		\begin{align}\label{eq:A,B-charge-conjugation}
			\mathsf C_{\mathfrak a}\mathbb A_{\mathfrak a}^{\mathfrak b}
			(v^{\mathfrak b},\dot G_t)
			 & =\mathbb A_{\bar{\mathfrak a}}^{\bar{\mathfrak b}}
			(v^{\bar{\mathfrak b}},\dot G_t),                                       \\
			\mathsf C_{\mathfrak a}\mathbb B_{\mathfrak a}^{\mathfrak b,\mathfrak c}
			(v^{\mathfrak b},v^{\mathfrak c},\dot G_t)
			 & =\mathbb B_{\bar{\mathfrak a}}^{\bar{\mathfrak b},\bar{\mathfrak c}}
			(v^{\bar{\mathfrak b}},v^{\bar{\mathfrak c}},\dot G_t),                 \\
			\Gamma_{t,s}^{\bar{\mathfrak a}}(\xi^{\bar{\mathfrak a}})
			 & =\Gamma_{t,s}^{\mathfrak a}(\bar\xi^{\bar{\mathfrak a}}).
		\end{align}
		In particular, the right-hand side of \eqref{eq:sG-va-full-flow} is compatible with charge conjugation.
		\item Whenever \(s\geqslant t\geqslant0\),
		\((v_t^{\mathfrak d})_{\mathfrak d\in\mathfrak A_{\tmop{adm}}^{<\ell^\ast}}\)
		is compatible with charge conjugation,
		\(v_t^{\mathfrak{b}}\in \mathcal{K}^{\mathfrak{b}}_{t}\), and \(v_t^{\mathfrak{c}}\in\mathcal{K}^{\mathfrak{c}}_{t}\),
		\begin{align}\label{eq:sG-A,B-est}
			\begin{split}
				\interleave
				\mathbb{A}^{\mathfrak{b}}_{\mathfrak{a}} (v_t^{\mathfrak{b}},
				\dot{G}_t) \interleave_{\mathfrak{a}, t}                                                     & \lesssim \langle t \rangle^{[
						\mathfrak{b} ] - [\mathfrak{a}] - 1} \interleave
				v^{\mathfrak{b}}_t \interleave_{\mathfrak{b}, t},                                                                            \\
				\interleave \mathbb{B}^{\mathfrak{b},
						\mathfrak{c}}_{\mathfrak{a}} (v_t^{\mathfrak{b}}, v_t^{\mathfrak{c}},
				\dot{G}_t) \interleave_{\mathfrak{a}, t}                                                     & \lesssim \langle t \rangle^{[
						\mathfrak{b} ] + [ \mathfrak{c} ] -[\mathfrak{a}] - 1} \interleave v^{\mathfrak{b}}_t \interleave_{\mathfrak{b}, t}
				\interleave v^{\mathfrak{c}}_t \interleave_{\mathfrak{c}, t},                                                                \\
				\| \Gamma_{t,s}^{\mathfrak a} (\xi^{\mathfrak{a}}) \|_{L^{\infty}(\mathd x^{\mathfrak{a}} )} & \lesssim
				\langle t \rangle^{[ \mathfrak{a} ]_{Q} -
				[ \mathfrak{a} ]} \langle s \rangle^{- ([
						\mathfrak{a} ]_{Q} - [ \mathfrak{a} ])}.
			\end{split}
		\end{align}
	\end{enumerateroman}
\end{proposition}
\begin{proof}
	The functions \(\mathbb{A}\) and \(\mathbb{B}\) are defined in Section \ref{sec:def-A,B}.
	The first claim \textit{i)} is the result of Lemma \ref{lem:va-flow-Va-flow}, while  \textit{ii)} follows from Lemmas \ref{lem:B-reconstruction} and \ref{lem:A-reconstruction}.
	The charge-conjugation identities for \(\mathbb A\) and \(\mathbb B\) in \textit{iii)} are proved in Lemmas \ref{lem:A-estimates} and \ref{lem:B-estimate}, respectively.
	The identity for \(\Gamma\) follows directly from \eqref{eq:def-Gamma-a}, since \(\mathbb G_{s,t}(\xi^{\bar{\mathfrak a}})=\mathbb G_{s,t}(\bar\xi^{\bar{\mathfrak a}})\).
	Regarding \textit{iv)}, the estimates on \(\mathbb{A}\) are proven in Lemma \ref{lem:A-estimates}, and on \(\mathbb{B}\) in Lemma \ref{lem:B-estimate}.
	The estimate on \(\Gamma\) is Lemma \ref{lem:gamma-est}.
\end{proof}
As an immediate consequence of Proposition \ref{prop:A-B-estimates}, we can prove the estimates, uniform in \(T\geqslant0\), on \((v^{\mathfrak{a}}_{t})_{\mathfrak{a}\in \mathfrak{A}_{\tmop{adm}}^{< \ell^\ast} }\).
\begin{corollary}\label{cor:va-scaling}
	Let \(\rho\prec 1\), \(T>0\).
	Suppose that \((v^{\mathfrak{a}}_{t})_{t\in [0,T], \mathfrak{a}\in \mathfrak{A}_{\tmop{adm}}^{< \ell^{\ast}}}\) satisfies \eqref{eq:va-init} and, for \(L(\mathfrak a)>1\), \eqref{eq:sG-va-full-flow}.
	Then the family is compatible with charge conjugation for every \(t\in[0,T]\).
	Moreover, uniformly in \(T\geqslant0\),
	\begin{align}\label{eq:va-scaling}
		\tnorm{v^{\mathfrak{a}, \rho, T}_{t}}_{\mathfrak{a}, t} \lesssim_{\rho} \bar{\lambda}^{L(\mathfrak{a})} \langle t \rangle^{-[\mathfrak{a}]}.
	\end{align}
	Moreover, there exists an \(\varepsilon=\varepsilon(\beta^2)>0\) such that for any \(T_1, T_2 >0\),
	\begin{align}\label{eq:va-T-convergence}
		\tnorm{v^{\mathfrak{a}, \rho, T_1 }_{t} - v_{t}^{\mathfrak{a}, \rho, T_2 } }_{\mathfrak{a},t }\lesssim \bar{\lambda}^{L(\mathfrak{a})} \langle t \rangle ^{-[\mathfrak{a}]+ \varepsilon}  \langle T_1 \wedge T_2   \rangle ^{-\varepsilon}, \qquad t\leqslant T_1\wedge T_2.
	\end{align}
\end{corollary}
\begin{proof}
	The proof is by induction over \((L(\mathfrak{a}), N(\mathfrak{a}))\) according to the partial order defined in \eqref{eq:partial-order}.
	We start with all indices of loop order \(1\), whose kernels are specified directly by \eqref{eq:va-init} and compatible with charge conjugation.
	The bound \eqref{eq:va-scaling} follows directly from the bound \(\bar{\lambda}_t \lesssim \bar\lambda\langle t\rangle^{1-\delta} = \bar\lambda\langle t\rangle^{-[\mathfrak{a}]}\).

	We carry out the induction simultaneously for the charge-conjugate pair \(\{\mathfrak a,\bar{\mathfrak a}\}\), which is possible because charge conjugation preserves \((L(\mathfrak a),N(\mathfrak a))\).
	Now let \(\mathfrak{a}\in \mathfrak{A}_{\tmop{adm}}^{< \ell^\ast}\) with \(L(\mathfrak a)>1\), and assume that \eqref{eq:va-scaling} holds and that charge-conjugation compatibility has been established for all preceding indices.
	The sets \(\mathfrak A_{\tmop{irr}}\) and \(\mathfrak A_{\tmop{ren}}\) are invariant under charge conjugation.
	Consequently, part \textit{iii)} of Proposition \ref{prop:A-B-estimates}, inserted into \eqref{eq:sG-va-full-flow}, shows that
	\(v_t^{\bar{\mathfrak a}}=\mathsf C_{\mathfrak a}v_t^{\mathfrak a}\).
	By \eqref{eq:A,B-triangular}, every non-zero linear contribution satisfies \(L(\mathfrak{b})=L(\mathfrak{a})\), whereas every non-zero bilinear contribution satisfies \(L(\mathfrak{b})+L(\mathfrak{c})=L(\mathfrak{a})\).
	Hence, by the induction hypothesis and \eqref{eq:sG-A,B-est},
	\begin{align}
		\tnorm{\mathbb{A}_{\mathfrak{a}}^{\mathfrak{b}}(v_{s}^{\mathfrak{b}}, \dot{G}_s)}_{\mathfrak{a}, s}
		+
		 \tnorm{\mathbb{B}_{\mathfrak{a}}^{\mathfrak{b}, \mathfrak{c}}(v_{s}^{\mathfrak{b}}, v_{s}^{\mathfrak{c}}, \dot{G}_s)}_{\mathfrak{a},s }
		 & \lesssim_{\rho} \bar{\lambda}^{L(\mathfrak{a})}\langle s\rangle^{-1-[\mathfrak{a}]}.
	\end{align}
	Since only finitely many source indices contribute for each \(\mathfrak{a}\), the same bound holds for their sums in \eqref{eq:sG-va-full-flow}.
	For any \(\mathfrak{a}\) with non-zero source term and \([\mathfrak{a}]_{Q} \leqslant 0\),
	we must have \([\mathfrak{a}] = [\mathfrak{a}]_{Q}\) and \(\mathfrak{a}\in \mathfrak{A}_{\tmop{ren}}\).
	Inserting the bounds,
	\begin{align}
		\tnorm{v_{t}^{\mathfrak{a}}}_{\mathfrak{a}, t}
		= \abs{v_{t}^{\mathfrak{a}}}
		\lesssim_{\rho} \bar{\lambda}^{L(\mathfrak{a})} \int_0^t \langle s\rangle^{-1-[\mathfrak{a}]}  \mathd {s}
		\lesssim_{\rho} \bar{\lambda}^{L(\mathfrak{a})} \langle t \rangle^{-[\mathfrak{a}]}.
	\end{align}
	If \([\mathfrak{a}]_{Q} > 0\), then following the same estimates as above using that \(\tnorm{\cdot}_{t} \leqslant \tnorm{\cdot}_{s}\) for \(s\geqslant t\), we have
	\begin{align}
		\tnorm{v_{t}^{\mathfrak{a}}}_{\mathfrak{a}, t}
		\lesssim        & \int_t^T  \norm{\Gamma_{t,s}(\xi^{\mathfrak{a}})}_{L^\infty(\mathd x^{\mathfrak{a}})}
		\left(
		\tnorm{\mathbb{A}_{\mathfrak{a}}^{\mathfrak{b}}(v_{s}^{\mathfrak{b}}, \dot{G}_s)}_{\mathfrak{a}, s}
		+
		\tnorm{\mathbb{B}_{\mathfrak{a}}^{\mathfrak{b}, \mathfrak{c}}(v_{s}^{\mathfrak{b}}, v_{s}^{\mathfrak{c}}, \dot{G}_s)}_{\mathfrak{a},s }
		\right)\mathd {s}                                                                                                                                                           \\
		\lesssim_{\rho} & \bar{\lambda}^{L(\mathfrak{a})} \langle t\rangle^{-([\mathfrak{a}]-[\mathfrak{a}]_{Q})} \int_{t}^{T} \langle s\rangle^{-1-[\mathfrak{a}]_{Q}}  \mathd {s} \\
		\lesssim_{\rho} & \bar{\lambda}^{L(\mathfrak{a})} \langle t\rangle^{-[\mathfrak{a}]},
	\end{align}
	as required.
	The convergence follows similarly by induction.
	For \(L(\mathfrak{a})= 1\), \(v^{\mathfrak{a},\rho, T_1} - v^{\mathfrak{a}, \rho, T_2 } =0\).
	Assume now that \eqref{eq:va-T-convergence} holds for any \(\mathfrak{b} \prec \mathfrak{a}\).
	Define the source term for \(T\in [0,\infty)\) and \(s\in [0,T]\),
	\begin{align}
		\mathcal S^{\mathfrak a, T}_s
		\assign \sum_{\mathfrak b} \mathbb A_{\mathfrak a}^{\mathfrak b}(v_s^{\mathfrak b,\rho,T},\dot G_s)
		- \sum_{\mathfrak b,\mathfrak c}\mathbb B_{\mathfrak a}^{\mathfrak b,\mathfrak c}	(v_s^{\mathfrak b,\rho,T},v_s^{\mathfrak c,\rho,T},\dot G_s).
	\end{align}
	By the induction hypothesis, together with the linearity of \(\mathbb{A}\) and the bilinearity of \(\mathbb{B}\),
	\begin{align}
		\tnorm{\mathcal S^{\mathfrak a, T_1 }_s - \mathcal S^{\mathfrak{a}, T_2 }_s}_{\mathfrak{a}, s}
		 & \lesssim \bar{\lambda}^{L(\mathfrak{a})} \langle s \rangle^{-1-[\mathfrak{a}]}
		\left( \frac{\langle s \rangle }{\langle T_1\wedge T_2 \rangle }\right)^{\varepsilon}.
	\end{align}
	Inserting this estimate into the equations implied by \eqref{eq:sG-va-full-flow} for \( \delta v^{\mathfrak{a}} = v^{\mathfrak{a},\rho, T_1} - v^{\mathfrak{a}, \rho, T_2}\) closes the induction.
\end{proof}

\subsection{Estimates on the effective potential and the force}\label{sec:pot-for-estimates}
The estimates on the kernels \((v^{\mathfrak{a}}_{t})_{\mathfrak{a}\in \mathfrak{A}_{\tmop{adm}}^{<\ell^{\ast}}}\) from Corollary \ref{cor:va-scaling} transfer directly to estimates on the (truncated) effective potential \(V^{[<\ell^{\ast}]}_{t}\), the force \(F^{[<\ell^{\ast}]}_{t}=-\nabla V_{t}^{[<\ell^{\ast}]}\) as well as their remainders \(\mathcal{H}^{<\ell^\ast}_{t}\) (as defined in \eqref{eq:int-def-Hcal}) and \(H^{[<\ell^{\ast}]}=-\nabla  \mathcal{H}^{[<\ell^{\ast}]}\).
To track the total degree \(K(\mathfrak{a})\) along the flow, define also the gradient degrees \(\tmop{deg}_{\nabla}^{V}, \tmop{deg}_{\nabla}^{F}\),
\begin{align}\label{eq:def-deg-V-F}
	\tmop{deg}_{\nabla}^{V}(\mathfrak{a}) \assign
	\begin{cases}
		K(\mathfrak{a}), & \text{if \(\nu_{2n,k}(\mathfrak{a})\leqslant 1\) for all \(n,k\)} \\
		2,               & \text{if \(\nu_{2n,k}(\mathfrak{a})>1\) for some \(n,k\)},
	\end{cases}
	\qquad
	\tmop{deg}_{\nabla}^{F}(\mathfrak{a}) \assign K_{\mathcal{L}}(\mathfrak{a}).
\end{align}
Recall the notation for the (semi-)norms defined in \eqref{eq:def-fbsde-norms} and \eqref{eq:def-fbsde-norms-2}.

\begin{proposition}\label{prop:Va-general-estimates}
	Let \(\beta^{2}\in [0,8\pi)\), \(\gamma<\kappa\) and \(\delta_\kappa>0\) according to \eqref{eq:def-delta-kappa}.
	Let \(E\in\{L^\infty(\chi), L^2(\chi)\}\).
	Let \((v^{\mathfrak{a}}_{t})_{\mathfrak{a}\in \mathfrak{A}_{\tmop{adm}}^{{<\ell^{\ast}}}}\) be as in Corollary \ref{cor:va-scaling} and \((V^{\mathfrak{a}}_{t})_{\mathfrak{a}\in \mathfrak{A}_{\tmop{adm}}^{<\ell^{\ast}}}\) be as in \eqref{eq:def-Va}. Extend \(V^{\mathfrak{a}}\) to the full real half-line by setting \(V^{\mathfrak{a}}_{t}=V_{T}^{\mathfrak{a}}\) for \(t\geqslant T\).
	Define \(\mathcal{Y}_{\gamma, D}:=\{h\in \mathcal{S}'(\mathbb{R}^{2});\; [h]_{\gamma, 1}\leqslant D\}\).
	Then,
	\begin{enumerate}[a)]
		\item \textbf{Compatibility.} For each \(T<\infty\), Assumption \ref{hyp:scale-interpolation} holds.
		\item \textbf{Support properties.}
		      For any \(t,T>0\), it holds that \(V_{t}^{\rho, T} = V_{t\wedge T}^{\rho ,T}\). Moreover, \(V^{\rho, T}\) depends only on the restriction of \(\varphi \) to \(\tmop{supp}(\rho )\).
		\item \textbf{Uniform bounds.} For all \(t\in [0,T]\),
		      \begin{align}\label{eq:V,F-uniform-bounds}
			      \sup_{T\geqslant 0}\abs{V^{\mathfrak{a}, \rho, T}_{t}[\varphi]}
			      \lesssim_{\rho} \bar{\lambda}^{L(\mathfrak{a})} \langle t  \rangle ^{1 - \delta L(\mathfrak{a})} [\varphi]_{0, 1, t}^{\tmop{deg }_{\nabla}^{V}(\mathfrak{a})}, \\
			      \sup_{T\geqslant 0}\norm{Q_{t}F^{\mathfrak{a}, \rho, T}_{t}[\varphi]}_{E}
			      \lesssim_{\rho} \bar{\lambda}^{L(\mathfrak{a})} \langle t  \rangle ^{- \delta L(\mathfrak{a})} [\varphi]_{0, 1, t}^{\tmop{deg }_{\nabla}^{F}(\mathfrak{a})}.
		      \end{align}
		\item \textbf{Local Lipschitz continuity.}
		      For every \(D>0\), \(\gamma<\kappa\)  and \(\varphi, \tilde{\varphi}\in \mathcal{Y}_{\gamma, D}\) the following holds:
		      \begin{align}\label{eq:V,F-local-Lipschitz}
			      \sup_{T\geqslant 0}\abs{V^{\mathfrak{a}, \rho, T}_{t}[\varphi]-V^{\mathfrak{a}, \rho, T}_{t}[\tilde{\varphi}]}
			      \lesssim_{\rho}\bar{\lambda}^{L(\mathfrak{a})} \langle t  \rangle ^{1-\delta_{\kappa} L(\mathfrak{a})} D^{\tmop{deg}_{\nabla}^{V}(\mathfrak{a})} \norm{\varphi-\tilde{\varphi}}_{t, \gamma}, \\
			      \sup_{T\geqslant 0}\norm{Q_{t}F^{\mathfrak{a}, \rho, T}_{t}[\varphi]-Q_{t}F^{\mathfrak{a}, \rho, T}_{t}[\tilde{\varphi}]}_{E}
			      \lesssim_{\rho} \bar{\lambda}^{L(\mathfrak{a})} \langle t  \rangle ^{-\delta_{\kappa} L(\mathfrak{a})} D^{\tmop{deg}_{\nabla}^{F}(\mathfrak{a})} \norm{\varphi-\tilde{\varphi}}_{t, \gamma}.
		      \end{align}
		\item \textbf{Convergence in \(T\).} Let \(T_2 > T_1 \geqslant 0\). There is an \(\varepsilon=\varepsilon(\beta^2)>0\) such that for any \(\varphi\in \mathcal{Y}_{\gamma, D}\)
		      \begin{align}\label{eq:V,F-T-convergence}
			      \abs{V^{\mathfrak{a}, \rho, T_1 }_{t}[\varphi]-V^{\mathfrak{a}, \rho, T_2 }_{t}[\varphi]}
			      \lesssim_{\rho} \bar{\lambda}^{L(\mathfrak{a})} \langle T_1 \rangle^{-\varepsilon} \langle t \rangle^{1-\delta_{\kappa} L(\mathfrak{a}) + \varepsilon} D^{\tmop{deg}_{\nabla}^{V}(\mathfrak{a}) } , \\
			      \,\norm{Q_{t}F^{\mathfrak{a}, \rho, T_1 }_{t}[\varphi]- Q_{t}F^{\mathfrak{a}, \rho, T_2 }_{t}[\varphi]}_{E}
			      \lesssim_{\rho} \bar{\lambda}^{L(\mathfrak{a})} \langle T_1 \rangle^{-\varepsilon} \langle t \rangle^{-\delta_{\kappa} L(\mathfrak{a}) + \varepsilon} D^{\tmop{deg}_{\nabla}^{F}(\mathfrak{a}) }.
		      \end{align}
	\end{enumerate}
\end{proposition}
Before we prove this result, we note some immediate consequences.
First, thanks to the convergence in \(T\), we immediately obtain the limiting objects as the regularization is removed.
\begin{corollary}\label{cor:V,F-limit}
	There is an \(\varepsilon=\varepsilon(\beta^{2})>0\) such that, for every \(\mathfrak{a}\in\mathfrak{A}_{\tmop{adm}}^{<\ell^{\ast}}\), there are unique \(V_{t}^{\mathfrak{a},\rho}=V_{t}^{\mathfrak{a},\rho,\infty}\) and \(F_{t}^{\mathfrak{a},\rho}\) for which the following bounds hold for every \(\varphi\in\mathcal{Y}_{\gamma,D}\) as \(T\to\infty\):
	\begin{align}
		\abs{V^{\mathfrak{a}, \rho, T}_{t}[\varphi] - V^{\mathfrak{a}, \rho}_{t}[\varphi]}
		 & \lesssim \langle T \rangle^{-\varepsilon} \bar{\lambda}^{L(\mathfrak{a})} \langle t \rangle^{1-L(\mathfrak{a})\delta_{\kappa}+\varepsilon}  D^{\tmop{deg}_{\nabla}^{V}(\mathfrak{a})}
		,                                                                                                                                                                                        \\
		\norm{Q_{t}(F^{\mathfrak{a}, \rho, T}_{t}[\varphi] - F^{\mathfrak{a}, \rho}_{t}[\varphi])}_{L^\infty(\chi )}
		 & \lesssim \langle T \rangle^{-\varepsilon} \bar{\lambda}^{L(\mathfrak{a})} \langle t \rangle^{-L(\mathfrak{a})\delta_{\kappa}+\varepsilon}  D^{\tmop{deg}_{\nabla}^{F}(\mathfrak{a})}.
	\end{align}
	The limits satisfy the estimates in Proposition \ref{prop:Va-general-estimates}.
\end{corollary}

By computing the remainders \(\mathcal{H}, H\),
we see that for \(\delta, \kappa>0\) according to \eqref{eq:def-delta-kappa},
the approximate functions \(V,F\) as defined in Corollary \ref{cor:combined-V,F-estimate} are indeed good approximate solutions to \eqref{eq:V-ell-full-flow}.
By analogy with \eqref{eq:def-deg-V-F}, define
\begin{align}\label{eq:def-deg-H}
	\tmop{deg}_{\nabla}^{\mathcal{H}} =  2\max_{\mathfrak{a}\in \mathfrak{A}_{\tmop{adm}}^{<\ell^\ast}} \tmop{deg}_{\nabla}^{F}(\mathfrak{a}).
\end{align}
\begin{corollary}\label{cor:H-estimates}
	In the notation and under the assumptions of Proposition \ref{prop:Va-general-estimates}, the functions \(\mathcal{H}, H\) defined in \eqref{eq:int-def-Hcal} satisfy the following estimates.
	Let \(\gamma< \kappa\) and \(\varphi, \tilde{\varphi}\in \mathcal{Y}_{\gamma , D}\).
	\begin{enumerate}[a)]
		\item \textbf{Support.}  \(\mathcal{H}^{\rho, T}\) depends only on the restriction of \(\varphi \) to \(\tmop{supp}(\rho )\).
		\item \textbf{Uniform bounds.}
		      For any \(0\leqslant t<T\leqslant\infty\), the following holds:
		      \begin{align}\label{eq:H-uniform-bound}
			      \sup_{T\geqslant0} \abs{\mathcal{H}_{t}^{\rho, T}[\varphi]}
			      \lesssim_{\rho} \bar{\lambda}^{\ell^\ast}(1+\bar{\lambda})^{\ell^\ast-2} \langle t  \rangle ^{-\delta_{\kappa} \ell^\ast} D^{\tmop{deg}_{\nabla}^{\mathcal{H}}}, \\
			      \sup_{T\geqslant0} \norm{H_{t}^{\rho, T} [\varphi]}_{L^\infty(\chi)}
			      \lesssim_{\rho} \bar{\lambda}^{\ell^\ast}(1+\bar{\lambda})^{\ell^\ast-2} \langle t  \rangle ^{-\delta_{\kappa} \ell^\ast} D^{\tmop{deg}_{\nabla}^{\mathcal{H}}}.
		      \end{align}
		\item \textbf{Local Lipschitz continuity.}
		      It holds that
		      \begin{align}\label{eq:H-loc-Lipschitz}
			      \sup_{T\geqslant0}\abs{\mathcal{H}^{\rho, T}_{t}[\varphi]-\mathcal{H}^{\rho, T}_{t}[\tilde{\varphi} ]}
			      \lesssim_{\rho}\bar{\lambda}^{\ell^\ast}(1+\bar{\lambda})^{\ell^\ast-2} \langle t  \rangle ^{-\delta_{\kappa}\ell^\ast} D^{\tmop{deg}_{\nabla}^{\mathcal{H}}} \norm{\varphi-\tilde{\varphi}}_{t, \gamma}, \\
			      \sup_{T\geqslant0}\norm{H^{\rho, T}_{t}(\varphi)-H^{\rho, T}_{t}(\tilde{\varphi} )}_{L^\infty(\chi)}
			      \lesssim_{\rho} \bar{\lambda}^{\ell^\ast}(1+\bar{\lambda})^{\ell^\ast-2}\langle t  \rangle ^{-\delta_{\kappa}\ell^\ast} D^{\tmop{deg}_{\nabla}^{\mathcal{H}}} \norm{\varphi-\tilde{\varphi}}_{t, \gamma}.
		      \end{align}
		\item \textbf{Convergence in \(T\).}
		      There is an \(\varepsilon= \varepsilon(\beta^{2})>0\) such that
		      \begin{align}\label{eq:H-T-convergence}
			      \abs{\mathcal{H}^{\rho, T_1 }_{t}(\varphi)-\mathcal{H}^{\rho, T_2 }_{t}(\varphi)}
			      \lesssim_{\rho} \bar{\lambda}^{\ell^\ast}(1+\bar{\lambda})^{\ell^\ast-2} \langle T_1 \rangle^{-\varepsilon} \langle t \rangle^{-\delta_{\kappa} \ell^\ast + \varepsilon} D^{\tmop{deg}_{\nabla}^{\mathcal{H}}} , \\
			      \norm{H^{\rho, T_1 }_{t}(\varphi)- H^{\rho, T_2 }_{t}(\varphi)}_{L^\infty(\chi)}
			      \lesssim_{\rho} \bar{\lambda}^{\ell^\ast}(1+\bar{\lambda})^{\ell^\ast-2} \langle T_1 \rangle^{-\varepsilon} \langle t \rangle^{-\delta_{\kappa} \ell^\ast + \varepsilon} D^{\tmop{deg}_{\nabla}^{\mathcal{H}}}.
		      \end{align}
	\end{enumerate}
\end{corollary}
\begin{proof}
	These estimates follow directly from the estimates on \(F\) obtained in Proposition \ref{prop:Va-general-estimates} combined with the representation \eqref{eq:H-truncated} for \(\mathcal{H}\) and the definition of \(H= -\nabla\mathcal{H}\).
	Indeed,
	\begin{align}\label{eq:HV-estimates}
		\abs{\mathcal{H}_{t}^{\rho, T}  [\varphi]}
		 & \leqslant \frac{1}{2} \sum_{\substack{\mathfrak{b}, \mathfrak{c}\in \mathfrak{A}_{\tmop{adm}}^{<\ell^{\ast}}:                                                                                                                                                                    \\ L(\mathfrak{b})+ L(\mathfrak{c})\geqslant\ell^{\ast}}}
		\abs{\langle F_{t}^{\mathfrak{b}, \rho, T}[\varphi], \dot{G}_{t} F_{t}^{\mathfrak{c}, \rho, T}[\varphi]\rangle_{L^{2}}}                                                                                                                                                             \\
		 & \leqslant \frac{1}{2} \sum_{\substack{\mathfrak{b}, \mathfrak{c}\in \mathfrak{A}_{\tmop{adm}}^{<\ell^{\ast}}                                                                                                                                                                     \\
				L(\mathfrak{b})+ L(\mathfrak{c})\geqslant\ell^{\ast}}}
		\norm{ Q_t F_{t}^{\mathfrak{b}, \rho, T}[\varphi]}_{L^2}  \norm{ Q_t F_{t}^{\mathfrak{c}, \rho, T}[\varphi]}_{L^2}                                                                                                                                                                  \\
		 & \lesssim_{\rho} \frac{1}{2} \sum_{\substack{\mathfrak{b}, \mathfrak{c}\in \mathfrak{A}_{\tmop{adm}}^{<\ell^{\ast}}:                                                                                                                                                              \\ L(\mathfrak{b})+ L(\mathfrak{c})\geqslant\ell^{\ast}}}
		\bar{\lambda}^{L(\mathfrak{b})+L(\mathfrak{c})} \langle t \rangle^{-\delta(L(\mathfrak{b})+L(\mathfrak{c}))} \left(\langle t  \rangle^{-\frac{1}{2}} \norm{\nabla \varphi}_{L^\infty(\chi)} \right)^{\tmop{deg}_{\nabla}^{F}(\mathfrak{b}) + \tmop{deg}_{\nabla}^{F}(\mathfrak{c})} \\
		 & \lesssim_{\rho} \bar{\lambda}^{\ell^\ast}(1+\bar{\lambda})^{\ell^\ast-2} \langle t \rangle^{-\delta\ell^{\ast}} \left(\langle t  \rangle^{-\frac{1}{2}} \norm{\nabla \varphi}_{L^\infty(\chi)} \right)^{\tmop{deg}_{\nabla}^{\mathcal{H}}},
	\end{align}
	as required.
	The remaining estimates follow analogously.
\end{proof}

Combining all contributions \(\mathfrak{a}\in \mathfrak{A}_{\tmop{adm}}^{<\ell^{\ast}}\), we obtain estimates on the approximate potential and force \(V\) and \(F\).
\begin{corollary}\label{cor:combined-V,F-estimate}
	Let \(\rho\prec1\) and \(T\in[0,\infty]\).
	Let \(\ell^{\ast} > \frac{1}{\delta }\) and define
	\begin{align}
		K_V(\beta^{2}) & := \max_{\mathfrak{a}\in \mathfrak{A}_{\tmop{adm}}^{<\ell^{\ast}}}
		\tmop{deg}_{\nabla}^{V}(\mathfrak{a}), \qquad
		K_F(\beta^{2}) := \max_{\mathfrak{a}\in \mathfrak{A}_{\tmop{adm}}^{<\ell^{\ast}}}
		\tmop{deg}_{\nabla}^{F}(\mathfrak{a}).
	\end{align}
	Then, the scale interpolations \((V^{\rho, T}_{t})_{t\in [0,T]}\) and \((F_{t}^{\rho, T})_{t\in[0,T]} \)
	satisfy the estimates in Theorem \ref{thm:int-Va-estimates}.
\end{corollary}

\begin{proof}[Proof of Proposition \ref{prop:Va-general-estimates}]
	Whenever no ambiguities arise, we suppress the dependence on \(\rho, T\) and write e.g. \(V^{\mathfrak{a}}= V^{\mathfrak{a}, \rho, T}\) etc.
	We include only the first estimate in full detail.
	The remaining estimates are similar, and we only explain the required modifications compared to the uniform bounds in c).
	\begin{enumerate}[a)]
		\item Compatibility is the content of Proposition \ref{prop:T-dependent-est}.
		\item The support property follows directly from the definition of \(V^{\mathfrak{a}}\) in \eqref{eq:def-Va}.
		\item Starting from the definition \eqref{eq:def-Va}, we distinguish the two cases \(\nu_{2n,k}(\mathfrak{a})\leqslant1\) for all \(n,k\) and \(\nu_{2n, k}(\mathfrak{a})>1\) for some \(n,k\).
		      \begin{enumerate}[I)]
			      \item If \(\nu_{2n,k}(\mathfrak{a})>1\), then by Definition \ref{def:Admissible}, the index \(\mathfrak{a}\) corresponds to a single neutral cluster.
			            Therefore, using \(\delta(\xi^{\mathfrak{a}})\lesssim\tmop{St}(\xi^{\mathfrak{a}})\) together with \eqref{eq:def-Steiner-weight} to convert the spatial regularity of \(\varphi\) into a gain in scaling, we compute
			            \begin{align}\label{eq:Va-bdd-single-cluster}
				            \abs{V_{t}^{\mathfrak{a}}[\varphi]}
				            \leqslant       & \int \rho(\mathd {\xi_{1:2n}^{\mathfrak{a}}}) \abs{v^{\mathfrak{a}}_{t}(\xi_{1:2n})} \delta(\xi_{1:2n}^{\mathfrak{a}})^{2-\nu_{2n}(\mathfrak{a})} \abs{\frac{\cos(\beta S_{\xi_{1:2n}^{\mathfrak{a}} }(\varphi))-1}{\delta(\xi_{1:2n}^{\mathfrak{a}})^{2}} }                                         \\
				            \leqslant       & \langle t \rangle ^{-1+\frac{\nu_{2n}(\mathfrak{a})}{2}} \int \rho(\mathd {\xi_{1:2n}^{\mathfrak{a}}}) \abs{v^{\mathfrak{a}}_{t}(\xi_{1:2n})\omega_t^{\mathfrak{a}}(\xi^{\mathfrak{a}}) } \abs{\frac{\cos(\beta S_{\xi_{1:2n}^{\mathfrak{a}} }(\varphi))-1}{\delta(\xi_{1:2n}^{\mathfrak{a}})^{2}} } \\
				            \lesssim_{\rho} & \tnorm{v_{t}^{\mathfrak{a}}}_{\mathfrak{a}, t} \langle t\rangle^{-1+\frac{K(\mathfrak{a})}{2}} \norm{\nabla \varphi}_{L^\infty(\chi)}^{2}                                                                                                                                                            \\
				            \lesssim        & \bar{\lambda}^{L(\mathfrak{a})} \langle t \rangle^{1-\delta L(\mathfrak{a})} [\varphi]_{0,1,t}^{\tmop{deg}_{\nabla}^{V}(\mathfrak{a})}.
			            \end{align}

			            Similarly, for \(F\), we compute using Lemma \ref{lem:hk-taylor-remainder} applied to \(S^1_{\xi^{\mathfrak{a}}}(Q_t(x-\cdot))\), Young's convolution inequality, and the estimates on \(v^{\mathfrak{a}}\) from Corollary \ref{cor:va-scaling}:
			            \begin{align}\label{eq:Fa-bdd-single-cluster}
				                      & \norm{Q_t F_{t}^{\mathfrak{a}}[\varphi]}_{L^\infty(\chi)}                                                                                                                                                                                                                                                              \\
				            \leqslant & \sup_{x} \chi(x) \int \rho(\mathd \xi^{\mathfrak{a}}) \abs{v_{t}^{\mathfrak{a}}(\xi^{\mathfrak{a}})} \abs{S_{\xi^{\mathfrak{a}}}^{1} Q_t(x-\cdot)} \abs{\frac{\sin (\beta S_{\xi^{\mathfrak{a}}}(\varphi))}{\delta(\xi^{\mathfrak{a}})^{K(\mathfrak{a})-1}}}                                                           \\
				            \lesssim  & \sup_{x}\chi(x)\mathe^{-\frac{m^{2}}{2t}} \int \rho(\mathd \xi^{\mathfrak{a}}) \mathe^{-\frac{1}{2}r(\mathfrak{a})\sqrt{t} \abs{x-x_1^{\mathfrak{a}}}} t^{\frac{1}{2}} \omega^{\mathfrak{a}}_t(\xi^{\mathfrak{a}}) \abs{v_t^{\mathfrak{a}}(\xi^{\mathfrak{a}})} \norm{\nabla \varphi}_{L^{\infty}}^{K(\mathfrak{a})-1} \\
				            \lesssim  & t^{\frac{1}{2}} \mathe^{-\frac{m^{2}}{2t}} \norm{\nabla \varphi}_{L^{\infty}}^{K(\mathfrak{a})-1} \tnorm{v_{t}^{\mathfrak{a}}}_{\mathfrak{a}, t} \int_{\mathbb{R}^{2}}\mathd x \mathe^{- \frac{1}{2}r(\mathfrak{a}) \sqrt{t} \abs{x}}                                                                                  \\
				            \lesssim  & t^{-\frac{1}{2}}\mathe^{-\frac{m^{2}}{2t}}  \norm{\nabla \varphi}_{L^{\infty}}^{K(\mathfrak{a})-1} \tnorm{v_{t}^{\mathfrak{a}}}_{\mathfrak{a}, t}                                                                                                                                                                      \\
				            \lesssim  & \bar{\lambda}^{L(\mathfrak{a})} \langle t \rangle^{-\delta L(\mathfrak{a})} [\varphi]_{0, 1, t}^{\tmop{deg}_{\nabla}^{F}(\mathfrak{a})}.
			            \end{align}
			            For the \(L^{2}(\chi)\) estimate, we argue in exactly the same way.
			      \item If \(\nu_{2n,k}(\mathfrak{a})\leqslant1\) for all \(n,k\), then Lemma \ref{lem:2n-estimate} gives
			            \begin{align}\label{eq:Va-bdd-more-clusters}
				            \abs{V_{t}^{\mathfrak{a}}[\varphi]}
				            &\leqslant      \int \rho(\mathd {\xi^{\mathfrak{a}}}) \abs{v^{\mathfrak{a}}_{t}(\xi^{\mathfrak{a}})}\abs{\Psi^{\mathfrak{a}}(\xi^{\mathfrak{a}} )}     \\
				            &\lesssim_{\rho}\tnorm{v_{t}^{\mathfrak{a}} }_{t,\mathfrak{a}} \norm{\nabla\varphi}_{L^\infty(\chi)}^{K(\mathfrak{a})}                                  \\
				            &\lesssim_{\phantom{\rho}}      \bar{\lambda}^{L(\mathfrak{a})} \langle t \rangle^{1-\delta L(\mathfrak{a})} [\varphi]_{0,1,t}^{\tmop{deg}_{\nabla}^{V}(\mathfrak{a})}.
			            \end{align}
			            Similarly, combining the formula from Lemma \ref{lem:D-Psi-branches} with Lemma \ref{lem:hk-taylor-remainder} and proceeding as in \eqref{eq:Fa-bdd-single-cluster}, we obtain
			            \begin{align}\label{eq:Fa-bdd-more-clusters}
				                      & \norm{Q_t F_{t}^{\mathfrak{a}}[\varphi]}_{L^\infty(\chi)}                                                                                                                                                                                                                                                                                                          \\
				            \leqslant & \sup_{x} \chi(x) \int \rho(\mathd \xi^{\mathfrak{a}}) \abs{v_{t}^{\mathfrak{a}}(\xi^{\mathfrak{a}} )} \sum_{(2n,k)\in I_{\mathrm n}(\mathfrak a)}\abs{S^{\nu_{2n, k}(\mathfrak{a})}_{\xi^{\mathfrak{a},2n,k}} Q_t(x-\cdot)} \abs{\mathcal{L}_{(2n, k)} \Psi^{\mathfrak{a}}}                                                                                        \\
				            \lesssim  & \sup_{x} \int \frac{\rho}{\chi}(\mathd \xi^{\mathfrak{a}}) \abs{v_{t}^{\mathfrak{a}}(\xi^{\mathfrak{a}} )} \sum_{(2n,k)\in I_{\mathrm n}(\mathfrak a)}\delta(\xi^{\mathfrak{a}, 2n, k})^{1-\nu_{2n, k}(\mathfrak{a})} \abs{S^{1}_{\xi^{\mathfrak{a},2n,k}} (\chi^{-1}Q_t)(x-\cdot)} \norm{\nabla \varphi}_{L^\infty(\chi)}^{K(\mathcal{L}_{(2n,k)}(\mathfrak{a}))} \\
				            \lesssim  & \tnorm{v_{t}^{\mathfrak{a}}}_{\mathfrak{a}, t}\mathe^{-\frac{m^{2}}{2t}} \sum_{(2n,k)\in I_{\mathrm n}(\mathfrak a)} t^{\frac{1}{2}-(\frac{1}{2}+\frac{\nu_{2n, k}(\mathfrak{a})}{2})}    \norm{\nabla \varphi}_{L^\infty(\chi)}^{K(\mathcal{L}_{(2n,k)}(\mathfrak{a}))} \int_{\mathbb{R}^{2}} \mathe^{- \frac{1}{2} \sqrt{t} r(\mathfrak{a})\abs{x}} \mathd x     \\
				            \lesssim  & \bar{\lambda}^{L(\mathfrak{a})} \langle t \rangle^{-\delta L(\mathfrak{a})} [\varphi]_{0, 1, t}^{\tmop{deg}_{\nabla}^{F}(\mathfrak{a})}.
			            \end{align}
		      \end{enumerate}
		\item To obtain the uniform-in-\(T\) Lipschitz estimates, we follow the same steps as above with the only difference that \(\Psi^{\mathfrak{a}} \) is replaced by \(\Psi^{\mathfrak{a}} - \tilde{\Psi}^{\mathfrak{a}}\).
		      Here, \(\tilde{\Psi}\) is defined in the obvious way, that is according to \eqref{eq:def-psi-a} with \(\tilde{\varphi}\) in place of \(\varphi\).
		      If \(\nu_{2n, k}(\mathfrak{a})\geqslant1\), then \(\mathfrak{a}\) contains only a single cluster.
		      We directly compute with \(\cos a-\cos b=-2\sin\!\big(\tfrac{a+b}{2}\big)\sin\!\big(\tfrac{a-b}{2}\big)\) and \(|\sin \eta|\leqslant |\eta|\) with \(\varphi\in \mathcal{Y}_{\gamma, D}\) on \(\tmop{supp}(\rho)\),
		      \begin{align}\label{eq:cos-Lipschitz}
			      \big|\cos(\beta S_{\xi^{\mathfrak{a}, 2n, k}}(\varphi))-\cos(\beta S_{\xi^{\mathfrak{a}, 2n, k}}(\tilde{\varphi}))\big|
			      \leqslant
			       & \frac{\beta^{2}}{2} |S_{\xi^{\mathfrak{a}, 2n, k}}({\varphi}- \tilde{\varphi})| \left(
			      \abs{S_{\xi^{\mathfrak{a}, 2n, k}}(\varphi)} + \abs{S_{\xi^{\mathfrak{a}, 2n, k}}(\tilde{\varphi})}
			      \right)                                                                                                       \\
			      \lesssim
			       & \beta^{2} |S_{\xi^{\mathfrak{a}, 2n, k}}({\varphi}- \tilde{\varphi})| \delta(\xi^{\mathfrak{a}, 2n, k}) D.
		      \end{align}
		      For the difference, we compute:
		      \begin{align}\label{eq:cos-Lipschitz-diff}
			      \rho(\xi^{\mathfrak{a}, 2n, k})\abs{S_{\xi^{\mathfrak{a}, 2n, k}}(\varphi - \tilde{\varphi})}
			      \leqslant \norm{\chi^{-1}\rho}_{L^\infty} \delta(\xi^{\mathfrak{a}, 2n, k})\norm{\nabla \varphi - \nabla \tilde{\varphi}}_{L^\infty(\chi)}.
		      \end{align}
		      Combining these estimates, we find
		      \begin{align}\label{eq:cos-Lipschitz-combined}
			      \rho(\xi^{\mathfrak{a}, 2n, k} )\abs{\frac{\cos(\beta S_{\xi^{\mathfrak{a}, 2n, k}}(\varphi))-\cos(\beta S_{\xi^{\mathfrak{a}, 2n, k}}(\tilde{\varphi}))}{\delta(\xi^{\mathfrak{a}, 2n, k})^{\nu_{2n,k}(\mathfrak{a})}}}
			      \lesssim_{\rho} \delta(\xi^{\mathfrak{a}, 2n, k})^{2-\nu_{2n,k}(\mathfrak{a})} D\norm{\nabla \varphi- \nabla \tilde{\varphi}}_{L^\infty(\chi)},
		      \end{align}
		      which, combined with the same estimates as in c), yields the claim.
		      Otherwise, \(\nu_{2n,k}(\mathfrak{a})\leqslant1\) for all \(n, k\in \mathbb{N}\).
		      We estimate
		      \begin{align}\label{eq:Psi-tilde-Psi}
			      \abs{\Psi^{\mathfrak{a}}-\tilde{\Psi ^{\mathfrak{a}} }}
			      =         & \abs{\prod_{k=1}^{N_1(\mathfrak{a})} \Psi^{\mathfrak{a}, 1, k} \prod_{n=1}^{\frac{L(\mathfrak{a})}{2}} \prod_{k=1}^{N_{2n}(\mathfrak{a})} \Psi^{\mathfrak{a}, 2n, k }
			      - \prod_{k=1}^{N_1(\mathfrak{a})} \tilde{\Psi}^{\mathfrak{a}, 1, k} \prod_{n=1}^{\frac{L(\mathfrak{a})}{2}} \prod_{k=1}^{N_{2n}(\mathfrak{a})} \tilde{\Psi}^{\mathfrak{a}, 2n, k }} \\
			      \leqslant &
			      \sum_{k=1}^{N_1(\mathfrak a)}
			      \big|\Psi^{\mathfrak a,1,k}-\tilde\Psi^{\mathfrak a,1,k}\big|
			      \Big|\prod_{i=1}^{k-1}\Psi^{\mathfrak a,1,i}\Big|
			      \Big|\prod_{i=k+1}^{N_1(\mathfrak a)}\tilde\Psi^{\mathfrak a,1,i}\Big|
			      \Big|\prod_{n=1}^{\frac{L(\mathfrak a)}{2}}\prod_{i=1}^{N_{2n}(\mathfrak a)}\tilde\Psi^{\mathfrak a,2n,i}\Big|
			      \\[6pt]
			                & \quad+
			      \sum_{n=1}^{\frac{L(\mathfrak a)}{2}}
			      \sum_{k=1}^{N_{2n}(\mathfrak a)}
			      \big|\Psi^{\mathfrak a,2n,k}-\tilde\Psi^{\mathfrak a,2n,k}\big|
			      \Big|\prod_{i=1}^{N_1(\mathfrak a)}\Psi^{\mathfrak a,1,i}\Big|
			      \Big|\prod_{m=1}^{n-1}\prod_{i=1}^{N_{2m}(\mathfrak a)}\Psi^{\mathfrak a,2m,i}\Big|
			      \\
			                & \quad\times
			      \Big|\prod_{i=1}^{k-1}\Psi^{\mathfrak a,2n,i}\Big|
			      \Big|\prod_{i=k+1}^{N_{2n}(\mathfrak a)}\tilde\Psi^{\mathfrak a,2n,i}\Big|
			      \Big|\prod_{m=n+1}^{\frac{L(\mathfrak a)}{2}}\prod_{i=1}^{N_{2m}(\mathfrak a)}\tilde\Psi^{\mathfrak a,2m,i}\Big|.
		      \end{align}
		      For the monopoles, we simply use the estimate \(\abs{\rho(\xi^{\mathfrak{a}, 1, k} )\Psi^{\mathfrak{a}, 1, k}}\leqslant 1\) and control the difference by
		      \begin{align}\label{eq:monopole-estimate}
			      \abs{\rho(\xi^{\mathfrak{a}, 1, k} )(\Psi^{\mathfrak{a}, 1, k}-\tilde{\Psi}^{\mathfrak{a}, 1,k} )}
			      \leqslant
			      \norm{\chi^{-1} \rho}_{L^\infty} \norm{ \varphi -  \tilde{\varphi}}_{L^\infty(\chi)}.
		      \end{align}
		      For the neutral clusters, we apply the standard estimate from Lemma \ref{lem:2n-estimate}. The difference is then estimated as
		      \begin{align}\label{eq:cluster-estimate}
			      \abs{\rho(\xi^{\mathfrak{a}, 2n, k} )(\Psi^{\mathfrak{a}, 2n, k}-\tilde{\Psi}^{\mathfrak{a}, 2n, k})}
			      \lesssim_{\rho}
			      \norm{\nabla \varphi - \nabla \tilde \varphi}_{L^\infty(\chi)} \delta(\xi^{\mathfrak{a}, 2n, k})^{1-\nu_{2n,k}(\mathfrak{a})}.
		      \end{align}
		      Inserting this bound into \eqref{eq:Psi-tilde-Psi} and using the Steiner weights to absorb the \(\delta(\xi^{\mathfrak{a}, 2n, k})\),
		      we arrive at \eqref{eq:V,F-local-Lipschitz}.
		\item The statement follows immediately from
		      \begin{align}\label{eq:Va-T1-T2}
			      V^{\mathfrak{a}, T_1 }_{t}[\varphi] - V^{\mathfrak{a}, T_2}_{t}[\varphi]
			      = \int \rho( \mathd \xi^{\mathfrak{a}} ) (v^{\mathfrak{a}, T_1}_{t}(\xi^{\mathfrak{a}}) - v^{\mathfrak{a}, T_2}_{t}(\xi^{\mathfrak{a}} )) \Psi^{\mathfrak{a}},
		      \end{align}
		      and the estimate \eqref{eq:va-T-convergence} on \(v^{\mathfrak{a}}_{t}\)  in Corollary \ref{cor:va-scaling}.
	\end{enumerate}
\end{proof}
If \(\beta^{2}\geqslant 6\pi\), the degrees \(\tmop{deg}_{\nabla}^{V}(\mathfrak{a})\) and \(\tmop{deg}_{\nabla}^{F}(\mathfrak{a})\) in our estimates here grow without bound as the loop order \(L(\mathfrak{a})\) increases.
If \(\beta^{2}\) is not too close to \(8\pi\), we can stop the Picard scheme early and show the following bounds.
\begin{corollary}\label{cor:sublinear-bounds}
	Let \(\beta^{2}<\frac{6}{7}\; 8\pi\).
	There exists a family \((V^{\rho,T})_{T\geqslant0}\) satisfying Assumption \ref{hyp:sublinear-est}.
\end{corollary}
\begin{proof}
	We will show that for all \(\mathfrak{a}\in \mathfrak{A}_{\tmop{adm}}^{<7} \),
	\begin{enumerate}[a)]
		\item  \(\tmop{deg}_{\nabla}^{V}(\mathfrak{a})\leqslant 2 \),
		\item  \(\tmop{deg}_{\nabla}^{F}(\mathfrak{a}) < 1\),
		\item  \(\tmop{deg}_{\nabla}^{\mathcal{H}} < 2\).
	\end{enumerate}
	Combined with Proposition \ref{prop:Va-general-estimates} and Corollary \ref{cor:H-estimates},
	this shows that \(V^{[<7]}\) as defined through the coefficients \((v^{\mathfrak{a}})_{\mathfrak{a}\in \mathfrak{A}_{\tmop{adm}}^{<7} } \) satisfies the conditions in Assumption \ref{hyp:sublinear-est}.
	By definition of \(\tmop{deg}_{\nabla}^{\mathcal{H}}\),
	it always holds that b) \(\implies\) c).
	It is thus sufficient to verify the conditions a) and b).
	This follows directly from the admissibility conditions.
	Indeed, \(\beta^{2}< \frac{6}{7}\; 8\pi \iff \delta>\frac{1}{7}\), so that
	\begin{align}\label{eq:67-alpha}
		\alpha_2 = 2(1-2\delta_\kappa)\vee 0 < \frac{10}{7}, \qquad \alpha_{4} = 2(1-4\delta_\kappa)\vee 0 < \frac{6}{7}.
	\end{align}
	We distinguish the cases according to the number of neutral clusters \(c(\mathfrak{a})\).
	\begin{enumerate}[ \(c(\mathfrak{a})=\)1. ]
		\item With only one cluster, by definition \(K(\mathfrak{a})= \nu(\mathfrak{a})\in [0,2)\), so that a) is always satisfied.
		      If \(Q(\mathfrak{a})=0\), that is \(N_1(\mathfrak{a}) = 0\), then the admissibility condition (see Definition \ref{def:Admissible}) ensures \(K_{\mathcal{L}}(\mathfrak{a}) = (K(\mathfrak{a})-1)\vee 0\), which implies b).
		      On the other hand, if \(Q(\mathfrak{a}) \neq 0\), then \(N_1(\mathfrak{a}) \neq 0\) and by part c) of Definition \ref{def:Admissible}, \(\nu(\mathfrak{a}) < 1\), which implies b) as well.
		\item If there are two clusters of sizes \(2n_1\) and \(2n_2\), we know from c) of Definition \ref{def:Admissible} that \(\nu_{2n_1}(\mathfrak{a}), \nu_{2n_2}(\mathfrak{a}) \leqslant \alpha_{2n_1}-1, \alpha_{2n_2}-1\), respectively.
		      Since \(\alpha_{2} \geqslant \alpha_{\ell}\) for every \(\ell \geqslant 1\), this implies that
		      \begin{align}\label{eq:67-K-2cluster}
			      K(\mathfrak{a})\leqslant 2(\alpha_{2}-1)\vee 0 \leqslant \frac{6}{7} < 1,
		      \end{align}
		      which implies both a) and b).
		\item Since every neutral cluster must contain at least \(2\) points, three clusters are only possible when \(L(\mathfrak{a})\geqslant 6\).
		      If \(\nu_{2, k}(\mathfrak{a})=0\) for at least one \(k\in\{1,2 ,3\}\), then the same computation as in \eqref{eq:67-K-2cluster} applies.
		      Otherwise, \(\nu_{2, k}(\mathfrak{a})= (\alpha_2-1)\vee 0\) for all \(k\in \{1,2,3\}\), and thus
		      \begin{align}\label{eq:67-K-3cluster}
			      K(\mathfrak{a}) = 3(\alpha_2 - 1)\vee 0 \leqslant \frac{9}{7} \leqslant 2, \qquad
			      K_{\mathcal{L}}(\mathfrak{a})  = 2(\alpha_2 - 1)\vee 0 \leqslant \frac{6}{7} <1.
		      \end{align}
	\end{enumerate}
\end{proof}
\begin{remark}\label{rem:L=7}
	For \(L(\mathfrak{a}) = 7\), the sublinear estimates break down:
	There is a configuration with \(L(\mathfrak{a}) = 7, Q(\mathfrak{a}) \neq 0\), and \(\nu_{2, k}(\mathfrak{a}) = \alpha_2-1\) for \(k=1,2,3\).
	Since \(Q(\mathfrak{a}) \neq 0\), the estimates remain sublinear only if
	\begin{align}
		K_{\mathcal{L}}(\mathfrak{a})=K(\mathfrak{a})=3(\alpha_2-1)< 1 \iff \delta_\kappa>\frac{1}{6},
	\end{align}
	which is more restrictive than the regime \(\delta>\frac{1}{7}\) we cover here.
\end{remark}
\begin{remark}\label{rem:L=4}
	The same considerations as in the proof of Corollary \ref{cor:sublinear-bounds} show that, for \(\beta^{2}<6\pi\), the dipoles, corresponding to \(\tmop{deg}_{\nabla}^{F}(\mathfrak{a})=(1-4\delta_{\kappa})\vee0=0\), are uniformly bounded for the force.
	This in fact implies that all contributions except the dipole itself can be controlled uniformly in the field, which is used to show convergence of the Mayer expansion in \cite{bauerschmidtLogSobolevInequalityContinuum2021},
	and to control the strong FBSDE and strong formulation of the variational problem in \cite{gubinelliFBDSEApproachSine2026}.
\end{remark}

Finally, we state the following proposition that is needed to relate the SDE \eqref{eq:int-truncated-SDE} to the approximate measures \eqref{eq:nu-rho-T}.
The proof is a simple propagation of the \(T\)-dependent but field-independent uniform bound on \(V^{\rho,T}_{\tmop{sG}}\) from the initial condition \eqref{eq:V-initial},
\begin{align}
	\abs{V_{t}^{[1], \rho, T}[\varphi]}
	\leqslant \bar{\lambda}_t\norm{\rho}_{L^1}
	\leqslant \bar{\lambda}_T\norm{\rho}_{L^1}
	\lesssim \bar{\lambda}\norm{\rho}_{L^1}\langle T \rangle ^{1-\delta}.
\end{align}
The proof is postponed to Appendix \ref{app:T-dependent}.
\begin{proposition}\label{prop:T-dependent-est}
	Let \(T<\infty\), \(\rho\prec1\).
	With \((V^{[\ell],\rho, T})_{\ell < \ell^\ast}\) as defined in \eqref{eq:V-ell-from-V-a} and \eqref{eq:V-truncated},
	the scale interpolation \((V^{\rho, T} _t)_{t\in[0,T]}\) of the approximate sine-Gordon interaction potential \(V_{\tmop{sG}}^{\rho, T}  \) defined in \eqref{eq:V-truncated}
	satisfies Assumption \ref{hyp:scale-interpolation}.
\end{proposition}

\section{Details for the flow equation}\label{sec:details-flow}
In this section, we complete the proof of Proposition \ref{prop:A-B-estimates}.

We first derive the flow equation for the coefficients in Section \ref{sec:action-flow-formal}.
This then allows us to define the functions \(\mathbb{A}, \mathbb{B}\) and \(\Gamma\) satisfying the triangularity condition \eqref{eq:A,B-triangular} and such that the flow equation for the coefficients \eqref{eq:sG-va-full-flow} implies the flow equations for the approximate effective potential as in \eqref{eq:V-ell}.
Fix \(T>0\), \(\rho\prec1\), and \(\beta^{2}\in(0,8\pi)\), with the corresponding \(\ell^{\ast}\) and \(\delta_\kappa\) satisfying \(\ell^{\ast}\delta_\kappa>1\).
We again write \(V^{\mathfrak{a}}= V^{\mathfrak{a}, \rho, T} \) whenever no ambiguities arise.

\subsection{The action of the flow equation on the vertex fields}\label{sec:action-flow-formal}
In this section, we compute the action of the flow equation on the ansatz \eqref{eq:def-Va}, and make the considerations in Section \ref{sec:action-flow-informal} precise.

We start by rewriting the flow equation \eqref{eq:V-ell} in terms of the coefficients. Starting from \eqref{eq:V-ell-from-V-a}, we obtain for every \(\ell< \ell ^{\ast} \),
\begin{align}\label{eq:V-ell-to-va-start}
	\sum_{\mathfrak{a}\in \mathfrak{A}_{\tmop{adm}}^{\ell}} \partial_t V_{t}^{\mathfrak{a}}
	+ \frac{1}{2} \sum_{\mathfrak{a}\in \mathfrak{A}_{\tmop{adm}}^{\ell}}\Delta_{\dot{G}_t} V_{t}^{\mathfrak{a}}
	= \frac{1}{2} \sum_{\ell'+ \ell''= \ell}\sum_{\substack{\mathfrak{b}\in \mathfrak{A}_{\tmop{adm}}^{\ell'} \\ \mathfrak{c}\in \mathfrak{A}_{\tmop{adm}}^{\ell''}}} \langle \nabla  V_{t}^{\mathfrak{b}}, \dot{G}_{t} \nabla V_{t}^{\mathfrak{c}} \rangle_{L^{2}},
\end{align}
and
\begin{align}\label{eq:H-from-Va}
	\mathcal{H}_{t}^{[< \ell^{\ast} ]} = -\frac{1}{2}\sum_{\substack{\ell',\ell''<\ell^\ast \\ \ell'+ \ell'' \geqslant \ell^{\ast} }}\sum_{\substack{\mathfrak{b}\in \mathfrak{A}_{\tmop{adm}}^{\ell'} \\ \mathfrak{c}\in \mathfrak{A}_{\tmop{adm}}^{\ell''}}} \langle \nabla  V_{t}^{\mathfrak{b}}, \dot{G}_{t} \nabla V_{t}^{\mathfrak{c}} \rangle_{L^{2}}.
\end{align}

For the bilinear term, we define \(\mathbb{B}^{\mathfrak{b}, \mathfrak{c}}_{\mathfrak{a}}\)  such that
\begin{align}\label{eq:B-Va}
	\frac{1}{2}\sum_{\mathfrak{b}, \mathfrak{c}\in \mathfrak{A}^{<\ell^{\ast}}_{\tmop{adm}}} \langle \nabla V^{\mathfrak{b}}_{t}[\varphi], \dot{G}_{t} \nabla V^{\mathfrak{c}}_{t}[\varphi]\rangle_{L^{2}}
	 & =: \sum_{\mathfrak{a},\mathfrak{b}, \mathfrak{c}\in \mathfrak{A}^{<\ell^{\ast}}_{\tmop{adm}}}
	\int \rho(\mathd \xi^{\mathfrak{a}}) \mathbb{B}^{\mathfrak{b}, \mathfrak{c}}_{\mathfrak{a}}(v^{\mathfrak{b}}_{t}, v^{\mathfrak{c}}_{t}, \dot{G}_{t}) \Psi^{\mathfrak{a}}(\xi^{\mathfrak{a}}).
\end{align}
Similarly, we will define \(\mathbb{A}\) such that
\begin{align}\label{eq:A,Gamma-Va}
	\frac{1}{2}\sum_{\mathfrak{a}\in \mathfrak{A}_{\tmop{adm}}^{<\ell^{\ast}}} \Delta_{\dot{G}_{t}} V_{t}^{\mathfrak{a}}[\varphi]
	=: &
	\sum_{\mathfrak{a}\in \mathfrak{A}_{\tmop{adm}}^{<\ell^{\ast}}}
	\int \rho(\mathd \xi^{\mathfrak{a}})  \dot{\mathbb{G}}_{t}(\xi^{\mathfrak{a}}) v^{\mathfrak{a}}_{t}(\xi^{\mathfrak{a}}) \Psi^{\mathfrak{a}}(\xi^{\mathfrak{a}}) \\
	   & +
	\sum_{\mathfrak{a},\mathfrak{b}\in \mathfrak{A}_{\tmop{adm}}^{<\ell^{\ast}}}
	\int \rho(\mathd \xi^{\mathfrak{a}})
	\mathbb{A}^{\mathfrak{b}}_{\mathfrak{a}}(v^{\mathfrak{b}}_{t}, \dot{G}_{t}) \Psi^{\mathfrak{a}}(\xi^{\mathfrak{a}}),
\end{align}
where \(\mathbb{G}\) is defined as in \eqref{eq:def-Gbb}.
We carry over to this setting the definition of \(\Gamma\) from \eqref{eq:def-Mayer-Gamma}, used in \eqref{eq:v-ell-integrated}, by letting
\begin{align}\label{eq:def-Gamma-a}
	\Gamma_{t,s}^{\mathfrak{a}}(\xi^{\mathfrak{a}}) = \exp(\mathbb{G}_{s,t}(\xi^{\mathfrak{a}})).
\end{align}
Since \(\Psi^{\mathfrak{a}}\) does not depend on \(t\), this results in the equation at the level of the coefficients \(v^{\mathfrak{a}}\),
\begin{align}\label{eq:va-integrated}
	v_{t}^{\mathfrak{a}} (\xi^{\mathfrak{a}})
	= \Gamma_{t,T}^{\mathfrak a}(\xi^{\mathfrak{a}}) v_{T}^{\mathfrak{a}}(\xi^{\mathfrak{a}})
	+ \int_{t}^{T} \Gamma_{t,s}^{\mathfrak a}(\xi^{\mathfrak{a}})
	\left(
	\sum_{\mathfrak b}\mathbb{A}^{\mathfrak{b}}_{\mathfrak{a}}(v_{s}^{\mathfrak{b}}, \dot{G}_{s})
	- \sum_{\mathfrak b,\mathfrak c}\mathbb{B}^{\mathfrak{b}, \mathfrak{c}}_{\mathfrak{a}}(v_{s}^{\mathfrak{b}}, v_{s}^{\mathfrak{c}}, \dot{G}_s)
	\right)\mathd {s}.
\end{align}
For \(L(\mathfrak a)>1\), the definitions of \(\mathfrak{A}_{\tmop{irr}}\) and \(\mathfrak{A}_{\tmop{ren}}\) ensure that
\begin{align}\label{eq:va-terminal}
	v_{T}^{\mathfrak{a}} =
	\begin{cases}
		0,                    & \text{if \(\mathfrak{a}\in \mathfrak{A}_{\tmop{irr}}\) }, \\
		c_{T}^{\mathfrak{a}}, & \text{if \(\mathfrak{a}\in \mathfrak{A}_{\tmop{ren}}\)},
	\end{cases}
\end{align}
is compatible with the renormalization condition \eqref{eq:int-V-cos}.
If \(\mathfrak{a}\in \mathfrak{A}_{\tmop{irr}}\), the flow equation \eqref{eq:sG-va-full-flow} is precisely \eqref{eq:va-integrated} after inserting \eqref{eq:va-terminal}.
On the other hand, if \(\mathfrak{a}\in \mathfrak{A}_{\tmop{ren}}\) and \(L(\mathfrak a)>1\), then \(\Psi^{\mathfrak{a}}=1\), and \(\Gamma_{t,s}^{\mathfrak a}(\xi^{\mathfrak{a}} ) = 1\).
We arrive at \eqref{eq:sG-va-full-flow} by letting
\begin{align}
	c_T^{\mathfrak{a}} \assign -\int_{0}^{T}
	\left(
	\sum_{\mathfrak b}\mathbb{A}^{\mathfrak{b}}_{\mathfrak{a}}(v_{s}^{\mathfrak{b}}, \dot{G}_{s})
	- \sum_{\mathfrak b,\mathfrak c}\mathbb{B}^{\mathfrak{b}, \mathfrak{c}}_{\mathfrak{a}}(v_{s}^{\mathfrak{b}}, v_{s}^{\mathfrak{c}}, \dot{G}_s)
	\right)\mathd {s}.
\end{align}
Provided \(\mathbb{A}_{\mathfrak{a}}^{\mathfrak{b}} \neq 0\) and \(\mathbb{B}^{\mathfrak{b}, \mathfrak{c}}_{\mathfrak{a}}  \neq 0\) only if \(\mathfrak{a}\in \mathfrak{A}_{\tmop{irr}}\cup \mathfrak{A}_{\tmop{ren}}\),
\eqref{eq:va-integrated} with these terminal conditions is thus equivalent to \eqref{eq:sG-va-full-flow}.

\subsection{Definitions of \(\mathbb{B}\) and \(\mathbb{A}\)}\label{sec:def-A,B}
Starting from \eqref{eq:B-Va}, we will define the coefficients such that
\begin{align}\label{eq:B-derivation}
	  & \frac{1}{2} \langle \nabla  V_{t} ^{\mathfrak{b}}, \dot{G}_{t} \nabla  V_{t}^{\mathfrak{c}} \rangle_{L^{2}}                                                                                                                                                        \\
	= & \frac{1}{2}\int \rho(\mathd {\xi^{\mathfrak{b}} })\int \rho(\mathd {\xi^{\mathfrak{c}}}) v_{t}^{\mathfrak{b}}(\xi^{\mathfrak{b}} ) v_{t}^{\mathfrak{c}}(\xi^{\mathfrak{c}}) (\tmop{D} \Psi^{\mathfrak{b}} \tmop{D} \Psi^{\mathfrak{c}})[\dot{G}_{t}(\cdot, \cdot)] \\
	  & =:
	B^{\mathfrak{b}, \mathfrak{c}}(v_{t}^{\mathfrak{b}}, v_{t}^{\mathfrak{c}}, \dot{G}_{t}),
\end{align}
and
\begin{align}\label{eq:A-derivation}
	\frac{1}{2}\Delta_{\dot{G}_{t}} V_{t}^{\mathfrak{b}}[\varphi]
	 & =
	\int \rho(\mathd {\xi^{\mathfrak{b}}}) v_{t}^{\mathfrak{b}}(\xi^{\mathfrak{b}}) \frac{1}{2}\Delta_{\dot{G}_{t}}\Psi^{\mathfrak{b}}(\xi^{\mathfrak{b}}) \\
	 & =  A^{\mathfrak{b}} (v_{t}^{\mathfrak{b}}, \dot{G}_{t}).
\end{align}

The precise coefficients \(\mathbb{A}\) and \(\mathbb{B}\) will be defined in the next two subsections.

\subsection{The bilinear term \(\mathbb{B}\)}\label{sec:bilinear}
We define the bilinear part \(\mathbb{B}\) and prove its estimates.
\begin{lemma}\label{lem:D-Psi-branches}
	For \(\mathfrak a\in\mathfrak A_{\tmop{adm}}\), we define the disjoint set of
	derivative actions
	\begin{align}\label{eq:def-D-branches}
		\mathcal D(\mathfrak a)
		\assign \{\mathtt{k}\}
		\sqcup
		\{(\mathtt{l},2n,k):(2n,k)\in I_{\mathrm n}(\mathfrak a)\}
		\sqcup
		\{(\mathtt{r},2n,k):(2n,k)\in I_{\mathrm n}(\mathfrak a)\}.
	\end{align}
	Here and below, the action \(\mathtt{k}\) is omitted when \(N_1(\mathfrak a)=0\).
	With each \(\iota\in\mathcal D(\mathfrak a)\) we associate an index \(\mathfrak a_\iota\in \mathfrak{A}_{\tmop{adm}}^{\circ}\), a configuration \(\xi_{\iota}^{\mathfrak a}\),
	\begin{align}\label{eq:def-D-branch-data}
		\mathfrak a_\iota \assign
		\begin{cases}
			\mathfrak a,                      & \iota=\mathtt{k},        \\
			\mathcal L_{(2n,k)}(\mathfrak a), & \iota=(\mathtt{l},2n,k), \\
			\mathfrak a\setminus(2n,k),       & \iota=(\mathtt{r},2n,k),
		\end{cases}\quad
		\xi_{\iota}^{\mathfrak a} \assign
		\begin{cases}
			\xi^{\mathfrak a},                & \iota=\mathtt{k}\text{ or }\iota=(\mathtt{l},2n,k), \\
			\xi^{\mathfrak a\setminus(2n,k)}, & \iota=(\mathtt{r},2n,k),
		\end{cases},
	\end{align}
	as well as a branch factor
	\begin{align}
		\mathsf S_{\mathfrak a}^{\iota}(\xi^{\mathfrak a})[h]
		 & \assign
		\begin{cases}
			S_{\xi^{\mathfrak a,1}}(h),                                              & \iota=\mathtt{k},        \\
			S_{\xi^{\mathfrak a,2n,k}}^{-\Delta_{\mathcal L(\mathfrak a)}(2n,k)}(h), & \iota=(\mathtt{l},2n,k), \\
			S_{\xi^{\mathfrak a,2n,k}}^{\nu_{2n,k}(\mathfrak a)}(h),                 & \iota=(\mathtt{r},2n,k).
		\end{cases}
	\end{align}
	With these definitions,
	\begin{align}\label{eq:D-Psi-branches}
		\tmop D\Psi^{\mathfrak a}(\xi^{\mathfrak a})[h]
		=
		i\beta\sum_{\iota\in\mathcal D(\mathfrak a)}
		\mathsf S_{\mathfrak a}^{\iota}(\xi^{\mathfrak a})[h]
		\Psi^{\mathfrak a_\iota}
		(\xi_{\iota}^{\mathfrak a}).
	\end{align}
\end{lemma}
\begin{proof}
	Starting from \eqref{eq:def-psi-a}, we compute for \(h\in C_{c}^{\infty }(\mathbb{R}^{2})\),
	\begin{align}\label{eq:D-Psi-a}
		\tmop{D}\Psi^{\mathfrak{a}} (\xi^{\mathfrak{a}} ) [h]
		 & = \tmop{D}\left\{\prod_{k=1}^{N_{1}(\mathfrak{a})}\psi(\xi^{\mathfrak{a}, 1,  k})\prod_{(2n,k)\in I_{\mathrm n}(\mathfrak{a})}\psi^{\nu_{2n,k}(\mathfrak{a})}_{0}(\xi^{\mathfrak{a}, 2n, k}_{1:2n})\right\}[h] \\
		 & = \Psi^{\mathfrak{a}}(\xi^{\mathfrak{a}} ) \left\{
		\sum_{k=1}^{N_1(\mathfrak{a})} \frac{\tmop{D} \psi(\xi^{\mathfrak{a}, 1, k})[h] }{\psi(\xi^{\mathfrak{a}, 1, k})}
		+
		\sum_{(2n,k)\in I_{\mathrm n}(\mathfrak{a})}\frac{\tmop{D}\psi^{\nu_{2n,k}(\mathfrak{a})}_{0}(\xi^{\mathfrak{a}, 2n, k}_{1:2n})[h] }{\psi^{\nu_{2n,k}(\mathfrak{a})}_{0}(\xi^{\mathfrak{a}, 2n, k}_{1:2n}) }
		\right\},
	\end{align}
	where
	\begin{align}\label{eq:D-psi-charge}
		\tmop{D} \psi(\xi^{\mathfrak{a}, 1, k})[h]
		 & = i\beta \sigma^{\mathfrak{a}, 1, k} h(x^{\mathfrak{a}, 1, k} )\psi(\xi^{\mathfrak{a}, 1, k}),
	\end{align}
	and
	\begin{align}\label{eq:D-psi-cluster}
		{\tmop{D}\psi^{\nu_{2n,k}(\mathfrak{a})}_{0}(\xi^{\mathfrak{a}, 2n, k}_{1:2n})[h] }
		= & i\beta S_{\xi^{\mathfrak{a}, 2n, k}}(h)\frac{\psi(\xi^{\mathfrak{a}, 2n, k}_{1:2n})}{\delta(\xi^{\mathfrak{a},  2n, k} )^{\nu_{2n,k}(\mathfrak{a})} } \\
		= & i\beta
		\begin{cases}
			{S_{\xi^{\mathfrak{a}, 2n, k}}^{\nu_{2n,k}(\mathfrak{a})}(h)} \psi^{0}_{0} (\xi^{\mathfrak{a}, 2n, k}_{1:2n}) + {S_{\xi^{\mathfrak{a}, 2n, k}}^{\nu_{2n,k}(\mathfrak{a})}(h)} , & \text{if \(\nu_{2n,k}(\mathfrak{a})\leqslant1\) } \\
			{S^{1}_{\xi^{\mathfrak{a}, 2n, k}}(h)} \psi_{0}^{\nu_{2n,k}(\mathfrak{a})-1}(\xi^{\mathfrak{a}, 2n, k}_{1:2n}) + S^{\nu_{2n, k}(\mathfrak{a})}_{\xi^{\mathfrak a,2n,k}}(h),     & \text{if \(\nu_{2n,k}(\mathfrak{a})>1\) }.
		\end{cases}
	\end{align}
	The division by the vertex fields in \eqref{eq:D-Psi-a} is purely a notational convenience.
	The denominator always appears in the numerator and there is no actual division in \eqref{eq:D-psi-cluster}.
	The formula could instead be written by summing over \(\Psi^{\mathfrak{a} } \) with one cluster replaced by its functional derivative.
	Splitting the terms according to whether they keep (\(\mathtt{k}\)), lower (\(\mathtt{l}\)), or remove (\(\mathtt{r}\)) a cluster, we obtain \eqref{eq:D-Psi-branches}.
	It remains to verify that \(\mathfrak{a}_\iota\in \mathfrak{A}_{\tmop{adm}}^{\circ}\) by verifying the conditions (a)--(d) from Definition \ref{def:Admissible} one by one.
	\begin{enumerate}[a)]
		\item This condition is always satisfied since \(m(\mathfrak{a}_\iota)\leqslant m(\mathfrak{a})\leqslant L(\mathfrak{a})= L(\mathfrak{a}_\iota)\).
		\item By definition of the lowering map in \eqref{eq:def-lowering-map}, the localization-degree condition \eqref{eq:nu-conditions} is preserved.
		\item If \(\nu_{2n,k}(\mathfrak{a})<1\) for all \((2n,k)\in I_{\mathrm n}(\mathfrak a)\), then the same is true for \(\mathfrak a_\iota\), and thus \(\mathfrak{a}_\iota\in \mathfrak{A}_{\tmop{adm}}^{\circ}\).
		      On the other hand, if \(\mathfrak{a}\) consists of a single cluster of size \(2n\) with \(\nu(\mathfrak{a})\geqslant 1\), then \(\tmop{D}\Psi^{\mathfrak{a}} [h] = i\beta \mathcal{L}_{(2n,k)} \Psi^{\mathfrak{a}} S^{-\Delta_{\mathcal{L}}(2n,k)}_{\xi^{\mathfrak{a}} }(h) + i\beta S^{\nu(\mathfrak{a})}_{\xi^{\mathfrak{a}} }(h) \),
		      which by definition of the lowering map again satisfies \eqref{eq:nu-conditions}.
		      Moreover, we see that \(\nu(\mathfrak{a}_{\iota})\in \{0, (\alpha_{2n}-1)\vee 0\}\) and thus \(\mathfrak{a}_{\iota}\in \mathfrak{A}_{\tmop{adm}}^{\circ}\).
		\item The charge of the contribution stays unchanged.
	\end{enumerate}
\end{proof}
\begin{remark}
	The above setup retains all the information required for the estimates and construction later.
	For example, if a cluster \(\xi^{\mathfrak{a}, 2n, k} \) is removed by \(\tmop{D}\), it is removed from the vertex field but we retain the information about its multiplicity in the factor \(\mathsf{S}\).
\end{remark}

\paragraph{Definition of \(\mathbb{B}\)}
With this notation, we can define the bilinear coefficient \(\mathbb{B}\).
Let \(\mathfrak{b}, \mathfrak{c}\in \mathfrak{A}_{\tmop{adm}}\).
For each action \(\iota\in \mathcal D(\mathfrak{b})\), define
\begin{align}
	\xi_{\mathtt{r}, \iota}^{\mathfrak{b}} \assign
	\begin{cases}
		\xi^{\mathfrak{b}, 2n, k}, & \iota=(\mathtt{r}, 2n, k) \\
		\varnothing,               & \text{otherwise},
	\end{cases}
\end{align}
with \(\xi_{\mathtt{r}, \jmath}^{\mathfrak{c}} \) defined in complete analogy for \(\jmath\in \mathcal D(\mathfrak{c})\).
Let \(\mathfrak{d}_{\iota, \jmath} \assign \mathfrak{b}_{\iota} + \mathfrak{c}_{\jmath}\).
We define the unprocessed kernel resulting from the actions \(\iota\) and \(\jmath\) as
\begin{align}\label{eq:def-B-raw}
	B^{\mathfrak{b}, \mathfrak{c}}_{\iota, \jmath} (\xi^{\mathfrak{d}_{\iota,\jmath}})
	= \int \rho(\mathd \xi_{\mathtt{r}, \iota}^{\mathfrak{b}} ) \rho (\mathd \xi^{\mathfrak{c}}_{\mathtt{r}, \jmath})
	v^{\mathfrak{b}}(\xi^{\mathfrak{b}}) v^{\mathfrak{c}}(\xi^{\mathfrak{c}} ) \mathbb{G}^{\mathfrak{b}, \mathfrak{c}}_{\iota, \jmath} (\xi^{\mathfrak{b}}, \xi^{\mathfrak{c}}),
\end{align}
where
\begin{align}\label{eq:def-Gb-DD}
	\mathbb{G}^{\mathfrak{b}, \mathfrak{c}}_{\iota, \jmath} (\xi^{\mathfrak{b}}, \xi^{\mathfrak{c}})
	= - \frac{\beta^{2}}{2} (\mathsf S_{\mathfrak{b}}^{\iota}\otimes \mathsf S_{\mathfrak{c}}^{\jmath} )[\dot{G}_{t}].
\end{align}
For every \(\mathfrak{a}\in \mathfrak{A}_{\tmop{adm}}\), we define
\begin{align}\label{eq:B-def}
	\mathbb{B}^{\mathfrak{b}, \mathfrak{c}}_{\mathfrak{a}}(v^{\mathfrak{b}}, v^{\mathfrak{c}}, \dot{G}_t)
	\assign
	\sum_{\iota\in \mathcal D(\mathfrak{b})}\sum_{\jmath\in \mathcal D(\mathfrak{c})}
	\sum_{\mathfrak{d}\in \mathfrak{A}_{\tmop{adm}}}
	\mathsf{R}_{\mathfrak{d} \to \mathfrak{a}} \mathsf{P}_{\mathfrak{d}_{\iota, \jmath} \to \mathfrak{d}} B^{\mathfrak{b}, \mathfrak{c}}_{\iota, \jmath} (\xi^{\mathfrak{d}_{\iota, \jmath}}).
\end{align}

We now verify that \(\mathbb{B}\) is compatible with the flow equation for the potential.

\begin{lemma}\label{lem:B-reconstruction}
	The coefficients defined in \eqref{eq:B-def} satisfy the triangularity condition \eqref{eq:A,B-triangular} and the identity
	\begin{align}\label{eq:B-reconstruction}
		\sum_{\mathfrak{a}\in \mathfrak{A}_{\tmop{adm}}} \int \rho(\mathd \xi^{\mathfrak{a}}) \mathbb{B}^{\mathfrak{b}, \mathfrak{c}}_{\mathfrak{a}}(v^{\mathfrak{b}}, v^{\mathfrak{c}}, \dot{G}_t)(\xi^{\mathfrak{a}}) \Psi^{\mathfrak{a}} (\xi^{\mathfrak{a}})
		=
		B^{\mathfrak{b},\mathfrak{c}}(v^{\mathfrak{b}}, v^{\mathfrak{c}}, \dot{G}_t).
	\end{align}
	Moreover, \(\mathbb{B}^{\mathfrak{b}, \mathfrak{c}}_{\mathfrak{a}} \neq 0 \) only if \(\mathfrak{a}\in \mathfrak{A}_{\tmop{irr}}\cup \mathfrak{A}_{\tmop{ren}}\).
\end{lemma}
\begin{proof}
	By Lemma \ref{lem:D-Psi-branches},
	\begin{align}
		\frac12\left(  \tmop{D}\Psi^{\mathfrak{b}} \tmop{D} \Psi^{\mathfrak{c}}\right)[\dot{G}_t]
		= & - \frac{\beta^{2}}{2} \sum_{\iota\in \mathcal D(\mathfrak{b})} \sum_{\jmath\in \mathcal D(\mathfrak{c})}
		(\mathsf S_{\mathfrak{b}}^{\iota}\otimes \mathsf S_{\mathfrak{c}}^{\jmath} )[\dot{G}_{t}]
		\Psi^{\mathfrak{b}_{\iota}}(\xi^{\mathfrak{b}_{\iota}}) \Psi^{\mathfrak{c}_{\jmath}} (\xi^{\mathfrak{c}_{\jmath}})                                                                                                                                                                                      \\
		= & \sum_{\iota\in \mathcal D(\mathfrak{b})} \sum_{\jmath\in \mathcal D(\mathfrak{c})} \mathbb{G}^{\mathfrak b,\mathfrak c}_{\iota, \jmath}(\xi^{\mathfrak{b}}, \xi^{\mathfrak{c}}) \Psi^{\mathfrak{b}_{\iota}}(\xi^{\mathfrak{b}_{\iota}}) \Psi^{\mathfrak{c}_{\jmath}} (\xi^{\mathfrak{c}_{\jmath}}).
	\end{align}
	By the same lemma, \(\mathfrak{b}_{\iota}, \mathfrak{c}_{\jmath}\in \mathfrak{A}_{\tmop{adm}}^{\circ} \).
	Therefore, we may apply Lemma \ref{lem:dual-P-R-reconstruction} to insert both the pairing and raising maps to arrive at \eqref{eq:B-reconstruction}.
	To see that \(\mathbb{B}\) satisfies the triangularity condition, note that none of the applied maps \(\mathcal{R}, \mathcal{P} , \mathcal{L}\) alter the loop order or the charge of the index.
	Therefore, \(L(\mathfrak{a}) = L(\mathfrak{b}) + L(\mathfrak{c})\) and \(Q(\mathfrak{a}) = Q(\mathfrak{b}) + Q(\mathfrak{c})\) hold trivially.
	This also ensures that \(\mathbb{B}\) is only non-zero for \(\mathfrak{a} \succ \mathfrak{b}, \mathfrak{c}\).
	Finally, the last claim follows from the definition of the maps \(\mathcal{P}\) and \(\mathcal{R}\).
\end{proof}

\paragraph{Estimates on \(\mathbb{B}\)}

\begin{lemma}\label{lem:B-estimate}
	Let \(t>0\) and let
	\((v^{\mathfrak d})_{\mathfrak d\in\mathfrak A_{\tmop{adm}}}\), with
	\(v^{\mathfrak d}\in\mathcal K_t^{\mathfrak d}\), be compatible with charge conjugation.
	For \(\mathfrak{a}, \mathfrak{b}, \mathfrak{c}\in \mathfrak{A}_{\tmop{adm}}\),
	the coefficient \(\mathbb{B}^{\mathfrak{b}, \mathfrak{c}}_{\mathfrak{a}}\) defined in \eqref{eq:B-def} satisfies
	\begin{align}
		\tnorm{\mathbb{B}^{\mathfrak{b}, \mathfrak{c}}_{\mathfrak{a}}(v^{\mathfrak{b}}, v^{\mathfrak c}; \dot{G}_t)}_{t, \mathfrak{a}}
		\lesssim t^{-1-[\mathfrak{a}]+ [\mathfrak{b}] + [\mathfrak{c}]} \mathe^{-\frac{m^{2}}{t}}
		\tnorm{v^{\mathfrak{b}}}_{\mathfrak{b}, t} \tnorm{v^{\mathfrak{c}}}_{\mathfrak{c}, t}.
	\end{align}
	Moreover, \eqref{eq:A,B-charge-conjugation} holds for \(\mathbb{B}\).
\end{lemma}
\begin{proof}
	Since the sums have only finitely many non-zero contributions, proving the estimate for each summand implies the estimate for the sum after adjusting the constant.

	Fix \(\iota \in \mathcal D(\mathfrak{b}), \jmath\in \mathcal D(\mathfrak{c})\).
	For ease of notation, define
	\begin{align}\label{eq:def-nu-iota}
		\nu^{\mathfrak{b}}_{\iota} \assign
		\begin{cases}
			0,                                          & \iota = \mathtt{k},          \\
			-\Delta_{\mathcal{L}(\mathfrak{b})}(2n, k), & \iota= (\mathtt{l}, 2n,k),   \\
			\nu_{2n,k }(\mathfrak{b}),                  & \iota = (\mathtt{r}, 2n, k),
		\end{cases}
		\quad
		\delta^{\mathfrak{b}}_{\iota} \assign
		\begin{cases}
			1,                                 & \iota = \mathtt{k}                                   \\
			\delta(\xi^{\mathfrak{b}, 2n, k}), & \iota\in \{(\mathtt{r}, 2n,k), (\mathtt{l}, 2n,k)\},
		\end{cases}
	\end{align}
	with \(\nu_{\jmath}^{\mathfrak{c}}, \delta^{\mathfrak{c}}_{\jmath} \) defined in complete analogy.
	Set
	\begin{align}
		r_{\tmop{raw}} & \assign\frac32r(\mathfrak a),
		               &
		r_{\tmop{hk}}  & \assign2r(\mathfrak a),
		               &
		r_{\tmop{abs}} & \assign\frac12r(\mathfrak a).
	\end{align}
	Then, by the definition of the rate \eqref{eq:def-r-steiner-rate} and since \(L(\mathfrak a)=L(\mathfrak b)+L(\mathfrak c)\),
	\begin{align}\label{eq:B-rate-budget}
		r_{\tmop{raw}}+r_{\tmop{hk}}+r_{\tmop{abs}}
		= 4r(\mathfrak a)
		< r(\mathfrak b)\wedge r(\mathfrak c).
	\end{align}

	We will show that with \(B^{\mathfrak{b}, \mathfrak{c}}_{\iota, \jmath}\) as in \eqref{eq:def-B-raw},
	\begin{align}\label{eq:B-raw-est}
		\tnorm{B^{\mathfrak{b}, \mathfrak{c}}_{\iota, \jmath}}_{\mathfrak d_{\iota,\jmath},t;r_{\tmop{raw}}}
		\lesssim \mathe^{-\frac{m^{2}}{t}} t^{-2 + \frac{1}{2}(\nu_{\iota}^{\mathfrak{b}} + \nu_{\jmath}^{\mathfrak{c}})}
		\tnorm{v^{\mathfrak{b}}}_{\mathfrak{b}, t} \tnorm{v^{\mathfrak{c}}}_{\mathfrak{c}, t}.
	\end{align}
	The pairing map does not change \(K\). Hence, by
	Lemmas~\ref{lem:dual-pairing-estimate} and
	\ref{lem:dual-R-kernel-estimate}, applied with base rate
	\(r(\mathfrak a)\) and
	\(\varepsilon=r_{\tmop{raw}}-r(\mathfrak a)=\frac12r(\mathfrak a)\),
	\begin{align}
		\tnorm{\mathsf{R}_{\mathfrak{d} \to \mathfrak{a}}\mathsf{P}_{\mathfrak{d}_{\iota, \jmath} \to \mathfrak{d}} B_{\iota, \jmath}^{\mathfrak{b}, \mathfrak{c}}}_{t, \mathfrak{a}}
		\lesssim  t^{-\frac{1}{2} (K(\mathfrak{a})-K(\mathfrak{d}_{\iota, \jmath}))}
		\tnorm{B^{\mathfrak{b}, \mathfrak{c}}_{\iota, \jmath}}_{\mathfrak d_{\iota,\jmath},t;r_{\tmop{raw}}}.
	\end{align}
	Finally, since
	\(K(\mathfrak d_{\iota,\jmath})
	=K(\mathfrak b)+K(\mathfrak c)
	-\nu_{\iota}^{\mathfrak b}-\nu_{\jmath}^{\mathfrak c}\),
	\begin{align}
		-2+\frac{1}{2}(\nu_{\iota}^{\mathfrak b}+\nu_{\jmath}^{\mathfrak c})
		-\frac{1}{2}(K(\mathfrak a)-K(\mathfrak d_{\iota,\jmath}))
		=-1-[\mathfrak a]+[\mathfrak b]+[\mathfrak c],
	\end{align}
	which yields the claim.

	It remains to prove the estimate \eqref{eq:B-raw-est} on the raw kernel.
	We distinguish two cases.

	\emph{Case 1. \(\nu^{\mathfrak{b}}_{\iota}, \nu^{\mathfrak{c}}_{\jmath} \leqslant 1\).}
	Let \(\tau_{\iota}^{\mathfrak{b}}=\mathbb{1}_{\iota\neq\mathtt{k}}\), and define \(\tau_{\jmath}^{\mathfrak{c}}\) analogously to track the required Taylor expansions in \(\mathbb{G}\).
	Assume first that both \(\xi^{\mathfrak{b}_{\iota }}, \xi ^{\mathfrak{c}_{\jmath}} \neq \varnothing\).
	Combining the pointwise bound in the charged--charged case, the corresponding one-sided Taylor estimate in the mixed charged--neutral case, and Lemma \ref{lem:hk-multipoint} below in the neutral--neutral case, we find with \(r=r_{\tmop{hk}}\),
	\begin{align}\label{eq:G-branch-estimate}
		\abs{\mathbb{G}_{\iota, \jmath}^{\mathfrak{b}, \mathfrak{c}}  (\xi^{\mathfrak{b}}, \xi^{\mathfrak{c}})}
		 & \lesssim t^{-1 + \frac{1}{2} (\tau^{\mathfrak{b}}_{\iota} + \tau^{\mathfrak{c}}_{\jmath})}
		\mathe^{-\frac{m^{2}}{t}}
		(\delta_{\iota}^{\mathfrak{b}})^{\tau_{\iota}^{\mathfrak{b}}  - \nu_{\iota}^{\mathfrak{b}}}
		(\delta_{\jmath}^{\mathfrak{c}})^{\tau_{\jmath}^{\mathfrak{c}} - \nu_{\jmath}^{\mathfrak{c}}}
		\omega_{t}^{r_{\tmop{hk}}}(\xi^{\mathfrak{b}}_{\iota})
		\omega_{t}^{r_{\tmop{hk}}}(\xi^{\mathfrak{c}}_{\jmath})
		\mathe^{-r_{\tmop{hk}} \sqrt{t} d(\xi_{\iota}^{\mathfrak{b}}, \xi_{\jmath}^{\mathfrak{c}})}   \\
		 & \lesssim t^{-1 + \frac{1}{2}(\nu_{\iota}^{\mathfrak{b}} + \nu_{\jmath}^{\mathfrak{c}})}
		\mathe^{-\frac{m^{2}}{t}}
		\omega_{t}^{r_{\tmop{hk}}+r_{\tmop{abs}}}(\xi^{\mathfrak{b}}_{\iota})
		\omega_{t}^{r_{\tmop{hk}}+r_{\tmop{abs}}}(\xi^{\mathfrak{c}}_{\jmath})
		\mathe^{-r_{\tmop{hk}} \sqrt{t} d(\xi_{\iota}^{\mathfrak{b}}, \xi_{\jmath}^{\mathfrak{c}})},
	\end{align}
	where in the last step we used
	\(\tau_{\iota}^{\mathfrak b}-\nu_{\iota}^{\mathfrak b}\geqslant0\),
	the analogous inequality for \(\mathfrak c\), and \eqref{eq:est-steiner} with rate \(r_{\tmop{abs}}\).
	By \eqref{eq:B-rate-budget} and the Steiner-tree triangle inequality,
	\begin{align}
		 & \omega_t^{r_{\tmop{raw}}}(\xi^{\mathfrak d_{\iota,\jmath}})
		\frac{
			\omega_t^{r_{\tmop{hk}}+r_{\tmop{abs}}}(\xi^{\mathfrak b}_{\iota})
			\omega_t^{r_{\tmop{hk}}+r_{\tmop{abs}}}(\xi^{\mathfrak c}_{\jmath})}
		{\omega_t^{r(\mathfrak b)}(\xi^{\mathfrak b})
			\omega_t^{r(\mathfrak c)}(\xi^{\mathfrak c})}
		\mathe^{-r_{\tmop{hk}}\sqrt t\,
			d(\xi^{\mathfrak b}_{\iota},\xi^{\mathfrak c}_{\jmath})}
		\lesssim
		\mathe^{-r_{\tmop{abs}}\sqrt t\,
			d(\xi^{\mathfrak b}_{\iota},\xi^{\mathfrak c}_{\jmath})}.
	\end{align}
	Inserting this estimate into \eqref{eq:def-B-raw} and applying Young's
	inequality to the connecting variable gives
	\begin{align}
		\tnorm{B_{\iota, \jmath}^{\mathfrak{b}, \mathfrak{c}}}_{\mathfrak d_{\iota,\jmath},t;r_{\tmop{raw}}}
		 & \lesssim t^{-2 + \frac{1}{2} (\nu_{\iota}^{\mathfrak{b}}+ \nu_{\jmath}^{\mathfrak{c}})}
		\mathe^{-\frac{m^{2}}{t}}
		\tnorm{v^{\mathfrak{b}}}_{\mathfrak{b}, t} \tnorm{v^{\mathfrak{c}}}_{\mathfrak{c}, t}.
	\end{align}

	In case either of \(\xi^{\mathfrak{b}_{\iota }} = \varnothing\) or \(\xi ^{\mathfrak{c}_{\jmath}} = \varnothing\), we estimate \eqref{eq:def-B-raw} directly by integrating the removed cluster and applying the corresponding one- or two-sided heat-kernel bound followed by Young's inequality.
	Using that \(\omega_t^r(\varnothing)=1\), this gives \eqref{eq:B-raw-est} with the same power of \(t\).

	\emph{Case 2. \(\nu_{\iota}^{\mathfrak{b}} \vee \nu_{\jmath}^{\mathfrak{c}} \in (1, 2)\). }
	Without loss of generality, suppose \(\nu_{\iota}^{\mathfrak{b}}>1\), with the other case being a mere change of notation.
	By the definition of \(\mathcal D(\mathfrak{b})\) in \eqref{eq:def-D-branches}, this implies \(\iota= (\mathtt{r}, 2n, k)\) and \(\nu_{\iota}^{\mathfrak{b}} = \alpha_{2n}\).
	By definition of admissibility, \(\mathfrak{b}\) consists of a single neutral cluster, so we must have \(c(\mathfrak b)=1\) and \(N_1(\mathfrak b)=0\).
	Therefore, using the charge-conjugation assumption as in the proof of Lemma \ref{lem:charge-moment},
	\begin{align}
		\int \rho(\mathd \xi_{\mathtt{r}, \iota}^{\mathfrak{b}}) v^{\mathfrak{b}} (\xi^{\mathfrak{b}}) \mathbb{G}^{\mathfrak{b}, \mathfrak{c}}_{\iota, \jmath} = 0,
	\end{align}
	and hence \(B^{\mathfrak{b},\mathfrak{c}}_{\iota,\jmath}=0\). 
	In particular, \eqref{eq:B-raw-est} holds.
	The compatibility with charge conjugation is immediate from the definition \eqref{eq:def-B-raw}.
\end{proof}

\begin{lemma}\label{lem:hk-multipoint}
	Let \(\xi_{I} = (\sigma_i, x_i)_{i\in I}\) and \(\eta_{J}=(\tau_{j}, y_j)_{j\in J}\) be two non-empty neutral configurations.
	For every \(t,r>0\), \(x_{\ast}\in \{x_{i}\}_{i\in I}\), and \(y_{\ast}\in \{y_j\}_{j\in J}\),
	\begin{align}\label{eq:hk-multipoint}
		\abs{S_{\xi_{I}}\otimes S_{\xi_{J}} [\dot{G}_t] }
		 & \lesssim
		\mathe^{-\frac{m^{2}}{t}} \delta(\xi_{I}, x_{\ast}) \delta(\eta_{J}, y_{\ast})
		\omega^{r}_t(\xi_{I}) \omega^{r}_{t}(\xi_{J}) \mathe^{-r \sqrt{t} \abs{x_{\ast} - y_{\ast}}} \\
		 & \times
		(1 + \sqrt{t}\abs{x_{\ast} - y_{\ast}} + \sqrt{t} \delta(\xi_I, x_{\ast}) + \sqrt{t} \delta(\eta_J, y_{\ast}))^{2}.
	\end{align}
\end{lemma}
\begin{proof}
	Using that \(S_{\xi_{I}}(1) = S_{\eta_{J}}(1)= 0\) by neutrality, we have
	\begin{align}
		          & \abs{S_{\xi_{I}}\otimes S_{\xi_{J}} [\dot{G}_t]} \\
		\leqslant & \sum_{i\in I, j\in J}
		\int_{0}^{1} \mathd  \theta_1 \int_{0}^{1} \mathd \theta_2 \abs{\nabla_1\nabla_2 \dot{G}_{t}(
			x_{\ast} + \theta_1 (x_{i}- x_{\ast}), y_{\ast} + \theta_2 (y_j-y_{\ast})
			)
			[x_i-x_\ast, y_j- y_{\ast}]}.
	\end{align}
	Inserting the estimate from Lemma \ref{lem:hk-scaling} yields \eqref{eq:hk-multipoint}, as claimed.
\end{proof}

\subsection{The linear term}\label{sec:linear}
This section is structured in complete analogy to Section \ref{sec:bilinear}.
We first define the action of the functional Laplacian on the vertex fields, point configurations and indices.
\begin{lemma}\label{lem:Laplacian-Psi-branches}
	Let \(\mathfrak{a}\in\mathfrak A_{\tmop{adm}}\).
	Define the set of the actions of the functional Laplacian,
	\begin{align}\label{eq:def-Laplacian-branches}
		\mathcal A_{\Delta}(\mathfrak{a})
		\assign &
		\{\mathtt{d}\}
		\sqcup
		\{(\mathtt{r},2n,k):(2n,k)\in I_{\mathrm n}(\mathfrak{a})\} \\
		        & \sqcup
		\left\{
		(\mathtt{rr},(2n,k),(2n',k')):\:(2n,k),(2n',k')\in I_{\mathrm n}(\mathfrak{a}),\ (2n,k)\neq(2n',k')
		\right\}.
	\end{align}
	For every action \(\iota\in\mathcal A_{\Delta}(\mathfrak{a})\), we define an index \(\mathfrak{a}_\iota\in \mathfrak{A}_{\tmop{adm}}\) and a configuration \(\xi_\iota^{\mathfrak{a}}\) by
	\begin{align}\label{eq:def-Laplacian-branch-indices}
		\mathfrak{a}_\iota & \assign
		\begin{cases}
			\mathfrak{a},                & \iota=\mathtt d,                     \\
			\mathfrak{a}\setminus(2n,k), & \iota=(\mathtt r,2n,k),              \\
			\mathfrak{a}\setminus
			\{(2n,k),(2n',k')\},         & \iota=(\mathtt{rr},(2n,k),(2n',k')),
		\end{cases} \\
		\xi_\iota^{\mathfrak{a}}
		                   & \assign
		\begin{cases}
			\xi^{\mathfrak{a}},                & \iota=\mathtt d,                     \\
			\xi^{\mathfrak{a}\setminus(2n,k)}, & \iota=(\mathtt r,2n,k),              \\
			\xi^{\mathfrak{a}\setminus
			\{(2n,k),(2n',k')\}},              & \iota=(\mathtt{rr},(2n,k),(2n',k')),
		\end{cases}
	\end{align}
	and the factors
	\begin{align}\label{eq:def-Laplacian-branch-factors}
		\mathsf G_{t,\mathfrak{a}}^{\iota}(\xi^{\mathfrak{a}})
		\assign
		\begin{cases}
			\dot{\mathbb G}_t(\xi^{\mathfrak{a}}), & \iota=\mathtt d,                     \\
			\displaystyle
			\frac{
				\dot{\mathbb G}_t(\xi^{\mathfrak{a},2n,k})
				+\dot{\mathbb G}_t(\xi^{\mathfrak{a}\setminus(2n,k)},\xi^{\mathfrak{a},2n,k})
			}{
				\delta(\xi^{\mathfrak{a},2n,k})^{\nu_{2n,k}(\mathfrak{a})}
			},                                     & \iota=(\mathtt r,2n,k),              \\
			\displaystyle
			\frac{
				\dot{\mathbb G}_t(\xi^{\mathfrak{a},2n,k},\xi^{\mathfrak{a},2n',k'})
			}{
				\delta(\xi^{\mathfrak{a},2n,k})^{\nu_{2n,k}(\mathfrak{a})}
				\delta(\xi^{\mathfrak{a},2n',k'})^{\nu_{2n',k'}(\mathfrak{a})}
			},                                     & \iota=(\mathtt{rr},(2n,k),(2n',k')).
		\end{cases}
	\end{align}
	Then,
	\begin{align}\label{eq:Laplacian-Psi-branches}
		\frac12\Delta_{\dot G_t}\Psi^{\mathfrak{a}}(\xi^{\mathfrak{a}})
		= \dot{\mathbb G}_t(\xi^{\mathfrak{a}}) \Psi^{\mathfrak{a}}(\xi^{\mathfrak{a}})
		+
		\sum_{\iota\in\mathcal A_{\Delta}(\mathfrak{a})\setminus\{\mathtt d\}}
		\mathsf G_{t,\mathfrak{a}}^{\iota}(\xi^{\mathfrak{a}})\Psi^{\mathfrak{a}_\iota}(\xi_\iota^{\mathfrak{a}}).
	\end{align}
\end{lemma}
\begin{proof}
	The proof is analogous to that of Lemma \ref{lem:D-Psi-branches}, and the expression follows by applying the functional derivative twice and contracting the result with \(\dot G_t\).
\end{proof}
\begin{remark}
	Here, \(\mathtt d\) denotes the diagonal action of the Laplacian, while \(\mathtt r\)
	and \(\mathtt{rr}\) denote the removal of one and two neutral clusters from the configuration, respectively.
	Note that in contrast to the bilinear term and Lemma \ref{lem:D-Psi-branches}, we never apply the lowering operation.
\end{remark}

\paragraph{Definition of \(\mathbb{A}\)}
With this action of the Laplacian, we can define \(\mathbb{A}\), which contains only the off-diagonal part of \(\mathcal{A}_{\Delta}\).
Let \(\mathfrak{b}\in \mathfrak{A}_{\tmop{adm}}\) and let \(\iota\in \mathcal{A}_{\Delta}(\mathfrak{b})\setminus \{\mathtt{d}\}\).
Define
\begin{align}
	\xi_{\mathtt{r}, \iota}^{\mathfrak{b}} \assign
	\begin{cases}\label{eq:xi-r}
		\xi^{\mathfrak{b}, (2n, k)},               & \iota = (\mathtt{r}, 2n, k)                     \\
		\xi^{\mathfrak{b}, (2n', k'),(2n'', k'')}, & \iota = (\mathtt{rr}, (2n', k'), (2n'', k'' )).
	\end{cases}
\end{align}
We define the unprocessed kernel resulting from the action \(\iota\) as
\begin{align}\label{eq:def-A-raw}
	A_{\iota}^{\mathfrak{b}}  (\xi^{\mathfrak{b}_{\iota}})
	= \int \rho(\mathd \xi^{\mathfrak b}_{\mathtt{r}, \iota}) v^{\mathfrak{b}} (\xi^{\mathfrak{b}}) \mathsf{G}_{t, \mathfrak{b}}^{\iota} (\xi^{\mathfrak{b}}).
\end{align}
Now define
\begin{align}\label{eq:A-def}
	\mathbb{A}^{\mathfrak{b}}_{\mathfrak{a}} (v^{\mathfrak{b}}, \dot{G}_{t})
	= \sum_{\iota\in \mathcal{A}_{\Delta}(\mathfrak{b})\setminus\{\mathtt{d}\}}\mathsf{R}_{\mathfrak{b}_{\iota} \to \mathfrak{a}} A_{\iota}^{\mathfrak{b}}.
\end{align}

\begin{lemma}\label{lem:A-reconstruction}
	The coefficients defined in \eqref{eq:A-def} satisfy the triangularity condition \eqref{eq:A,B-triangular} and
	\begin{align}\label{eq:A-reconstruction}
		\sum_{\mathfrak{a}\in \mathfrak{A}_{\tmop{adm}}} \int \rho(\mathd \xi^{\mathfrak{a}})\mathbb{A}^{\mathfrak{b}}_{\mathfrak{a}}(v^{\mathfrak{b}}, \dot{G}_t )(\xi^{\mathfrak{a}} ) \Psi^{\mathfrak{a}}(\xi^{\mathfrak{a}} )
		= A^{\mathfrak{b}}(v^{\mathfrak{b}}, \dot{G}_t)
		- \int \rho(\mathd \xi^{\mathfrak{b}}) \dot{\mathbb{G}}_{t}(\xi^{\mathfrak{b}}) v^{\mathfrak{b}}(\xi^{\mathfrak{b}}) \Psi^{\mathfrak{b}}(\xi^{\mathfrak{b}}).
	\end{align}
	Moreover, \(\mathbb{A}^{\mathfrak{b}}_{\mathfrak{a}}  \neq 0 \) only if \(\mathfrak{a}\in \mathfrak{A}_{\tmop{irr}} \cup \mathfrak{A}_{\tmop{ren}}\).
\end{lemma}
\begin{proof}
	The conditions on \(L\) and \(Q\) are trivially satisfied, as none of the operations applied affect charge or loop order.
	Since \(\iota\in \mathcal{A}_{\Delta}(\mathfrak{b})\setminus \{\mathtt{d}\}\), \(N(\mathfrak{a}) \succ N(\mathfrak{b})\), and since we never apply a pairing operation to create a new cluster, we obtain \eqref{eq:A,B-triangular}.
	The identity \eqref{eq:A-reconstruction} follows directly from Lemma \ref{lem:dual-P-R-reconstruction}.
	Finally, the last claim follows from the definition of the raising maps \(\mathcal{R}\) and \(\mathsf{R}\).
\end{proof}
\paragraph{Estimates for \(\mathbb{A}\)}
\begin{lemma}\label{lem:A-estimates}
	Let \(t>0\) and let
	\((v^{\mathfrak d})_{\mathfrak d\in\mathfrak A_{\tmop{adm}}}\), with
	\(v^{\mathfrak d}\in\mathcal K_t^{\mathfrak d}\), be compatible with charge conjugation.
	For \(\mathfrak{a}, \mathfrak{b}\in \mathfrak{A}_{\tmop{adm}}\),
	the coefficient defined in \eqref{eq:A-def} satisfies
	\begin{align}
		\tnorm{\mathbb{A}^{\mathfrak{b}}_{\mathfrak{a}}(v^{\mathfrak{b}}; \dot{G}_t)}_{t, \mathfrak{a}}
		\lesssim
		t^{-1-[\mathfrak{a}] + [\mathfrak{b}]} \mathe^{-\frac{m^{2}}{t}} \tnorm{v^{\mathfrak{b}}}_{t, \mathfrak{b}}.
	\end{align}
	Moreover, \eqref{eq:A,B-charge-conjugation} holds for \(\mathbb{A}\).
\end{lemma}

\begin{proof}
	Mirroring the proof of Lemma \ref{lem:B-estimate}, we show the estimate for each summand.
	Since there are only finitely many terms, this proves the estimate for the sum after adjusting the constant.
	Fix \(\iota\in \mathcal{A}_{\Delta}(\mathfrak{b})\setminus\{\mathtt d\}\).
	As there are only the cases \(\iota= (\mathtt{r}, (2n,k))\) and \(\iota= (\mathtt{rr}, (2n', k'), (2n'', k''))\), we will treat the cases separately instead of introducing notation to treat the cases simultaneously.
	By definition of the rates \((r(\mathfrak{d}))_{\mathfrak{d}\in \mathfrak{A}_{\tmop{adm}}}\),
	\(r(\mathfrak{a}) = r(\mathfrak{b}_{\iota}) < r(\mathfrak{b})\).
	Define
	\begin{align}\label{eq:def-riota}
		r_{\iota} = \frac{1}{2} \left(r(\mathfrak{a}) + r(\mathfrak{b})\right)
		\qquad
		\varepsilon_{\iota} = r_{\iota} - r(\mathfrak{a})
		= \frac{1}{2} \left(r(\mathfrak{b}) - r(\mathfrak{a})\right)>0.
	\end{align}
	This ensures
	\begin{align}\label{eq:r-iota-absorbtion}
		\frac{\omega_{t}^{r_{\iota}}(\xi^{\mathfrak{b}_{\iota}}) }{\omega_{t}^{r(\mathfrak{b})}(\xi^{\mathfrak{b}})} \tmop{St}(\xi^{\mathfrak{b}})^{\gamma}
		\lesssim t^{-\gamma/2}.
	\end{align}
	We first show the estimate on the raw kernel \eqref{eq:def-A-raw}:
	\begin{align}\label{eq:A-raw-est}
		\tnorm{ A^{\mathfrak{b}}_{\iota}}_{t, \mathfrak{b}_{\iota}; r_{\iota}}
		\lesssim
		t^{-1 + \frac{1}{2}( K(\mathfrak{b})- K(\mathfrak{b}_{\iota}))} \mathe^{-\frac{m^{2}}{t}} \tnorm{v^{\mathfrak{b}}}_{t, \mathfrak{b}}.
	\end{align}
	With Lemma \ref{lem:dual-R-kernel-estimate}, we obtain
	\begin{align}
		\tnorm{\mathsf{R}_{\mathfrak{b}_{\iota} \to \mathfrak{a}} A^{\mathfrak{b}}_{\iota}}_{t, \mathfrak{a}}
		\lesssim t^{-\frac{1}{2}(K(\mathfrak{a})-K(\mathfrak{b}_{\iota}))}  \tnorm{ A^{\mathfrak{b}}_{\iota}}_{t, \mathfrak{b}_{\iota}; r_{\iota}},
	\end{align}
	which implies the claim.

	It remains to prove \eqref{eq:A-raw-est}.
	We again distinguish two cases.

	\emph{Case 1. Single cluster removed \(\iota = (\mathtt{r}, 2n, k)\).}
	If \(m(\mathfrak b_{\iota})>0\), then \(\mathfrak b\) does not consist solely of the removed neutral cluster.
	The admissibility condition therefore gives \(\nu_{\iota}^{\mathfrak{b}} = \nu_{2n, k}^{\mathfrak{b}} \leqslant (\alpha_{2n}-1)\vee 0 \in [0,1]\).
	Hence, either \(\nu_{2n, k}({\mathfrak{b}}) = K(\mathfrak{b})- K(\mathfrak{b}_{\iota}) \in [0,1]\), so that \(\mathfrak{b}\in \mathfrak{A}_{\tmop{adm}}^{\circ}\), or \(m(\mathfrak b_{\iota})=0\).

	If \(m(\mathfrak b_{\iota})=0\), then \(\xi^{\mathfrak b}\) is neutral and \eqref{eq:def-Gbb} gives
	\begin{align}
		\dot{\mathbb{G}}_{t}(\xi^{\mathfrak{b}})
		= -\frac{\beta^{2}}{2}\,(S_{\xi^{\mathfrak b}}\otimes S_{\xi^{\mathfrak b}})[\dot G_t].
	\end{align}
	Thus, using Lemma \ref{lem:hk-multipoint} and \eqref{eq:r-iota-absorbtion},
	\begin{align}
		\abs{\mathsf{G}_{t, \mathfrak{b}}^{\iota}(\xi^{\mathfrak{b}} )}
		= \abs{\dot{\mathbb{G}}_{t}(\xi^{\mathfrak{b}})}  \delta(\xi^{\mathfrak{b}})^{-K(\mathfrak{b})}
		\lesssim t^{-1 + \frac{1}{2}\nu_{2n, k}(\mathfrak{b})} \mathe^{-\frac{m^{2}}{t}} \omega_{t}^{\mathfrak{b}}(\xi^{\mathfrak{b}}).
	\end{align}
	Inserting this into the definition of \(A^{\mathfrak{b}}_{\iota}\) and using that \(\xi^{\mathfrak{b}_{\iota}} = \varnothing \), we arrive at \eqref{eq:A-raw-est}.

	If \(m(\mathfrak b_{\iota})>0\), then \(\mathfrak{b}\in \mathfrak{A}_{\tmop{adm}}^{\circ} \) and, as in \eqref{eq:G-branch-estimate}, we have
	\begin{align}
		\abs{\mathsf{G}^{\iota}_{t, \mathfrak{b}}(\xi^{\mathfrak{b}} ) }
		\lesssim t^{-1+\frac{1}{2} (K(\mathfrak{b}) - K(\mathfrak{b}_{\iota}))} \mathe^{-\frac{m^{2}}{t}}
		\frac{\omega_{t}^{r(\mathfrak{b})} (\xi^{\mathfrak{b}})}{\omega^{r_{\iota}}_{t}(\xi^{\mathfrak{b}_{\iota}})},
	\end{align}
	which again gives the claim after inserting this into \eqref{eq:def-A-raw}.

	\emph{Case 2. Two removed clusters, \(\iota=(\mathtt{rr}, (2n',k'), (2n'',k''))\).}
	In this case \(\mathfrak{b}\in \mathfrak{A}_{\tmop{adm}}^{\circ}\).
	For \(K_1,K_2\geqslant1\) and \(H\in C^{K_1,K_2}(\mathbb R^{d_1}\times\mathbb R^{d_2})\), denote the mixed Taylor remainder around \((x_1,x_2)\) by
		{\small
			\begin{align}\label{eq:Taylo-rem-two-var}
				R_{x_1,x_2}^{K_1,K_2}[H](y_1,y_2)
				 & :=(I-L_{x_1,1}^{K_1-1})(I-L_{x_2,2}^{K_2-1})H(y_1,y_2),
				\\
				 & =H(y_1,y_2)
				-\sum_{|\alpha_1|<K_1}
				\frac{\partial_1^{\alpha_1}H(x_1,y_2)
					[(y_1-x_1)^{\otimes\alpha_1}]}{\alpha_1!}
				\nonumber                                                  \\
				 & \quad
				-\sum_{|\alpha_2|<K_2}
				\frac{\partial_2^{\alpha_2}H(y_1,x_2)
					[(y_2-x_2)^{\otimes\alpha_2}]}{\alpha_2!}
				\nonumber                                                  \\
				 & \quad
				+\sum_{\substack{|\alpha_1|<K_1                            \\|\alpha_2|<K_2}}
				\frac{\partial_1^{\alpha_1}\partial_2^{\alpha_2}H(x_1,x_2)
					[(y_1-x_1)^{\otimes\alpha_1},
							(y_2-x_2)^{\otimes\alpha_2}]}
				{\alpha_1!\,\alpha_2!}.
			\end{align} }
	Here \(L_{x_i,i}^{K_i-1}\) denotes the projection onto the Taylor jet of order \(K_i-1\) in the \(i\)-th variable.
	A Taylor expansion, again using neutrality in each component, gives
	\begin{align}
		\dot{\mathbb{G}}_{t} (\xi^{\mathfrak{b}, (2n', k')}, \xi^{\mathfrak{b}, (2n'', k'' )})
		= - \beta^{2} R_{x_1, x_2}^{1,1}[\dot G_t](\xi^{\mathfrak{b}, (2n', k')}, \xi^{\mathfrak{b}, (2n'', k'' )}),
	\end{align}
	and thus by the heat kernel estimates \eqref{eq:hk-taylor-estimate},
	\begin{align}
		\abs {\mathsf{G}_{t, \mathfrak{b}}^{\iota}}
		\lesssim  t^{-1+\frac{1}{2} (K(\mathfrak{b})- K(\mathfrak{b}_{\iota}))}\mathe^{-\frac{m^{2}}{t}}
		\frac{\omega_{t}^{r(\mathfrak{b})} (\xi^{\mathfrak{b}})}{\omega^{r_{\iota}}_{t}(\xi^{\mathfrak{b}_{\iota}})},
	\end{align}
	which again yields the claim.
	The compatibility with charge conjugation is true by the symmetry of \(\mathbb{G}\) for the unprocessed kernel \(A\) and thus also for \(\mathbb{A}\).
\end{proof}

\paragraph{Estimates for \(\Gamma\)}
The estimate \eqref{eq:sG-A,B-est} on \(\Gamma\) is an immediate consequence of the screening lemma proved in Appendix \ref{app:screening-proof} and is analogous to the contraction effect of the charged terms in Section~\ref{sec:ell=2}.
\begin{lemma}\label{lem:gamma-est}
	For any \(\mathfrak a\in\mathfrak A_{\tmop{adm}}\) and \(s\geqslant t\geqslant0\), the bound from \eqref{eq:sG-A,B-est},
	\begin{align}\label{eq:gamma-est}
		\norm{\Gamma_{t,s}^{\mathfrak{a}} (\xi^{\mathfrak{a}} )}_{L^{\infty}(\mathd x^{\mathfrak{a}})}
		\leqslant 1\wedge C_{\mathfrak{a}} \langle t\rangle^{[\mathfrak{a}]_{Q}-[\mathfrak{a}]} \langle s\rangle^{-([\mathfrak{a}]_{Q}-[\mathfrak{a}])},
	\end{align}
	holds.
\end{lemma}
\begin{proof}
	Combining the screening estimate in Lemma \ref{lem:screening} with the definition of \([\cdot]_{Q}\) in \eqref{eq:def-scale-mfra} and \(\bar{\lambda}_{t}\) in \eqref{eq:def-lambda-t} yields \eqref{eq:gamma-est}.
\end{proof}

\subsection{Reconstruction of the flow equation}

\begin{lemma}\label{lem:va-flow-Va-flow}
	The coefficients \(\mathbb{A}\) and \(\mathbb{B}\), as defined in \eqref{eq:A-def} and \eqref{eq:B-def}, satisfy \eqref{eq:B-Va} and \eqref{eq:A,Gamma-Va}.
	In particular,
	\begin{align}\label{eq:va-Va-equivalent}
		\sum_{\mathfrak{a}\in \mathfrak{A}_{\tmop{adm}}^{\ell}} \int \rho(\mathd \xi^{\mathfrak{a}})\,\partial_{t}v_{t}^{\mathfrak{a}} (\xi^{\mathfrak{a}}) \Psi^{\mathfrak{a}}(\xi^{\mathfrak{a}}) + \sum_{\mathfrak{b}\in \mathfrak{A}_{\tmop{adm}}^{\ell}} A^{\mathfrak{b}}   (v^{\mathfrak{b}}_{t}, \dot{G}_{t})
		= \sum_{\ell'+\ell''=\ell}\sum_{\substack{\mathfrak{b}\in \mathfrak{A}_{\tmop{adm}}^{\ell'} \\ \mathfrak{c}\in \mathfrak{A}_{\tmop{adm}}^{\ell''}}} B^{\mathfrak{b}, \mathfrak{c}}(v^{\mathfrak{b}}_{t}, v^{\mathfrak{c}}_{t}, \dot{G}_{t}).
	\end{align}
	If \((v^{\mathfrak{a}})_{\mathfrak{a}}\) satisfies the initial condition \eqref{eq:va-init} and, for \(L(\mathfrak a)>1\), \eqref{eq:sG-va-full-flow}, then \(V^{[\ell ]}\) defined via \eqref{eq:V-ell-from-V-a} satisfies \eqref{eq:V-ell}.
\end{lemma}
\begin{proof}
	The first statement follows from \eqref{eq:B-reconstruction} and \eqref{eq:A-reconstruction}.
	Together with \eqref{eq:sG-va-full-flow}, these identities yield
	\eqref{eq:va-Va-equivalent}, and hence \eqref{eq:V-ell}, which is the second claim.
\end{proof}

\appendix
\section{ \(T\)-dependent estimates on the flow equation}\label{app:T-dependent}
In this Appendix, we prove Proposition \ref{prop:T-dependent-est}.
The estimates follow from the same steps as for the main \(T\)-independent estimates, except for the application of the raising map, which improves the scaling of the kernels \(v^{\mathfrak{a}} \) at the cost of introducing field dependence in the bounds.

The proof is by induction on \(\ell<\ell^\ast\), starting from \(\ell=1\).
Define
\begin{align}
	\mathtt{w}^{\mathfrak{a}} (\xi^{\mathfrak{a}})
	\assign
	\prod_{(2n,k)\in I_{\mathrm n}(\mathfrak{a})} \delta(\xi^{\mathfrak{a}, 2n, k})^{-\nu_{2n, k}(\mathfrak{a})}.
\end{align}
\par\noindent
Define the maps
\begin{align}\label{eq:def-W-map}
	(\mathsf{W}_{\mathfrak{a}} v)(\xi^{\mathfrak{a}}) \assign \mathtt{w}^{\mathfrak{a}}(\xi^{\mathfrak{a}}) v(\xi^{\mathfrak{a}}).
\end{align}
Then,
\begin{align}
	V^{\mathfrak{a}}[\varphi]
	= \int \rho(\mathd \xi^{\mathfrak{a}}) v^{\mathfrak{a}}(\xi^{\mathfrak{a}} ) \Psi^{\mathfrak{a}} (\xi^{\mathfrak{a}})
	= \int \rho(\mathd \xi^{\mathfrak{a}}) (\mathsf{W}_{\mathfrak{a}}v^{\mathfrak{a}} ) (\Psi^{\mathfrak{a}} (\mathtt{w}^{\mathfrak{a}})^{-1})(\xi^{\mathfrak{a}}).
\end{align}
By definition of the weight,
\begin{align}
	\norm{\Psi^{\mathfrak{a}} (\mathtt{w}^{\mathfrak{a}})^{-1} }_{L^\infty} \leqslant 2^{c(\mathfrak a)},
	\quad
	\norm{(\Psi^{\mathfrak{a}} - \tilde{\Psi}^{\mathfrak{a}}) (\mathtt{w}^{\mathfrak{a}})^{-1}}_{L^\infty}
	\leqslant C_{\beta,\mathfrak a} \norm{\varphi - \tilde{\varphi}}_{L^\infty}.
\end{align}
Therefore, Proposition \ref{prop:T-dependent-est} follows from the uniform bounds on the family \((\mathsf{W}_{\mathfrak{a}} v^{\mathfrak{a}} )_{\mathfrak{a}\in \mathfrak{A}_{\tmop{adm}}^{<\ell^{\ast}}}\) in the next Lemma.
\begin{lemma}\label{lem:T-dependent-est}
	Let \(T<\infty\).
	There is a constant \(C_{T, \rho}\) such that the family \((v^{\mathfrak{a}})_{\mathfrak{a}\in \mathfrak{A}_{\tmop{adm}}^{<\ell^{\ast}}} = (v^{\mathfrak{a}, \rho, T})_{\mathfrak{a}\in \mathfrak{A}_{\tmop{adm}}^{<\ell^{\ast} }}\) constructed in Corollary \ref{cor:va-scaling} satisfies
	\begin{align}
		\sup_{t\in [0,T]} \langle t\rangle^{[\mathfrak{a}]} \tnorm{\mathsf{W}_{\mathfrak{a}} v_t^{\mathfrak{a}}}_{\mathfrak{a}} \leqslant C_{T, \rho}.
	\end{align}
\end{lemma}
\begin{proof}
	By definition of \(\mathsf{P}\) and \(\mathsf{R}\),
	\begin{align}
		\mathsf{W}_{\mathfrak{a}} \mathsf{R}_{\mathfrak{d} \to \mathfrak{a}} v
		 & = \frac{1}{\abs{\mathcal{R}(\mathfrak{d})}} \mathsf{W}_{\mathfrak{d}} v, \quad \mathfrak{a}\in \mathcal{R}(\mathfrak{d}),    \\
		\mathsf{W}_{\mathfrak{a}} \mathsf{P}_{\mathfrak{d} \to \mathfrak{a}} v
		 & = \mathsf{P}_{\mathfrak{d} \to \mathfrak{a}} (\mathsf{W}_{\mathfrak{d}} v), \quad \mathfrak{a}\in \mathcal{P}(\mathfrak{d}).
	\end{align}
	Moreover, it is straightforward to check that
	\begin{align}
		\tnorm{\mathsf{W}_{\mathfrak{a}} \mathsf{R}_{\mathfrak{d} \to \mathfrak{a}} v}_{\mathfrak{a}}
		= \frac{1}{\abs{\mathcal{R}(\mathfrak{d})}} \tnorm{ \mathsf{W}_{\mathfrak{d}} v }_{\mathfrak{d}},
		\qquad
		\tnorm{\mathsf{W}_{\mathfrak{a}} \mathsf{P}_{\mathfrak{d} \to \mathfrak{a}} v}_{\mathfrak{a}}
		\leqslant C_{\rho, \mathfrak{d}} \tnorm{\mathsf{W}_{\mathfrak{d}} v}_{\mathfrak{d}}.
	\end{align}
	Combined with the exact representations for the unprocessed coefficients \eqref{eq:def-A-raw} and \eqref{eq:def-B-raw}, this implies
	\begin{align}\label{eq:A,B-T-dependent}
		\tnorm{\mathsf{W}_{\mathfrak{a}} \mathbb{A}^{\mathfrak{b}}_{\mathfrak{a}}(v, \dot{G}_t)}_{\mathfrak{a}}
		 & \lesssim \langle t \rangle^{1} \tnorm{\mathsf{W}_{\mathfrak{b}} v}_{\mathfrak{b}} \\
		\tnorm{\mathsf{W}_{\mathfrak{a}}\mathbb{B}^{\mathfrak{b}, \mathfrak{c}}_{\mathfrak{a}}(v^{\mathfrak{b}}, v^{\mathfrak{c}}, \dot{G}_{t})}_{\mathfrak{a}}
		 & \lesssim \langle t \rangle^{-2}
		\tnorm{\mathsf{W}_{\mathfrak{b}} v^{\mathfrak{b}} }_{\mathfrak{b}}
		\tnorm{\mathsf{W}_{\mathfrak{c}} v^{\mathfrak{c}}}_{\mathfrak{c}}.
	\end{align}

	To start the induction, first note that for \(\mathfrak{a}\in \mathfrak{A}_{\tmop{adm}}^{1} \) such that \(v^{\mathfrak{a}} \neq 0 \), the operator \(\mathsf{W}\) acts trivially, i.e. \(\mathsf{W}_{\mathfrak{a}} = \tmop{id}\).
	Hence, for every \(t\in [0,T]\),
	\begin{align}
		\tnorm{\mathsf{W}^{\mathfrak{a}} v^{\mathfrak{a}}_t}_{\mathfrak{a}} 
		= \tnorm{v^{\mathfrak{a}}_t}_{\mathfrak{a}} 
		\leqslant \frac{\bar{\lambda}_{t}}{2}
		\leqslant \frac{\bar{\lambda}_{T}}{2}.
	\end{align}
	We can now proceed as in the proof of Corollary \ref{cor:va-scaling} to propagate these bounds using \(\abs{\Gamma}_{t, s}\leqslant1\), the estimates \eqref{eq:A,B-T-dependent}, and \(\int_0^T C\,\mathd s\leqslant CT\).
\end{proof}

\section{Heat kernel estimates}\label{app:hk-estimates}
\subsection{General estimates}
All the estimates are formulated in terms of a general (unnormalized) Gaussian kernel \(p_t:\mathbb{R}^d\to\mathbb{R}\) defined for \(d\in\mathbb{N}\), \(t>0\), and \(c>0\) as
\begin{align}\label{eq:def-pt}
	p_t(x)=p_{t}(x; c) = \mathe^{-c t|x|^2},\qquad (t,x) \in(0,\infty)\times \mathbb{R}^d.
\end{align}
The implicit constants in the estimates in this section generally depend on \(c>0\) even when it is suppressed in the notation.
Since we only use the estimates for \(\dot{G}, Q_t\), this does not lead to any ambiguity.
For some statements it will be more convenient to use the notation
\begin{align}\label{eq:def-qt}
	q_{t}(x, y) = p_t(x-y).
\end{align}
The kernels of \(\dot{G}_{t}\) and \(Q_{t}\), which are defined in \eqref{eq:def-Q}, can be written in terms of \(p_{t}\) defined above,
\begin{align}\label{eq:G,Q-as-pt}
	\dot{G}_{t}(x) = \frac{1}{4\pi t} \mathe^{-\frac{m^{2}}{t}} p_{t}\left(x; \frac{1}{4}\right), \qquad
	Q_{t}(x) = \frac{1}{2\pi} \mathe^{-\frac{m^{2}}{2t}} p_{t}\left(x; \frac{1}{2}\right).
\end{align}

\subsubsection{Pointwise estimates}
\begin{lemma}\label{lem:G-pw}
	There are uniformly bounded functions \(g_1, g_2\) such that
	\begin{align}\label{eq:G-log-estimates}
		G_{t}(0)         & = \frac{1}{4\pi}\log(t\vee 1) + g_1(t),                       \\
		G_{\infty, t}(x) & = \frac{1}{4\pi} \log(\abs{x}^{-2} t^{-1}\vee 1) + g_2(t, x).
	\end{align}
\end{lemma}
\begin{proof}
	See \cite[Lemma A.1]{gubinelliFBDSEApproachSine2026}.
\end{proof}

\begin{lemma}\label{lem:hk-scaling}
	For any multi-index \(\alpha\in\mathbb{N}^d\), there exists a constant \(C_{\alpha,d,c}>0\) (depending on \(c,\alpha\) and \(d\)) such that for all \(x\in\mathbb{R}^d\)

	\begin{align}\label{eq:dalpha-hk-pw}
		\big|\nabla^\alpha p_t(x)\big|
		\leqslant
		C_{\alpha,d, c}\,(\sqrt{t})^{|\alpha|}\big(1+\sqrt{t}\,|x|\big)^{|\alpha|} \mathe^{-ct|x|^2}.
	\end{align}
\end{lemma}
\begin{proof}
	Because \(p_t(x)=\prod_{j=1}^d \mathe^{-ct x_j^2}\), differentiation in each coordinate separates. Thus there are polynomials \(P_{\alpha,c}\), with coefficients depending on \(c\) and total degree \(\abs{\alpha}\), such that
	\begin{align}
		\nabla^\alpha p_t(x)
		= t^{\frac{|\alpha|}{2}} P_{\alpha,c}(\sqrt{t}\,x)\,\mathe^{-ct|x|^2}.
	\end{align}
	Taking absolute values and using the elementary bound for a polynomial of degree \(\leqslant|\alpha|\),
	\begin{align}
		|P_{\alpha,c}(y)| \leqslant C_{\alpha,d,c}\,(1+|y|)^{|\alpha|}.
	\end{align}
	we obtain the claim after adjusting the constant.
\end{proof}

\begin{lemma}\label{lem:hk-taylor-remainder}
	Recall the definition of the Taylor remainder \eqref{eq:Taylor-rem-single-var}.
	Let \(b>0\) and \(K\geqslant1\).
	There exists a constant \(C>0\) depending only on \(K,b,d,c\) such that for all \(x,y\in\mathbb{R}^d\) and all \(t>0\),
	\begin{align}\label{eq:hk-taylor-estimate}
		\big|R^{K}_{x}[p_t](y)\big|
		\lesssim \;\sqrt{t}^{K} \;|x-y|^{K}\;
		\big(1+\sqrt{t}|x|+\sqrt{t}|x-y|\big)^{K}\;\mathe^{-b\sqrt{t}|x|}\mathe^{b\sqrt{t}|x-y|}.
	\end{align}
\end{lemma}
\begin{proof}
	Use the integral form of the multivariate Taylor remainder,
	\begin{align}
		R^{K}_{x}[p_t](y)
		=\sum_{|\alpha|=K}\frac{K}{\alpha!}\int_0^1(1-s)^{K-1}\nabla^{\alpha}p_t\big(x+s(y-x)\big)\,[(y-x)^{\otimes\alpha}]\,\mathd s.
	\end{align}
	Taking absolute values and bounding the finite combinatorial sum by a constant yields
	\begin{align}
		\big|R^{K}_{x}[p_t](y)\big|
		\lesssim \;|x-y|^{K}\int_0^1
		\sup_{|\alpha|=K}\norm{\nabla^{\alpha}p_t(x+s(y-x))}_{\tmop{op}}\,\mathd s.
	\end{align}
	From Lemma \ref{lem:hk-scaling},
	\begin{align}
		\big|R^{K}_{x}[p_t](y)\big|
		\lesssim \sqrt{t}^{K} |x-y|^{K}\,
		\sup_{s\in[0,1]}(1+\sqrt t|x|+\sqrt t|x-y|)^{K}
		\sup_{s\in[0,1]}\mathe^{-ct|x+s(y-x)|^2}.
	\end{align}
	where we also used \(|x+s(y-x)|\le |x|+|x-y|\) for \(s\in [0,1]\).
	To make use of the decay in \(\abs{x}\), use that  \(\exp(-cr^{2})\leqslant \exp(\frac{b^{2}}{4c})\exp(-b r)\),
	and that for any \(s\in[0,1]\) we have \(|x+s(y-x)|\geqslant |x|-|x-y|\), to write
	\begin{align}
		\mathe^{-ct|x+s(y-x)|^2}
		\lesssim \mathe^{-b\sqrt{t}|x+s(y-x)|}
		\le \mathe^{-b\sqrt{t}|x|}\,\mathe^{b\sqrt{t} |y-x|}.
	\end{align}
	which is the claim.
\end{proof}

\subsubsection{Integral estimates}

\begin{lemma}\label{lem:hk-Lp-bounds}
	For \(d\in\mathbb{N}\), \(p\in[1,\infty)\), \(k\in\mathbb{R}\), \(\gamma>0\), and a multi-index \(\alpha\), set \(k_+:=\max\{k,0\}\). Then
	\begin{align}\label{eq:hk-Lp-estimates}
		\norm{\abs{x}^{\frac{\gamma}{p}} \nabla^{\alpha} p_{t}}_{L^{p}(\langle x\rangle^{k})}
		\lesssim t^{-\frac{d+\gamma}{2p}}t^{\frac{\abs{\alpha}}{2}}
		\langle t^{-1/2}\rangle^{k_+}.
	\end{align}
\end{lemma}
\begin{proof}
	See, e.g., \cite[Lemma A.4]{gubinelliFBDSEApproachSine2026}.
\end{proof}

\begin{lemma}\label{lem:hk-regularity}
	For any \(\vartheta>0\) and \(k\in \mathbb R\), it holds that
	\begin{align}\label{eq:hk-regularity}
		\norm{\int_{0}^{\infty} Q_s u_s \mathd s }_{H^{\vartheta }(\langle x\rangle^{k})}^{2}
		\lesssim
		\int_{0}^{\infty} \norm{u_s}^{2}_{L^{2}(\langle x\rangle^{k})} \langle s \rangle^{\vartheta-1}  \mathd s.
	\end{align}
\end{lemma}
\begin{proof}
	See, e.g., \cite[Lemma A.5]{gubinelliFBDSEApproachSine2026}.
\end{proof}

\subsection{Screening bound}\label{app:screening-proof}

\begin{lemma}[Screening bound]\label{lem:screening}
	Let \(\xi_{1:M}\in \mathcal{X}^{M}\) be a point charge configuration
	with charge \(q=q(\xi_{1:M})\).
	For all \(s>t\), it holds that
	\begin{align}\label{eq:Gamma-Screening}
		\sup_{x_{1:M}\in \mathbb{R}^{2M}}\Gamma_{t,s}(\xi_{1:M})
		\leqslant 1\wedge C_M\exp\left(-\abs q\,\frac{\beta^2}{2}G_{s,t}(0)\right),\qquad s>t.
	\end{align}
\end{lemma}
\begin{proof}
	First, \(\Gamma_{t,s}(\xi_{I})\leqslant1\) always holds since \(G\) is positive definite. 

	If \(q:=q(\xi_{1:M})= \pm M\), then this follows directly from \(\sigma_k \sigma_j = 1\) for all \(j,k\) and
	\begin{align}\label{eq:purely-charged-gamma}
		\Gamma_{t,s }(\xi_{1:M})
		=         & \exp \left(-M\frac{\beta^{2}}{2} G_{s,t}(0) - \beta^{2} \sum_{k<j} G_{s,t}(x_{k}-x_{j})\right) \\
		\leqslant & \exp \left(-M\frac{\beta^{2}}{2} G_{s,t}(0)\right)
		= \exp\left(-\abs q\,\frac{\beta^2}{2}G_{s,t}(0)\right).
	\end{align}
	Otherwise, if \(\abs{q}<M\), we will show that we can remove neutral clusters recursively until we arrive at a contribution \(\tilde\xi_{1:\abs{q}}\) that is purely charged (i.e. \(\abs{q(\tilde{\xi}_{1:\abs {q}})}=\abs{q}\)) and such that
	\begin{align}\label{eq:gamma-reduction}
		\Gamma_{t,s }(\xi_{1:M})  \lesssim_{M} \Gamma_{t,s}(\tilde{\xi}_{1:\abs{q}}).
	\end{align}
	Combined with \eqref{eq:purely-charged-gamma}, this implies the claim.
	To prove \eqref{eq:gamma-reduction}, we first note that if \(\abs{q} \neq M\), then there must be some \(k,j\) such that \(\sigma_{k}\sigma_{j}=-1\).
	Let \(k_+, k_-\) be the closest neutral pair in \(\xi_{1:M}\), that is
	\begin{align}\label{eq:closest-neutral}
		\abs{x_{k_+}-x_{k_-}} = \min_{j,k: \: \sigma_j \sigma_k=-1} \abs{x_k-x_j}.
	\end{align}
	By possibly relabelling the points, we assume that \(\sigma_{1},\dots, \sigma_{\abs{q}}=\tmop{sign}(q)\), and \(\sigma_{\abs{q}+k}=-\sigma_{\abs{q}+k+1}\) for every \(k=1,\dots, M-\abs{q}-1\).
	Let \(k_+\neq k_-\in\{M-1,M\}\) and define
	\begin{align}\label{eq:def-eta}
		\eta := ((+1, y_+), (-1,y_-)) =  ((+1, x_{k_+}), (-1,x_{k_-})).
	\end{align}
	We can then write
	\begin{align}\label{eq:eta-extraction}
		\Gamma_{t,s}(\xi_{1:M}) = \Gamma_{t,s }(\xi_{1:M-2}) \Gamma_{t, s} (\xi_{1:M-2}; \eta) \Gamma_{t,s} (\eta).
	\end{align}
	We claim that there is a constant \(C_{M}\) depending only on the number of points \(M\) such that
	\begin{align}\label{eq:eta-extraction-bounded}
		\Gamma_{t, s} (\xi_{1:M-2}; \eta) \Gamma_{t,s} (\eta) \leqslant C_{M}.
	\end{align}
	Assuming for the moment that \eqref{eq:eta-extraction-bounded} holds, we can repeat the procedure \eqref{eq:eta-extraction} \(n=\frac{M-\abs{q}}{2}\) times,
	until all neutral pairs are removed, and we obtain \(\tilde{\xi}_{1:\abs{q}}\) as in \eqref{eq:gamma-reduction}.
	To see that \eqref{eq:eta-extraction-bounded} holds, first use \(\Gamma_{t,s}(\eta)\leqslant1\).
	Since this factor is at most one, it remains to bound the cross term, whose logarithm is
	\begin{align}\label{eq:gamma-extracted-log}
		\log\Gamma_{t, s}(\xi_{1:M-2};\eta)
		=-\beta^{2}\sum_{k=1}^{M-2}\sigma_k
		\left(G_{s,t}(x_k-y_+)-G_{s,t}(x_k-y_-)\right).
	\end{align}
	Now fix \(k\in[M-2]\) and write \(x=x_k\).
	Without loss of generality, suppose that \(\sigma_k=+1\); otherwise, exchange the roles of \(y_+\) and \(y_-\).
	Let \(r:=\abs{y_+-y_-}\). If \(r=0\), then the contribution vanishes.
	Assume therefore that \(r>0\). 
	If \(\abs{x-y_+}\leqslant\abs{x-y_-}\), then its contribution to \eqref{eq:gamma-extracted-log} is nonpositive.
	It remains to consider \(\abs{x-y_+}>\abs{x-y_-}\).
	Since \(x\) and \(y_-\) have opposite charges and \(\eta\) is a closest neutral pair,
	\(\abs{x-y_-}\geqslant r\), and hence also \(\abs{x-y_+}\geqslant r\).
	We split the scales \([t,s]\) into low and high regions according to \(r\)
	\begin{align}\label{eq:T-low-high-decomp}
		  & G_{s,t}(x-y_{+})-G_{s,t}(x-y_{-})                                                                      \\
		= & \int_{t}^{s} (\dot{G}_{u}(x-y_{+})- \dot{G}_{u}(x-y_{-})) \mathd {u}                                   \\
		= & \int_{t}^{s} \mathbb{1}_{\{u\leqslant r^{-2}\}}(\dot{G}_{u}(x-y_{+})- \dot{G}_{u}(x-y_{-})) \mathd {u}
		+ \int_{t}^{s} \mathbb{1}_{\{u\geq r^{-2}\}}(\dot{G}_{u}(x-y_{+})- \dot{G}_{u}(x-y_{-})) \mathd {u}        \\
		= & I_{\tmop{low} } + I_{\tmop{high} }.
	\end{align}
	In the low region, the Gaussian is very flat and the mean value theorem is a good approximation.
	We estimate
	\begin{align}\label{eq:low-est}
		\abs{I_{\tmop{low} }}
		\leqslant r \int_{t}^{s} \mathbb{1}_{\{u\leqslant r^{-2}\}}\norm{\nabla \dot{G}_{u}}_{L^{\infty}}
		\lesssim r (s\wedge r^{-2})^{\frac{1}{2}}
		\lesssim 1.
	\end{align}
	In the high region, the Gaussian is no longer flat, but instead offers exponential decay.
	Since \(\abs{x-y_+},\abs{x-y_-}\geq r\), this decay can be quantified entirely in terms of \(r\), and
	\begin{align}\label{eq:high-est}
		\abs{I_{\tmop{high} }}
		 & \leqslant 2\int_{t}^{s} \mathbb{1}_{\{u\geqslant r^{-2}\}}\abs{\dot{G}_{u}(y_{-}-y_{+})}\mathd {u} \\
		 & \lesssim \int_{t}^{s} \mathbb{1}_{\{u\geqslant r^{-2}\}}u^{-1} \exp(-u r^{2}-\frac{m^2}{u})        \\
		 & \lesssim 1 + \int_{tr^{2}\vee 1}^{\infty}  \frac{\mathe^{-w}}{w} \mathd {w}                        \\
		 & \lesssim 1.
	\end{align}
	Applying this argument for every \(k\in[M-2]\), and summing over \(k\), we see that
	\(\Gamma_{t,s}(\xi_{1:M-2};\eta)\leqslant C_M\).
\end{proof}
\begin{remark}
	\begin{enumerate}[a)]
		\item The relation \eqref{eq:gamma-reduction} expresses a ``screening'' effect:
		      a cloud or cluster with overall neutral charge is effectively screened, producing only a weak long-range effect on the remaining particles.
		      Therefore, for the global contraction coming from the linear propagator, the global charge is the main contributor.
		\item This is the generalization of \cite[Lemma 3.4]{gubinelliFBDSEApproachSine2026}.
	\end{enumerate}
\end{remark}

\section{Stochastic estimates}
\begin{lemma}\label{lem:entropy-UI}
	Let \(u\) be a progressively measurable process with its associated
	stochastic exponential \(\mathcal{E}(u)\) as defined in \eqref{eq:def-Ecal-u}.
	It holds that
	\begin{align}
		\frac{1}{2}\,\mathbb{E}^{\mathbb{P}^u}
		\int_0^T \|Q_su_s\|_{L^2}^2 \, \mathd s
		= \mathbb{H}(\mathbb{P}^u \vert \mathbb{P}),
	\end{align}
	where one side is finite if and only if the other is.
	In this case, \(\mathbb{P}^u \ll \mathbb{P}\)
	and \(\mathcal{E}(u)\) is a uniformly integrable martingale.
\end{lemma}
\begin{proof}
	See, e.g., \cite[Theorem 2]{lehecRepresentationFormulaEntropy2013}.
\end{proof}

\subsection{Estimates on the Gaussian free field}\label{app:GFF}
\begin{lemma}\label{lem:W-regularity}
	For every \(\varepsilon>0\), it holds that \(W_{t}\in C^\infty(\mathbb{R}^{2})\) almost surely whenever \(t<\infty\).
	Moreover, there exists \(c_\varepsilon>0\) such that
	\begin{align}\label{eq:W-regularity}
		\mathbb{E} \exp\left[c_\varepsilon\left(\sup_{T\geqslant0}\norm{W_T}^{2}_{H^{-\varepsilon}(\chi) }\right)\right] < \infty.
	\end{align}
	Furthermore, \(W_t \to W_\infty \in H^{-\varepsilon}(\chi)\) almost surely and in every \(L^p(\mathd \mathbb{P})\).
\end{lemma}

\begin{proof}
	By definition of the weighted Sobolev norm,
	\begin{align}
		\mathbb E\|W_T\|_{H^{-\varepsilon}(\chi)}^2
		=
		\int_{\mathbb R^2}\chi^{2}(x)\,
		\mathbb E\bigl|(1-\Delta)^{-\varepsilon/2}W_T(x)\bigr|^2\,\mathd x.
	\end{align}
	Since \(\tmop{Law}(W_T)\) is translation invariant, the integrand is independent of \(x\), and
	\begin{align}
		\mathbb E\bigl|(1-\Delta)^{-\varepsilon/2}W_T(0)\bigr|^2
		=         & \frac{1}{(2\pi)^2}\int_{\mathbb R^2}(1+\abs{\xi}^2)^{-\varepsilon}
		\int_0^T s^{-2}\mathe^{-\frac{m^2+\abs{\xi}^2}{s}}\,\mathd s\,\mathd\xi        \\
		\leqslant & \frac{1}{(2\pi)^2}\int_{\mathbb R^2}
		(1+\abs{\xi}^2)^{-\varepsilon}(m^2+\abs{\xi}^2)^{-1}\,\mathd\xi.
	\end{align}
	Hence
	\begin{align}
		\sup_{T\geqslant0}\, \mathbb E\|W_T\|_{H^{-\varepsilon}(\chi)}^2
		\leqslant
		\frac{\|\chi\|_{L^2}^{2}}{(2\pi)^2}
		\int_{\mathbb R^2}
		(1+\abs{\xi}^2)^{-\varepsilon}(m^2+\abs{\xi}^2)^{-1}\,\mathd\xi
		<\infty.
	\end{align}
	The estimate with the supremum inside the expectation now follows in the standard way from the Burkholder--Davis--Gundy inequalities.
	The exponential integrability follows from Fernique's theorem.
	For the regularity, recall that \(G_t=\int_0^t Q_s^2\,\mathd s\) with \(Q_s\) smooth and rapidly
	decaying in Fourier space.
	Realize \(W_t=\int_0^t Q_s\,\mathd B_s\) using a cylindrical Brownian motion \(B\).
	For every multiindex \(\alpha\), \(\nabla^\alpha Q_s\in L^2(\mathbb R^2)\), and
	\begin{align}
		\nabla^\alpha W_t = \int_0^t \nabla^\alpha Q_s\,\mathd B_s.
	\end{align}
	The estimates on the heat kernel (see Lemma \ref{lem:hk-scaling} and Lemma \ref{lem:hk-Lp-bounds}) yield a version of \(W_t\) with continuous derivatives of
	all orders, hence \(W_t\in C^\infty(\mathbb R^2)\) for every \(t<\infty\).
\end{proof}

\begin{lemma}\label{lem:W-gamma-nrm-finite}
	Let \(\gamma>0\) and let \(E_\gamma\) be the completion under \(\norm{\cdot}_\gamma\) of the space
	\begin{align}
		\left\{h\in C^\infty([0,\infty)\times\mathbb R^2):\norm{h}_\gamma<\infty\right\}.
	\end{align}
	The Gaussian process \((W_t)_{t\geqslant0}\) has a version in \((E_\gamma, \norm{\cdot}_{\gamma})\).
	In particular, there is a \(c>0\) such that
	\begin{align}\label{eq:W-gamma-nrm-finite}
		\mathbb{E} \left[
			\mathe^{c \norm{W}_{\gamma}^{2}}
			\right] < \infty,
		\qquad
		\mathbb{E} \left[
			\mathe^{-c \norm{W}_{\gamma}^{2}}
			\right] > 0.
	\end{align}
\end{lemma}

\begin{proof}
	By Lemma \ref{lem:W-regularity}, \(W_{t}\in C^\infty(\mathbb{R}^{2})\) almost surely for every fixed \(t<\infty\).
	Thus the spatial derivatives entering \(\norm{W}_\gamma\) are defined at every finite time.
	Standard Gaussian maximal estimates e.g. on dyadic space--time blocks, combined with the heat-kernel bounds in Lemmas \ref{lem:hk-scaling} and \ref{lem:hk-Lp-bounds}, imply
	\begin{align}
		\sup_{t\geqslant0}\langle t\rangle^{-\frac{\abs{\alpha}}{2}-\gamma}
		\norm{\nabla^\alpha W_t}_{L^\infty(\chi)}
		<\infty,
		\qquad\text{a.s. for }\abs{\alpha}\in\{0,1\},
	\end{align}
	so that \(\|W\|_\gamma<\infty\) almost surely.
	By approximating the supremum over \(\mathbb{R}^{2}\) and \([0,\infty)\) using the dense subsets \(\mathbb{Q}^{2}\) and \(\mathbb{Q}\),
	the seminorms \([h]_{\gamma, \alpha}\) are measurable as the countable supremum of measurable maps.
	Combined, this proves that \(W\) is a centered Gaussian random element in the measurable seminormed space
	\((E_\gamma,\|\cdot\|_\gamma)\). Fernique's theorem therefore gives some \(c>0\) such
	that \eqref{eq:W-gamma-nrm-finite} holds.
\end{proof}
\bibliography{local.bib}
\end{document}